\documentclass[a4paper, 12pt]{amsart}
\usepackage[foot]{amsaddr}

\usepackage[english]{babel}

\usepackage[left=2.5cm, right=2.5cm, top=3.5cm, bottom=3cm]{geometry}
\usepackage[onehalfspacing]{setspace}

\usepackage{bbold}	
\usepackage{amssymb}

\usepackage{graphicx}  

\usepackage{algpseudocode}
\usepackage{algorithm}

\renewcommand{\Return}{\State \textbf{return} }
\algrenewcomment[1]{\(\hfill\backslash\backslash\) #1}
\usepackage{caption}[2008/08/24]
\usepackage{subcaption}
\usepackage{hyperref}

\usepackage{color}

\usepackage[authoryear,round]{natbib}

\numberwithin{equation}{section}
\newcounter{satze}
\numberwithin{satze}{section}
\theoremstyle{plain} 

\newtheorem{Lemma}[satze]{Lemma}

\newtheorem{Theorem}[satze]{Theorem}

\theoremstyle{definition}

\theoremstyle{remark} 
\newtheorem{remark}[satze]{Remark}

\newcommand{\R}{\mathbb{R}}

\newcommand{\N}{\mathbb{N}}

\newcommand{\E}{\mathbb{E}}       

\newcommand{\EINS}{\mathbb{1}}           

\newcommand{\OO}{\mathcal{O}} 
\newcommand{\oo}{\hbox{o}} 

\newcommand{\Pj}{\mathbb{P}}             

\newcommand\argmin{\operatornamewithlimits{argmin}}

\newcommand{\valmu}{m}
\newcommand{\valsd}{s}
\newcommand{\sn}{{d_n}}
\newcommand{\snn}{{d_n}}

\newcommand{\HSMUCE}{\operatorname{H-SMUCE}}
\newcommand{\SMUCE}{\operatorname{SMUCE}}

\newcommand{\PELT}{\operatorname{PELT}}
\newcommand{\BS}{\operatorname{BS}}
\newcommand{\CBS}{\operatorname{CBS}}
\newcommand{\LOOVF}{\operatorname{LOOVF}}
\newcommand{\cumSeg}{\operatorname{cumSeg}}

\begin{document}

\title{Heterogeneous Change Point Inference}

\author{Florian Pein$^1$, Hannes Sieling$^1$ and Axel Munk$^1,^2$}
\address{$^1$Institute for Mathematical Stochastics, Georg-August-University of G{\"o}ttingen, Goldschmidtstra{\ss}e 7, 37077 G{\"o}ttingen}
\address{$^2$Max Planck Institute for Biophysical Chemistry, Am Fa{\ss}berg 11, 37077 G{\"o}ttingen}
\email{\{fpein, hsieling, munk\}@math.uni-goettingen.de}

\date{\today}
\subjclass[2010]{62G08,62G15,90C39}
\keywords{change-point regression, deviation bounds, dynamic programming, heterogeneous noise, honest confidence sets, ion channel recordings, multiscale methods, robustness, scale dependent critical values}

\begin{abstract}
We propose $\HSMUCE$ (\textbf{h}eterogeneous \textbf{s}imultaneous \textbf{mu}l\-ti\-scale \textbf{c}hange-point \textbf{e}stimator) for the detection of multiple change-points of the signal in a heterogeneous gaussian regression model. A piecewise constant function is estimated by minimizing the number of change-points over the acceptance region of a multiscale test which locally adapts to changes in the variance. The multiscale test is a combination of local likelihood ratio tests which are properly calibrated by scale dependent critical values in order to keep a global nominal level $\alpha$, even for finite samples.\\
We show that $\HSMUCE$ controls the error of over- and underestimation of the number of change-points. To this end, new deviation bounds for \textit{F}-type statistics are derived. Moreover, we obtain confidence sets for the whole signal. All results are non-asymptotic and uniform over a large class of heterogeneous change-point models. $\HSMUCE$ is fast to compute, achieves the optimal detection rate and estimates the number of change-points at almost optimal accuracy for vanishing signals, while still being robust.\\
We compare $\HSMUCE$ with several state of the art methods in simulations and analyse current recordings of a transmembrane protein in the bacterial outer membrane with pronounced heterogeneity for its states. An R-package is available online.
\end{abstract}

\maketitle

\section{Introduction}
\subsection{Change-point regression} Multiple change-point detection is a long standing task in statistical research and related areas. One of the most fundamental models in this context is (homogeneous) gaussian change-point regression
\begin{equation}\label{eq:modelsmuce}
 Y_i = \mu(i/n) + \sigma\epsilon_i,\ i=1,\ldots,n.
\end{equation}
Here, $Y=(Y_1,\ldots,Y_n)$ denotes the observations, $\mu$ is an unknown piecewise constant mean function, $\sigma^2$ a constant (homogeneous) variance and $\epsilon_1,\ldots,\epsilon_n$ are independent standard gaussian distributed errors. For simplicity, we restrict ourself in this paper to an equidistant sampling scheme $x_{i, n}=i/n$, but extensions to other designs are straightforward. Methods for estimating the change-points in \eqref{eq:modelsmuce} and in related models are vast, see for instance \citep{{Yao88},{DonohoJohnstone94},{CsoergoHorvath97},{BaiPerron98},{Braunetal00},{BirgeMassart01},{KolaczykNowak05},{Boysenetal09},{HarchaouiLevy10},{Jengetal10},{PELT},{RigolletTsybakov12},{ZhangSiegmund12},{WBS}} and the references in these papers.\\
A crucial condition in most of the afore-mentioned papers is the assumption of homogeneous noise, i.e. a constant variance $\sigma^2$ in \eqref{eq:modelsmuce}. In many applications, however, this assumption is violated and the variance $\sigma^2$ varies over time, $\sigma^2(i/n)$, say. This problem arises for instance in the analysis of array CGH data, see \citep{{MuggeoAdelfio11},{ARLOT}}. Further examples include economic applications, e.g. the real interest rate is modelled by \citet{BaiPerron03} as piecewise linear regression with covariates and heterogeneous noise. In this paper we will discuss an example from membrane biophysics, the recordings of ion channels, see Section \ref{sec:ion}. It is well known that the noise of the open state can be much larger than the background noise, see \citep[Section 3.4.4]{SakmannNeher95} and the references therein, rendering the different states as a potential source for variance heterogeneity.\\
To illustrate the effects of missing heterogeneity we show in Figure \ref{fig:example} a reconstruction by $\SMUCE$\footnote{\url{http://cran.r-project.org/web/packages/stepR}, v. 1.0-3, 2015-06-18} \citep{SMUCE}, a method that has been designed for homogeneous noise.
\begin{figure}[ht]
\centering
\begin{subfigure}{\textwidth}
\includegraphics[width=0.98\textwidth]{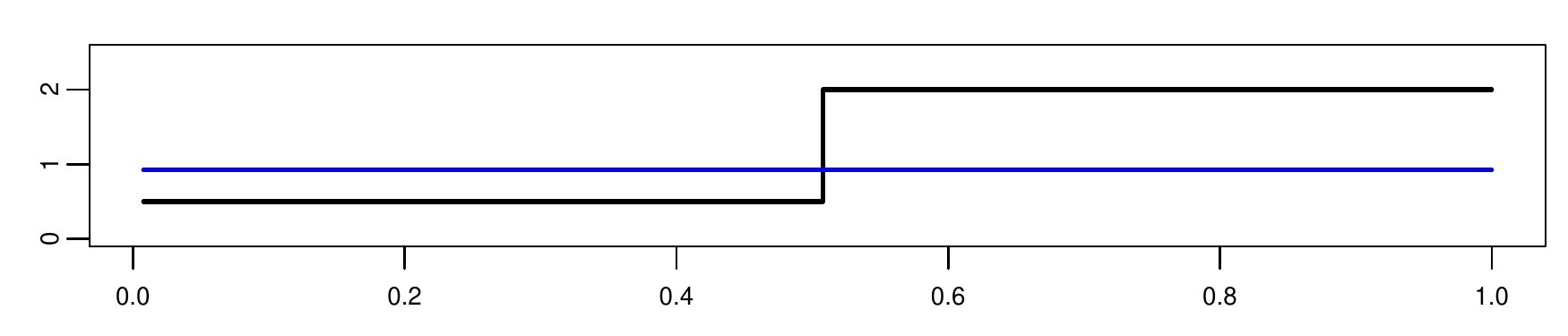}
\subcaption{True (black) and estimated (blue) standard deviation.}
\end{subfigure}
\begin{subfigure}{\textwidth}
\includegraphics[width=0.98\textwidth]{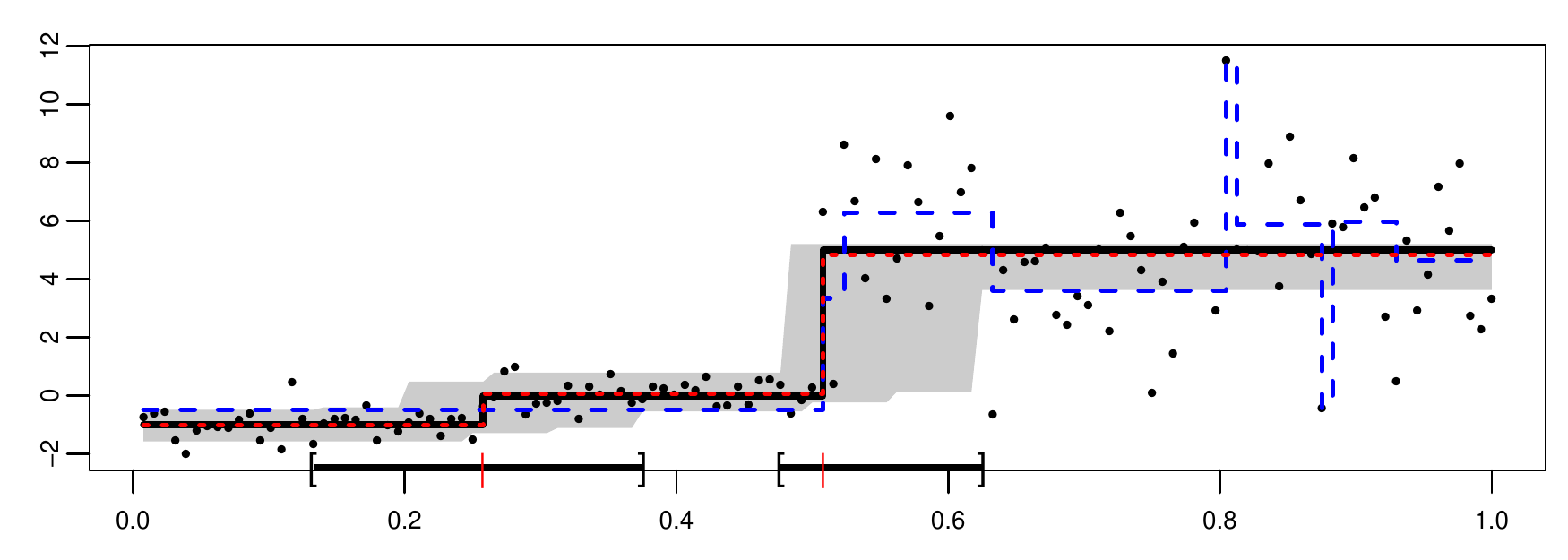}
\subcaption{Simulated observations (black dots) together with the true signal (black line), the confidence band (grey), the confidence intervals for the change-point locations (brackets and thick lines), estimated change-points locations (red dashes) as well as the estimates by $\operatorname{H-SMUCE}$ (red dotted line) and by $\operatorname{SMUCE}$ (blue dashed line), both with $\alpha=0.1$.}
\end{subfigure}
\caption{Illustration of missing heterogeneity.}
\label{fig:example}
\end{figure}
The constant variance assumption of $\SMUCE$ leads to an overestimation of the standard deviation (which is pre-estimated by a global $\operatorname{IQR}$ type estimator) in the first half and an underestimation in the second half. Therefore, in Figure \ref{fig:example} $\SMUCE$ misses the first change-point and includes artificial change-points in the second half to compensate for the too small variance it is forced to use, see also \citep{DiscussionZhou}. Note, that this flaw is not a particular feature of $\SMUCE$, it will occur for any sensible segmentation method which relies on a constant variance assumption. Hence, from Figure \ref{fig:example} the fundamental difficulty of the heterogeneous (multiscale) change-point regression problem becomes apparent: How to decide whether a change of fluctuations of the data result from high frequent changes in the mean $\mu$ or merely from an increase of the noise level? Apparently, if changes can occur on any scale (i.e. the length of an interval of neighbouring observations) this is a notoriously difficult issue and proper separation of signal and noise cannot be performed without extra information.\\
Indeed, the basis of the presented theory is that often a reasonable assumption is to exclude changes of the variance in constant segments of $\mu$ (see Section \ref{sec:model}). Under this relatively weak assumption, we show in this paper that estimation of $\mu$ for heterogeneous data in a multiscale fashion becomes indeed feasible. In addition, we also aim for a method which is robust when changes in the variance occur at locations where the signal is constant, as we believe that this cannot be excluded in many practical cases. To this end, we introduce a new estimator $\HSMUCE$ (\textbf{h}eterogeneous \textbf{s}imultaneous \textbf{mu}l\-ti\-scale \textbf{c}hange-point \textbf{e}stimator) which recovers the signal under heterogeneous noise over a broad range of scales, controls the family-wise error rate to overestimate the number of change-points, allows for confidence statements for all unknown quantities, obeys certain statistical optimality properties, and can be efficiently computed. At the same hand it is robust against heterogeneous noise on constant signal segments, which as a byproduct reveals it also as robust against more heavily tailed errors.

\subsection{Heterogeneous change-point model}\label{sec:model}
To be more specific, from now on we consider the \textit{heterogeneous} gaussian change-point model
\begin{equation}\label{eq:model}
 Y_i = \mu(i/n) + \sigma(i/n) \epsilon_i,\ i=1,\ldots,n,
\end{equation}
where now the variance $\sigma^2$ is also given by an unknown piecewise constant function. For the following theoretical results we assume that it only can have possible change-points at the same locations as the mean function $\mu$. In other words, $(\mu, \sigma^2)$ is a pair of unknown piecewise constant functions in
\begin{equation}\label{eq:S}
 \mathcal{S} := \left\{ (\mu,\sigma^2): [0,1]\mapsto \R^2,\ \mu=\sum_{k=0}^K \valmu_k \EINS_{[\tau_k,\tau_{k+1})},\ \sigma^2=\sum_{k=0}^K \valsd_k^2 \EINS_{[\tau_k,\tau_{k+1})},\ K\in\N\right\},
\end{equation}
with unknown change-point locations $\tau_0=0<\tau_1<\cdots<\tau_K<1=\tau_{K+1}$ for some unknown number of change-points $K\in\N$ and also unknown function values $\valmu_k\in\R$ and $\valsd_k^2\in \R_{+}$ of $\mu$ and $\sigma^2$. By technical reasons, we define $\mu(1)$ and $\sigma^2(1)$ by continuous extension of $\mu$ and $\sigma^2$, respectively. For identifiability of $\mu$ we assume $\valmu_k\neq \valmu_{k+1}\ \forall\ k=0,\ldots,K$ and exclude isolated changes in the signal by assuming that $\mu:[0,1]\to\R$ is a right continuous function. It is important to stress that in \eqref{eq:S} we allow the variance to \textit{potentially} have changes at the locations of the changes of the signal, but the variance $\sigma^2$ need not necessarily change when $\mu$ changes, as we do not assume $s_k^2\neq s_{k+1}^2$. In particular, homogeneous observations are still part of the model. The other way around, we assume that within a constant segment of $\mu$ it may \textit{not} happen that the variance changes, i.e. the local signal to noise ratio is assumed to be constant on $[\tau_k,\tau_{k+1})$ for all $k=0,\ldots,K$. We argue that this is a reasonable assumption in many applications (recall the examples given above and see our data example in Section \ref{sec:ion}), since a change-point represents typically a change of the condition of the underlying state. Moreover, for example, in many engineering applications locally a constant signal to noise ratio is assumed \citep{Guillaumeetal90}, which motivates our modelling as well. However, we stress that the restriction to model \eqref{eq:S} is only required for our theory. For the practical application we will show in simulations in Section \ref{sec:simrobust} that $\HSMUCE$ is in addition robust against a violation of this assumption (i.e. when a variance change may occur without a signal change) and hence works still well in the general heterogeneous change-point model \eqref{eq:model}ſ with arbitrary variance changes.

\subsection{Heterogeneous change-point regression}\label{sec:heterogeneitydiscussion}
Up to our best knowledge there are only few methods which explicitly take into account the heterogeneity of the noise in change-point regression, either in the model considered here or in related models. Thereby, we have to distinguish two settings.\\
First, that also changes in the variance are considered as relevant structural changes of the underlying data (even when the mean does not change) and to seek for changes in the mean and in the variance, respectively. In this spirit are local search methods, such as binary segmentation ($\BS$) \citep{{ScottKnott74},{Vostrikova81}} (if the corresponding single change-point detection method takes the heterogeneous variance into account), but also global methods can achieve this goal, e.g. $\PELT$ \citep{PELT}. For a Bayesian approach in this context see \citep{Duetal15} and the references therein. In addition, methods which search for more general structural changes in the distribution potentially apply to this setup as well, see e.g. \citep{{CsoergoHorvath97},{Arlotetal12},{MattesonJames13}}.\\
This is in contrast to the setting we address in this paper: The variance is considered as a nuisance parameter and we primarily seek for changes in the signal $\mu$. Hence, we aim for statistically efficient estimation of the mean function, but still being robust against heterogeneous noise. Obviously, this cannot be achieved by methods addressing the first setting. Although of great practical relevance, this situation has only rarely been considered and in particular no theory exists, to our knowledge. The cross-validation method $\LOOVF$ \citep{ARLOT} and $\cumSeg$ \citep{MuggeoAdelfio11} have been designed specifically to be robust against heterogeneous noise. Moreover, also circular binary segmentation ($\CBS$), see \citep{CBS07}, applies to this.\\
For a better understanding of the problem considered here it is illustrative to distinguish our setting further, namely from the case when it is known \textit{before hand} that changes in the variance will necessarily occur with changes in the signal. This will potentially increase the detection power as under this assumption variance changes can be used for finding signal changes, as well. The information gain due to the variance changes for this case has been recently quantified by \citet{Enikeevaetal15} in terms of the minimax detection boundary for single vanishing signal bumps of size $\delta_n\searrow 0$. More precisely, if the base line variance is $\sigma_0^2$ and the variance at the bump is $\sigma_0^2+\sigma_n^2$ then the constant in the minimax detection boundary is $b=\sqrt{2}\sigma_0\sqrt{2/(2+c^2)}$ for $c=\sigma_0^{-1}\lim_{n\to\infty}{\sigma_n/\delta_n}$, see \citep[Theorems 3.1-3.3]{Enikeevaetal15}. For the particular case of homogeneous variance, i.e. $\sigma_n^2=0$, we obtain $b=\sqrt{2}\sigma_0$ and the factor $\sqrt{2/(2+c^2)}=1$ becomes maximal, see also \citep{DuembgenWalther08, SMUCE}. This reflects that no additional information on the location of a change can be gained from the variance in the homogeneous case. Comparing this to the inhomogeneous case we see that when the variance change is known to be large enough, i.e. $\sigma_0^ {-1}\lim_{n\to\infty}{\sigma_n/\delta_n}>0$, additional information for the signal change can be gained from the variance change, as then $b<\sqrt{2}\sigma_0$, provided it is known that signal and variance change \textit{simultaneously}.\\
In contrast, in the present setting the variance need not necessarily change when the signal changes, hence the "worst case" of no variance change from above is contained in our model, which lower bounds the detection boundary. The situation is further complicated due to the fact that missing knowledge of a variance change can potentially even have an adverse effect because in model \eqref{eq:S} detection power will be potentially decreased further as the nuisance parameter $\sigma^2(\cdot)$ hinders estimation of change-points of $\mu$. For this situation the optimal minimax \textit{constants} are unknown to us, but from the fact that the model with a constant variance is a submodel of our model \eqref{eq:S} it immediately follows that the minimax constant for a single bump has to be at least $\sqrt{2}\sigma_0$. This will allow us to show that $\HSMUCE$ attains the same optimal minimax detection \textit{rate} as for the homogeneous case and $4\sigma_0$ instead of $\sqrt{2}\sigma_0$ as the constant appearing in the minimax detection boundary. Remarkably, the only extra assumption we have to suppose is that signal and variance have to be constant on segments at least of order $\log(n)/n$, see Theorem \ref{theorem:multiscalesignalsrate}. This reflects the additional difficulty to separate "locally" signal and noise levels in a multiscale fashion. In other words, when we assume that the number of i.i.d. neighbouring observations (no change in signal and variance) in each segment is at least of order $\log(n)$, separation of signal and noise will be done by $\HSMUCE$ in an optimal way (possibly up to a constant).

\subsection{Heterogeneous change-point inference}
We define $\HSMUCE$ as the multiscale constrained maximum likelihood estimator restricted to all solutions of the following optimisation problem
\begin{equation}\label{eq:optproblemhsmuce}
\argmin_{\mu\in\mathcal{M}}{\left\vert \mathcal{I}(\mu)\right\vert}\quad \text{s.t. }\max_{[\frac{i}{n},\frac{j}{n}]\in \mathcal{D}(\mu)}\left[T_i^j(Y,\mu([i/n, j/n]))-q_{ij}\right]\leq 0,
\end{equation}
see also \citep{Boysenetal09, Daviesetal12, SMUCE} for related approaches. Here, $\mathcal{M}$ (as a subset of $\mathcal{S}$) is the set of all piecewise constant mean functions, $\vert\mathcal{I}(\mu)\vert$ the cardinality of the set of change-points of $\mu$ and the right hand side of \eqref{eq:optproblemhsmuce} a multiscale constraint to be explained now. Given a candidate function $\mu$ this tests \textit{simultaneously} over the system of all intervals $\mathcal{D}(\mu)$ on which $\mu$ is constant, whether its function value $\mu([i/n, j/n])$ is the mean value of the observations on the respective interval $[i/n, j/n]$. In order to perform each test, i.e. to decide whether the observations $Y_i,\ldots,Y_j$ have constant mean $\mu([i/n, j/n])$, the local log-likelihood-ratio statistic
\begin{equation}\label{eq:localtest}
T_i^j(Y,\mu([i/n, j/n])):=\left(j-i+1\right)\frac{\left(\overline{Y}_{ij}-\mu([i/n, j/n])\right)^2}{\hat{\valsd}^2_{ij}},
\end{equation}
with $\overline{Y}_{ij}:=(j-i+1)^{-1}\sum_{l=i}^j{Y_l}$ and local variance estimate
$\hat{\valsd}^2_{ij}:=(j-i)^{-1}\sum_{l=i}^j{(Y_l-\overline{Y}_{ij})^2}$, is compared with a local threshold $q_{ij}$ in a multiscale fashion, to be discussed now.\\
In what follows, we restrict the multiscale test to intervals in the \textit{dyadic partition}
\begin{equation}\label{eq:dyapar}
\mathcal{D}:=\bigcup_{k=1}^{\sn}\mathcal{D}_k,
\end{equation}
where $\sn:=\lfloor\log_2(n)\rfloor$ is the number of different scales and
\begin{equation}\label{eq:dyapar_k}
\mathcal{D}_k:=\bigcup_{l=1}^{\left\lfloor \frac{n}{2^k}\right\rfloor}{\left[\frac{1+(l-1)2^k}{n},\frac{l 2^k}{n}\right]}
\end{equation}
the set of intervals from the dyadic partition with length $n^{-1}2^k$. This allows fast computation and simplifies the asymptotic analysis. Nevertheless, our methodology can be adapted to other intervals systems, see Remark \ref{remark:otherintervalsets}.\\
It remains to determine thresholds $q_{ij}$ for $\mathcal{D}$ in \eqref{eq:dyapar} that combine the local tests appropriately. To this end, note that logarithmic (or related) scale penalisation as in the homogeneous case \citep{DuembgenSpokoiny01, DuembgenWalther08, SMUCE} does not balance scales anymore appropriately in the heterogeneous case. In particular, this will give a multiscale statistic which diverges, since due to the local variance estimation the test statistic fails to have subgaussian (but still has subexponential) tails. To overcome this burden we introduce in Section \ref{sec:criticalvalues} scale dependent critical values such that the multiscale test has global significance level $\alpha$, see \eqref{eq:significancelevel}. To this end, the different scales are balanced appropriately by weights $\beta_1,\ldots,\beta_\sn$, with $\sn:=\lfloor\log_2(n)\rfloor$, see \eqref{eq:weights} and \eqref{eq:balancing}. More precisely, these weights determine the ratios between the rejection probabilities of the multiscale test on a corresponding scale. Existence and uniqueness of the so defined scale dependent critical values is shown in Lemma \ref{lemma:uniqunesscritval} and explicit bounds are given in Lemma \ref{lemma:boundcritval}. The weights also allow to incorporate prior scale information, see Section \ref{sec:tuningparameters}.\\
Using the so obtained thresholds $q_{ij}$ allows to obtain several confidence statements which are a main feature of $\HSMUCE$. First of all, we show in Section \ref{sec:theory} that the probability to overestimate the number of change-points is bounded by the significance level $\alpha$ uniformly over $\mathcal{S}$ in \eqref{eq:S}, $\Pj(\hat{K} > K)\leq \alpha$, see Theorem \ref{theorem:boundoverestimate}. More specifically, we show the \textit{overestimation} bound
\begin{equation}\label{eq:overestimationexponential}
\sup_{(\mu,\sigma^2)\in\mathcal{S}}\Pj_{(\mu, \sigma^2)}\left(\hat{K}>K+2k\right)\leq \alpha^{k+1},\ \forall\ k\in\N_{0},
\end{equation}
see Theorem \ref{theorem:overestimationk}. In Theorem \ref{theorem:boundunderestimate} we provide an exponential bound for the underestimation of the number of change-points by $\HSMUCE$, $\Pj(\hat{K}<K)$. To this end, we show new exponential deviation bounds for $F$-statistics (Section \ref{sec:prooflargedeviation}), which might be of interest by its own. Combining the over- and the underestimation bound provides upper bounds for the errors $\Pj(\hat{K}\neq K)$ and $\E[\vert \hat{K}-K\vert]$. For a fixed signal both bounds vanish super polynomially in $n$ if $\alpha=\alpha_n\searrow 0$ when the weights are chosen appropriately, see Remark \ref{remark:errorbounds}. Consequently, the estimated number of change-points converges almost surely to the true number, see Theorem \ref{theorem:consistency}. Further, these exponential bounds enable us to obtain a confidence band for the signal $\mu$ as well as confidence intervals for the locations of the change-points, for an illustration see Figures \ref{fig:example} and \ref{fig:example1000}. 
\begin{figure}[!htp]
\centering
\begin{subfigure}{\textwidth}
\includegraphics[width=0.98\textwidth]{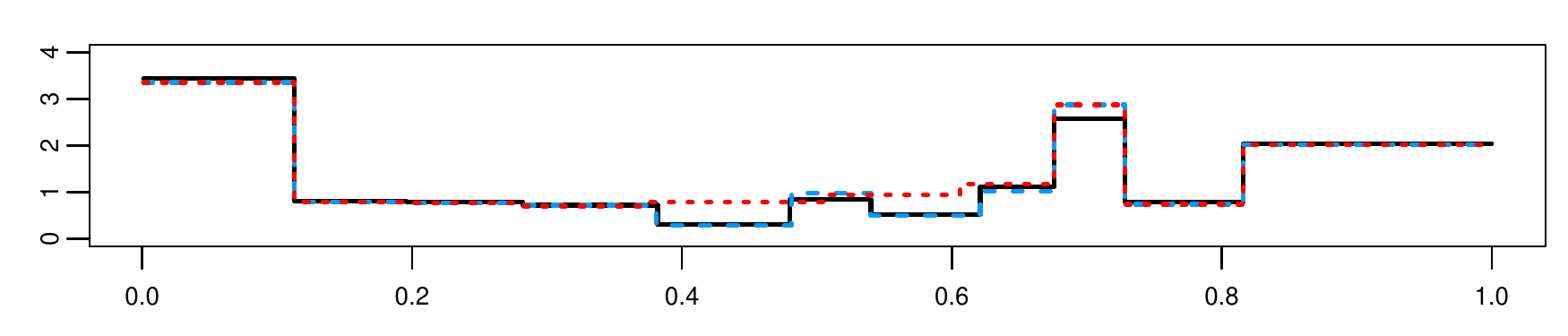}
\subcaption{True standard deviation (black) and estimates resulting from $\HSMUCE$ at $\alpha = 0.1$ (red line) and $\alpha = 0.3, 0.5$ (blue line).}
\end{subfigure}
\begin{subfigure}{\textwidth}
\includegraphics[width=0.98\textwidth]{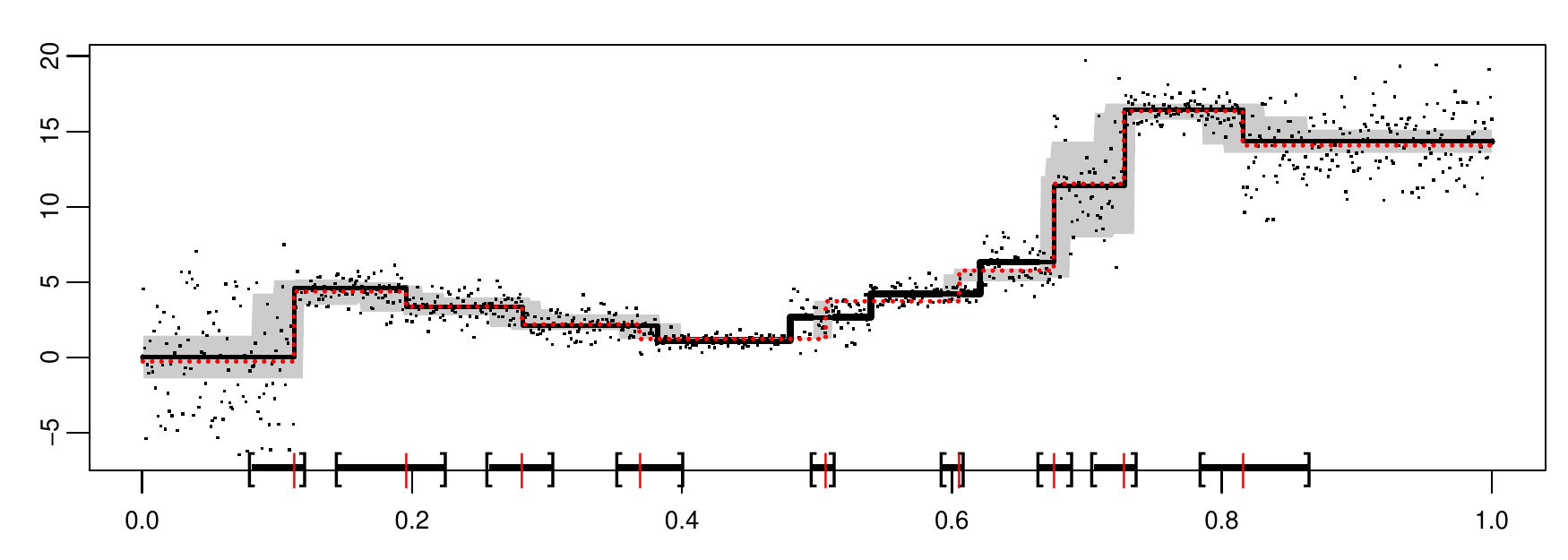}
\subcaption{$\alpha = 0.1$}
\label{subfig:B:example1000}
\end{subfigure}
\begin{subfigure}{\textwidth}
\includegraphics[width=0.98\textwidth]{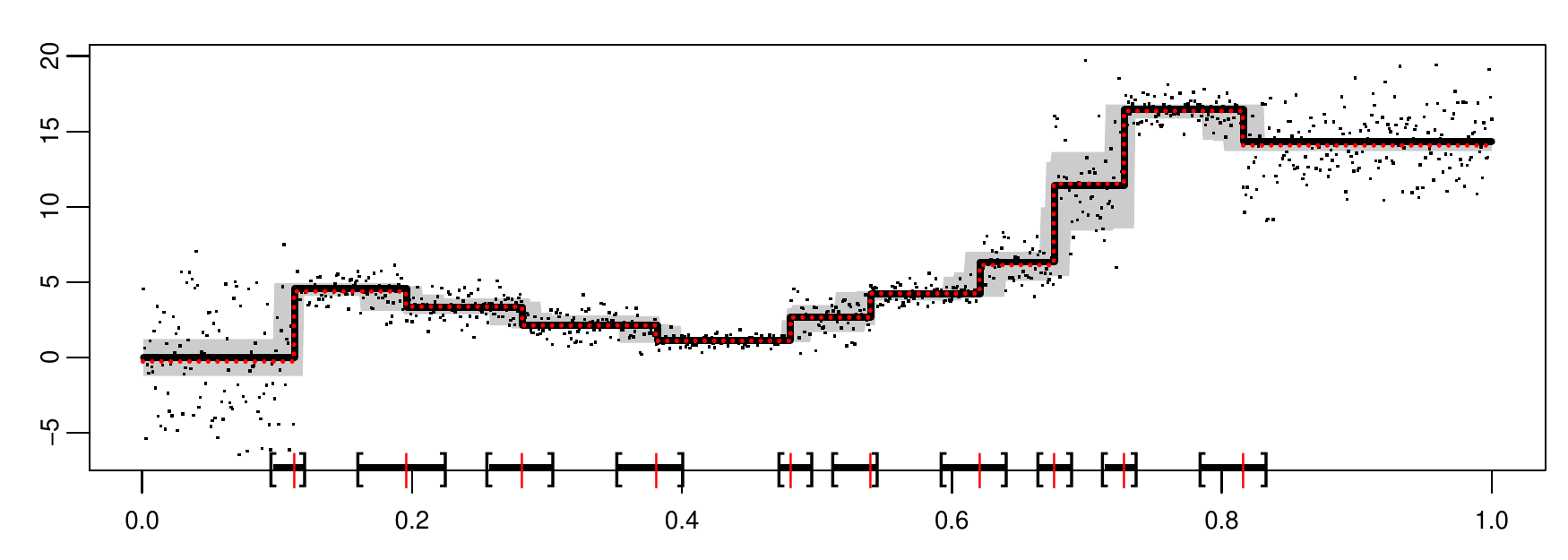}
\subcaption{$\alpha = 0.3$}
\end{subfigure}
\begin{subfigure}{\textwidth}
\includegraphics[width=0.98\textwidth]{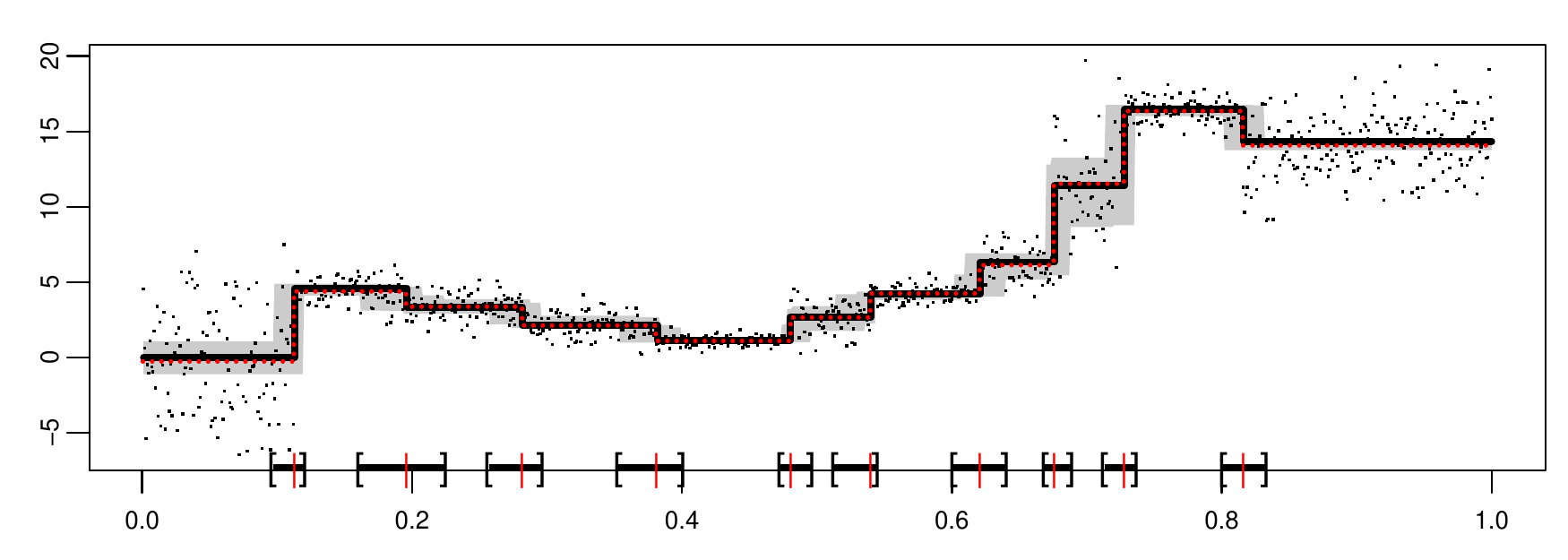}
\subcaption{$\alpha = 0.5$}
\label{subfig:D:example1000}
\end{subfigure}
\caption{\subref{subfig:B:example1000}-\subref{subfig:D:example1000}: Observations (black dots), true signal (black line), confidence band (grey), confidence intervals for the change-point locations (brackets and thick lines), estimated change-points locations (red dashes) and estimate (red line) $\HSMUCE$ at given $\alpha$ and with equal weights $\beta_1=\cdots=\beta_\sn=1/\sn$, see \eqref{eq:weights} and \eqref{eq:balancing}.}
\label{fig:example1000}
\end{figure}
We show that the diameters of these confidence intervals decrease asymptotically as fast as the (optimal) sampling rate up to a log factor. All confidence statements hold uniformly over $\mathcal{S}_{\Delta,\lambda}\subset\mathcal{S}$, all functions with minimal signal to noise ratio $\geq\Delta$ and minimal scale $\geq\lambda:=\min_{k=0,\ldots,K}{(\tau_{k+1}-\tau_k)}$, with $\Delta$ and $\lambda$ arbitrarily, but fixed, see Theorems \ref{theorem:confidencefunction} and \ref{theorem:confidencelocation}.

\subsection{H-SMUCE in action} Figure \ref{fig:example1000} illustrates the performance of $\HSMUCE$ in an example with $n=1\,000$ observations and $K=10$ change-points. We found that $\HSMUCE$ misses for $\alpha = 0.1$ one change-point (as the choice $\alpha=0.1$ tunes $\HSMUCE$ to provide the strong guarantee not to overestimate the number of change-points $K$ with probability $0.9$, see \eqref{eq:overestimationexponential}), whereas for $\alpha$ between $0.15$ and $0.99$ (only displayed for $\alpha = 0.3$ and $\alpha = 0.5$) the correct number of change-points is detected always (while providing a weaker guarantee for not overestimating $K$). In addition, for $\alpha$ between $0.15$ and $0.99$ each true change-point is covered by the associated confidence interval at level $1-\alpha$. This illustrates the influence of the significance level $\alpha$. Notably, we find that the reconstructions are remarkably stable in $\alpha$. In fact, combining Lemma \ref{lemma:boundcritval} and \eqref{eq:bounds} shows that the width of the confidence band is proportional to $\sqrt{\log(1/\alpha)}$ which decreases only logarithmically for increasing $\alpha$.\\
We compare the performance of $\HSMUCE$ with $\CBS$ \citep{CBS07}, $\cumSeg$ \citep{MuggeoAdelfio11} and $\LOOVF$ \citep{ARLOT} in several simulation studies in Section \ref{sec:simulations} (see also Figure \ref{fig:example1000other} in Supplement \ref{sec:add} for their performance on the data in Figure \ref{fig:example1000}), where we also examine robustness issues, see Section \ref{sec:simrobust}. In all of these simulations and in the subsequent application $\HSMUCE$ performs very robust and includes too many change-points only rarely in accordance with \eqref{eq:overestimationexponential}.\\
In Section \ref{sec:ion} we apply $\HSMUCE$ to current recordings of a transmembrane protein with pronounced heterogeneity for its states. In contrast to segmentation methods which rely on homogeneous noise, we found that $\HSMUCE$ provides a reasonable reconstruction, where all visible gating events are detected.\\
Finally, we stress that the confidence band and confidence intervals for the change-point locations provided by $\HSMUCE$ can be used to accompany any segmentation method to assess significance of its estimated change-points. This is illustrated in Section \ref{sec:ion} as well.\\
Computation of the estimator by a pruned dynamic program and of the critical values based on Monte-Carlo simulation is explained carefully in Supplement \ref{sec:computation}. There we also study the theoretical and empirical computation time of $\HSMUCE$. Due to the underlying dyadic partition the computation of $\HSMUCE$ is very fast, in some scenarios even linear in the number of observations. Additional simulations results are collected in Supplement \ref{sec:add} and all proofs are given together with some auxiliary statements in Supplement \ref{sec:proofs}. An R-package is available online\footnote{\url{http://www.stochastik.math.uni-goettingen.de/hsmuce}}.

\section{Scale dependent critical values}\label{sec:criticalvalues}
For the definition of $\HSMUCE$ it remains to determine the local thresholds $q_{ij}$ in \eqref{eq:optproblemhsmuce}. First of all, the multiscale test on the r.h.s. in \eqref{eq:optproblemhsmuce} should be a level $\alpha$ test, i.e.
\begin{equation}\label{eq:multiscaletestlevelalpha}
\sup_{(\mu,\sigma^2)\in\mathcal{S}}\Pj_{(\mu,\sigma^2)}\left(\max_{[\frac{i}{n},\frac{j}{n}]\in \mathcal{D}(\mu)} \left[T_i^j\big(Y,\mu([i/n,j/n])\big) - q_{ij}\right] > 0\right)\leq \alpha.
\end{equation}
Here, we use the same threshold for all intervals of the same length as no a-priori information on the change-point locations is assumed. More precisely, as we have restricted the multiscale test to the dyadic partition in \eqref{eq:dyapar} we aim to find a vector of critical values $\textbf{q}:=(q_1,\ldots,q_\sn)$, where now $q_{ij}:=q_k$ if and only if $j-i+1=2^k$. To this end, w.l.o.g., we may consider standard gaussian observations $Z_1,\ldots,Z_n$ instead of $Y_1,\ldots,Y_n$, since the supremum in \eqref{eq:multiscaletestlevelalpha} is attained at $\mu\equiv 0$ and $\sigma^2\equiv 1$, see the proof of Theorem \ref{theorem:boundoverestimate}. We then define the statistics $T_1,\ldots,T_\sn$ with $\mathcal{D}_k$ in \eqref{eq:dyapar_k} as
\begin{equation}\label{eq:Tk}
T_k:=\max_{[i/n,j/n]\in \mathcal{D}_k}{T_i^j(Z,0)} \quad \text{for} \quad k=1,\ldots,\sn.
\end{equation}
Then, the critical values $q_1,\ldots,q_\sn$ fulfil \eqref{eq:multiscaletestlevelalpha} if
\begin{equation}\label{eq:significancelevel}
\Pj\left(\max_{k=1,\ldots,\sn} \left[ T_k-q_{k} \right] > 0 \right) = 1-F\left(q_1,\ldots,q_\sn\right)=\alpha,
\end{equation}
with $F$ the cumulative distribution function of $(T_1,\ldots,T_\sn)$.\\
As the critical values $q_1,\ldots,q_\sn$ are not uniquely determined by \eqref{eq:significancelevel} they can be chosen to render the multiscale test particularly powerful for certain scales. To this end, we introduce weights
\begin{equation}\label{eq:weights}
\beta_1,\ldots,\beta_\sn \geq 0,\text{ with }\sum_{k=1}^{\sn}{\beta_k}=1,
\end{equation}
where $\beta_k=0$ means to omit the $k$-th scale, i.e. $q_k=\infty$. Finally, we define $q_1,\ldots,q_\sn$ implicitly through
\begin{equation}\label{eq:balancing}
\frac{1-F_1\left(q_1\right)}{\beta_1}=\cdots=\frac{1-F_\sn\left(q_{\sn}\right)}{\beta_\sn},
\end{equation}
with $F_k$ the cumulative distribution function of $T_k$. If $\beta_k=0$ this will not enter the systems of equations in \eqref{eq:balancing}. The weights determine the fractions between the probabilities that a test on a certain scale rejects, and hence regulate the allocation of the level $\alpha$ among the single scales. In summary, the choice of the local thresholds $q_{ij}$ boils down to choosing the significance level $\alpha$ and the weights $\beta_1,\ldots,\beta_\sn$, we discuss these choices in Section \ref{sec:tuningparameters} more carefully. If no prior information on scales is available a default option is always to set all weights equal, i.e. $\beta_1=\ldots=\beta_{\sn}=1/\sn$.\\
The next result shows that the vector of critical values satisfying \eqref{eq:significancelevel}-\eqref{eq:balancing} is always well-defined.
\begin{Lemma}[Existence and uniqueness]\label{lemma:uniqunesscritval}
For any $\alpha\in(0,1)$ and for any weights $\beta_1,\ldots,\beta_\sn$, s.t. \eqref{eq:weights} holds, there exits a unique vector of critical values $\textbf{q}=(q_1,\ldots,q_\sn)\in\R_+^\sn$ which fulfils the equations \eqref{eq:significancelevel} and \eqref{eq:balancing}.
\end{Lemma}
An explicit computation of the vector $\textbf{q}$ (or $F$) appears to be very hard, since the statistics $T_1,\ldots,T_\sn$ are dependent, although the dependence structure is explicitly known. Alternatively, it would be helpful to have an approximation for the distribution (and hence its quantiles) of the maximum in \eqref{eq:multiscaletestlevelalpha}, which, however, appears to be rather difficult, as well. For the case when $q_{ij}\equiv q$ (which does not apply to $\HSMUCE$), see \citet{Davies87, Davies02}. Therefore, we determine in Section \ref{sec:detcriticalvalues} the vector $\textbf{q}$ by Monte-Carlo simulations. Note that the distribution does not depend on the specific element $(\mu,\sigma^2)\in\mathcal{S}$ and hence the critical values can be computed in a universal manner. We stress that the determination of the scale dependent critical values is not restricted to our setting and can also be applied to multiscale testing in other contexts. Different to scale penalisation and like the block criterion in \citep{RufibachWalther10} no model dependent derivations are required and the critical values are adapted to the exact finite sample distribution of the local test statistics. However, our approach allows additionally a flexible scale calibration by the choice of the weights (see Section \ref{sec:tuningparameters}) and arbitrary interval sets can be used as the following remark points out. 

\begin{remark}[Other interval sets]\label{remark:otherintervalsets}
$\HSMUCE$ can be easily adjusted to other interval sets as follows. Let $\mathcal{I}$ be an arbitrary set of intervals. Then, we replace in the definition of $\HSMUCE$ in \eqref{eq:K} and \eqref{eq:defC} the set $\mathcal{D}$ by the set $\mathcal{I}$ and the vector $(T_1,\ldots,T_\sn)$ by the vector $(\widetilde{T}_2,\ldots,\widetilde{T}_{n})$ (empty scales should be omitted) in Section \ref{sec:criticalvalues}, with
\begin{equation}\label{eq:Tktilde}
\widetilde{T}_k:=\max_{\substack{[i/n,j/n]\in\mathcal{I},\\ j-i+1 = k}}{T_i^j(Z,0)}.
\end{equation}
Again it remains to choose the significance level $\alpha\in (0,1)$ and the weights $\beta_2,\ldots,\beta_n$ to determine the critical values required for $\HSMUCE$. Note, however, that the critical values and its bounds in Lemma \ref{lemma:boundcritval} and therefore the results in Section \ref{sec:theory} (besides of Theorems \ref{theorem:boundoverestimate} and \ref{theorem:overestimationk}) will depend on the specific system $\mathcal{I}$ and have to be computed for each $\mathcal{I}$ separately.\\
Employing a larger interval set than $\mathcal{D}$ may lead to a better detection power, but at the price of a larger computation time. Hence, in practice, a trade-off between computational and statistical efficiency may guide this choice as well. Our R-package includes beside of the dyadic partition also the system of all intervals (of order $\OO(n^2)$, statistically most efficient, but computationally expensive) and the system of all intervals of dyadic length ($\OO(n\log(n)$, intermediate efficiency and computational time). Interesting choices might be also approximating sets like $\mathcal{J}_{\operatorname{app}}$ introduced in \citep{{Walther10},{RiveraWalther13}} which are larger than the dyadic partition, but achieve the minimax boundary in the context of density estimation.
\end{remark}

\section{Theory}\label{sec:theory}
In this section we collect our theoretical results. We start with finite bounds for the critical values. These will allow to bound $\Pj(\hat{K} \neq K)$. With these bounds we obtain confidence statements for the signal $\mu$ and its main characteristics. Finally, we investigate asymptotic detection rates of $\HSMUCE$ for vanishing signals.

\subsection{Finite bounds for over- and underestimation}
In the following we require upper bounds for the critical values, since the definition of the critical values by the equations \eqref{eq:significancelevel}-\eqref{eq:balancing} is implicit.
\begin{Lemma}[Bound on critical values]\label{lemma:boundcritval}
Let $\textbf{q}=(q_1,\ldots,q_\sn)$ be the vector of critical values defined by \eqref{eq:significancelevel}-\eqref{eq:balancing}, then for every $k\in\{2,\ldots,\sn\}$ such that 
\begin{equation}\label{eq:conditionboundcritval}
2^{-k}\log\left(\frac{n}{2^k\alpha\beta_k}\right)\leq \frac{1}{2}
\end{equation}
we have
\begin{equation}\label{eq:boundcritval}
q_k \leq 8\log\left(\frac{n}{2^k\alpha\beta_k}\right).
\end{equation}
\end{Lemma}
\begin{remark}
The log term of the bound \eqref{eq:boundcritval} can be split into a scale dependent penalty term $\log(n2^{-k})$ which is of the same order as the penalties in the homogeneous case in \citep{{DuembgenSpokoiny01},{SMUCE}}, and into the term $\log((\alpha\beta_k)^{-1})$ which incorporates the significance level $\alpha$ and the weight $\beta_k$.
\end{remark}
The following theorem shows that the significance level $\alpha$ controls the probability to overestimate the number of change-points.
\begin{Theorem}[Overestimation control I]\label{theorem:boundoverestimate}
Assume the heterogeneous gaussian change-point model \eqref{eq:model}. Let $K:=\vert \mathcal{I}(\mu)\vert$ be the number of change-points of a signal $\mu\in\mathcal{M}$. Let further $\hat{K}$ be the estimated number of change-points by $\HSMUCE$, i.e.
\begin{equation}\label{eq:K}
\hat{K}:=\min\left\{\vert \mathcal{I}(\mu)\vert\ :\ \mu\in\mathcal{M} \text{ with } \max_{[\frac{i}{n},\frac{j}{n}]\in \mathcal{D}(\mu)} \left[T_i^j\big(Y,\mu([i/n,j/n])\big) - q_{ij}\right] \leq 0  \right\}.
\end{equation}
Then, for any vector of critical values $\textbf{q}$ with significance level $\alpha\in (0,1)$ and weights $\beta_1,\ldots,\beta_{\sn}$ in \eqref{eq:significancelevel}-\eqref{eq:balancing}, uniformly over $\mathcal{S}$ in \eqref{eq:S} it holds
\[\sup_{(\mu,\sigma^2)\in\mathcal{S}}\Pj_{(\mu, \sigma^2)}\left(\hat{K}>K\right)\leq \alpha.\]
\end{Theorem}
The theorem gives us a direct interpretation of the parameter $\alpha$ as the probability to overestimate the number of change-points. This even holds locally, i.e. on every union of adjoining segments of the estimator $\HSMUCE$ with probability $1-\alpha$ there are at least as many change-points as detected. Moreover, we strengthen the result by showing that the probability to estimate additional changes decays exponentially fast and hence the expected overestimation is small.
\begin{Theorem}[Overestimation control II]\label{theorem:overestimationk}
Under the assumptions of Theorem \ref{theorem:boundoverestimate}, we have
\begin{equation*}
\sup_{(\mu,\sigma^2)\in\mathcal{S}}\Pj_{(\mu, \sigma^2)}\left(\hat{K}>K+2k\right)\leq \alpha^{k+1},\ \forall\ k\in\N_{0}.
\end{equation*}
Moreover,
\begin{equation*}
\sup_{(\mu,\sigma^2)\in\mathcal{S}}\E_{(\mu, \sigma^2)}\left[(\hat{K}-K)_+\right]\leq \frac{2\alpha}{1-\alpha}.
\end{equation*}
\end{Theorem}
To control the probability $\Pj(\hat{K}\neq K)$ we need additionally an upper bound for the probability to \textit{underestimate} $K$. Unlike to the \textit{overestimation} bounds in the Theorems \ref{theorem:boundoverestimate} and \ref{theorem:overestimationk} the probability to underestimate cannot be bounded uniformly over $\mathcal{S}$, since size and scale of changes could be arbitrarily small. This is made more precise in Theorem \ref{theorem:multiscalesignalsrate} which gives the detection boundary in terms of the  smallest (standardized) jump size $\Delta$ and the smallest scale $\lambda$. The next theorem provides an exponential bound uniformly over the subset 
\begin{equation}\label{eq:subsetSdeltalambda}
\mathcal{S}_{\Delta,\lambda}:=\left\{(\mu,\sigma^2)\in\mathcal{S}\ :\ \Delta\leq\inf_{1\leq k\leq K}{\frac{\vert\mu_k-\mu_{k-1}\vert}{\max\left(\valsd_{k-1},\valsd_{k}\right)}} \text{ and } \lambda\leq\inf_{0\leq k\leq K}{(\tau_{k+1}-\tau_{k})}\right\},
\end{equation}
with $\Delta,\lambda >  0$ arbitrary, but fixed.
\begin{Theorem}[Underestimation control]\label{theorem:boundunderestimate}
Let $\mathcal{S}_{\Delta,\lambda}$ be as in \eqref{eq:subsetSdeltalambda} with $\Delta,\lambda >0$ arbitrary, but fixed, and $k_n:=\lfloor\log_2(n\lambda/4)\rfloor$. We define
\begin{equation*}\label{eq:eta}
\eta:=\left[1-3\exp\left(-\frac{1}{48}\left(\sqrt{\frac{n\lambda\Delta^2}{32}}-\sqrt{16\log\left(\frac{8}{\lambda\alpha\beta_{k_n}}\right)}\right)_+^2\right)\right]_+^2.
\end{equation*}
Under the assumptions of Theorem \ref{theorem:boundoverestimate} and if $n\lambda \geq 32$ and
\begin{equation*}
\left(n\lambda\right)^{-1}\log\left(\frac{8}{\lambda\alpha\beta_{k_n}}\right)\leq \frac{1}{512}
\end{equation*}
are satisfied, then uniformly in $\mathcal{S}_{\Delta,\lambda}$
\begin{equation}\label{eq:epineq}
\Pj_{(\mu, \sigma^2)}\left(\hat{K} < K\right)\leq 1-\eta^K \text{ and }\E_{(\mu, \sigma^2)}\left[\left(K-\hat{K}\right)_+\right]\leq K \left(1-\eta\right).
\end{equation}
\end{Theorem}
Roughly speaking, $\HSMUCE$ detects any change-point of the signal $\mu$ under assumptions of Theorem \ref{theorem:boundunderestimate} at least with probability $\eta$. A sharper version with different probabilities $\eta_1,\ldots,\eta_K$ is given in Theorem \ref{theorem:boundunderestimatesharp} in the supplement. Such a result clarifies the dependence on the different weights, but is technically way more difficult. Combining Theorems \ref{theorem:boundoverestimate}, \ref{theorem:overestimationk} and \ref{theorem:boundunderestimate} gives upper bounds for the probability $\Pj(\hat{K}\neq K)$ and the expectation $\E[\vert\hat{K}-K\vert]$ that $\HSMUCE$ missspecifies the number of change-points.
\begin{remark}[Vanishing errors]\label{remark:errorbounds}
For a fixed signal (fixed $\Delta$ and $\lambda$ are sufficient) both errors vanish asymptotically if $\alpha=\alpha_n\to 0$ is chosen such that $\log(\alpha_n\beta_{k_n,n})/n \to 0$, with triangular scheme $\beta_{1,n},\ldots,\beta_{\snn,n}$ for the weights in \eqref{eq:balancing}. We can achieve a rate arbitrary close to the exponential rate by the choice $\alpha_n = \exp(-n/r_n)$, with $r_n\to\infty$ arbitrarily slow. The condition on the sequence $\beta_{k_n,n}$ allows a variety of possible choices of the weights, too. For instance, the choice $\beta_{1,n}=\cdots=\beta_{\snn,n}=1/\snn$, which weights all scales equally, fulfils this condition.
\end{remark} 
A direct consequence is the \textit{strong model consistency} of $\HSMUCE$.
\begin{Theorem}[Strong model consistency]\label{theorem:consistency}
Assume the setting of Theorem \ref{theorem:boundoverestimate} and let $(\hat{K}_n)_n$ be the sequence of estimated numbers of change-points by $\HSMUCE$, where $\hat{K}_n$ is as $\hat{K}$ with significance level $\alpha_n$ and corresponding weights $\beta_{1,n},\ldots,\beta_{\snn,n}$. Moreover, let $\mathcal{S}_{\Delta,\lambda}$ be as in \eqref{eq:subsetSdeltalambda} with $\Delta,\lambda>0$ arbitrary, but fixed, and $k_n:=\lfloor\log_2(n\lambda/4)\rfloor$. Let $\rho>0$ be arbitrary, but fixed. If
\begin{equation}\label{eq:conditionconsistency}
\lim_{n\rightarrow \infty}{\frac{n^{1+\rho}}{\alpha_n}}=0 \text{ and } \lim_{n\rightarrow \infty}{\frac{\log\left(\alpha_n\beta_{k_n,n}\right)}{n}}=0
\end{equation}
holds, then $\hat{K}_n\to K$, almost surely and uniformly in $\mathcal{S}_{\Delta,\lambda}$.
\end{Theorem}
Again, there is a wide range of sequences $\alpha_n$ and $\beta_{k_n,n}$ to satisfy \eqref{eq:conditionconsistency}. Moreover, we still have (weak) model consistency, if $\alpha_n\to 0$ and the second condition of \eqref{eq:conditionconsistency} holds. 

\subsection{Confidence sets}\label{sec:confidencestatements}
In this section we obtain confidence sets for the signal $\mu$ and for the locations of the change-points. First, we show that the set of all solutions of \eqref{eq:optproblemhsmuce}
\begin{equation}\label{eq:defC}
C(\textbf{q}):=\left\{\mu\in\mathcal{M}\ :\ \vert\mathcal{I}(\mu)\vert = \hat{K}\text{ and } \max_{[\frac{i}{n},\frac{j}{n}]\in \mathcal{D}(\mu)} \left[T_i^j\big(Y,\mu([i/n,j/n])\big) - q_{ij}\right] \leq 0\right\}
\end{equation}
is a confidence set for the unknown signal $\mu$.
\begin{Theorem}[Confidence set]\label{theorem:confidencefunction}
Assume the setting of Theorem \ref{theorem:boundoverestimate} and let $\mathcal{S}_{\Delta,\lambda}$ be as in \eqref{eq:subsetSdeltalambda} with $\Delta,\lambda >0$ arbitrary, but fixed, and $k_n:=\lfloor\log_2(n\lambda/4)\rfloor$. Let $C(\cdot)$ be as in \eqref{eq:defC} and $\textbf{q}_n$ be a vector of critical values determined by significance level $\alpha$ and weights $\beta_{1,n},\ldots,\beta_{\snn,n}$, with $\lim_{n\rightarrow \infty}{n^{-1}\log(\beta_{k_n,n})}=0$. Then,
\begin{equation}\label{eq:confidenceset}
\lim_{n\rightarrow \infty}{\inf_{(\mu,\sigma^2)\in\mathcal{S}_{\Delta,\lambda}}\Pj_{(\mu, \sigma^2)}\left(\mu\in C\left(\textbf{q}_n\right)\right)}\geq 1-\alpha.
\end{equation}
\end{Theorem}
This shows that the asymptotic coverage of $C\left(\textbf{q}_n\right)$ is at least $1-\alpha$. Lemma \ref{lemma:confidencefunction} gives an exponential inequality similar to \eqref{eq:epineq} which shows that $C\left(\textbf{q}_n\right)$ is also a non-asymptotic confidence set. We further derive from this set confidence intervals for the change-point locations.
\begin{Theorem}[Change-point locations]\label{theorem:confidencelocation}
Assume the setting of Theorem \ref{theorem:confidencefunction}, where $\alpha$ is replaced by a sequence $\alpha_n\rightarrow 0$. Let $c_n:=r_n/n\leq \lambda/2$ and $k_n:=\lfloor\log_2(nc_n/2)\rfloor$ s.t. 
\begin{equation}\label{eq:conditioncnvanishing}
\liminf_{n\to\infty}{\frac{r_n}{\log(n)}} > \frac{216}{\min(\Delta^2,1)} \text{ and } \lim_{n\to\infty}{\frac{\log\left(\alpha_n\beta_{k_n,n}\right)}{r_n}}=0.
\end{equation}
Then,
\begin{equation}\label{eq:confidenceinterval}
\lim_{n\rightarrow \infty}{\sup_{(\mu,\sigma^2)\in \mathcal{S}_{\Delta,\lambda}}\Pj_{(\mu, \sigma^2)}\left(\sup_{\hat{\mu}\in C(\textbf{q}_n)}{\max_{k=1,\ldots,K}c_n^{-1}\left\vert \tau_k - \hat{\tau}_k\right\vert} > 1\right)}=0.
\end{equation}
\end{Theorem}
Here, the rate $c_n$ is equal to the sampling rate $1/n$ up to the (logarithmic) rate $r_n$ depending on the tuning parameters $\alpha_n\beta_{k_n,n}$. For example, if $\alpha_n\beta_{k_n,n} \asymp n^{-\gamma}, \gamma \geq 0$, $r_n/\log(n)\to\infty$ is sufficient to satisfy \eqref{eq:conditioncnvanishing}. A non-asymptotic statement is given in Lemma \ref{lemma:confidencelocation} in the supplement. For visualization of the confidence statements it is useful to further derive a confidence band $B(\textbf{q}_n)$ for the signal as in \citep[Corollary 3 and the explanation around]{SMUCE}. It can be shown that also the collection $I(\textbf{q}_n)=\{\hat{K}_n, B(\textbf{q}_n),[L_k,R_k]_{k=1,\ldots,\hat{K}_n}\}$, with $[L_k,R_k]$ confidence intervals for the change-point locations according to Theorem \ref{theorem:confidencelocation}, satisfies \eqref{eq:confidenceset}. Recall Figures \ref{fig:example} and \ref{fig:example1000} for an illustration. It is also possible to strengthen the statements of this section to sequences of vanishing signals with $\Delta_n\to 0$ and $\lambda_n\to 0$ slow enough, but we omit such results.

\subsection{Asymptotic detection rates for vanishing signals}\label{sec:rates}
For the detection of a single vanishing bump against a noisy background see Theorem \ref{theorem:onesignalrate} in the supplement. The following theorem deals with the detection of a signal with several vanishing change-points.
\begin{Theorem}[Multiple vanishing change-points]\label{theorem:multiscalesignalsrate}
Assume the heterogeneous gaussian change-point model \eqref{eq:model}. Let $(K_n)_n:=(\vert \mathcal{I}(\mu_n)\vert)_n$ be the sequence of true number of change-points. Let further $(\hat{K}_n)_n$ be the sequence of the estimated numbers of change-points by $\HSMUCE$ \eqref{eq:K}, with significance levels $\alpha_n$ and weights $\beta_{1,n},\ldots,\beta_{\snn,n}$. Let $\mathcal{S}_{\Delta_n,\lambda_n}\subset \mathcal{S}$ be a sequence of submodels as in \eqref{eq:subsetSdeltalambda} and $k_n:=\lfloor\log_2(n\lambda_n/4)\rfloor$. We further assume
\begin{equation}\label{conditionmultiscalesignalsrate}
\liminf_{n\to \infty}{\frac{n\lambda_n}{\log(n)}} > 512 \text{ and } \lim_{n\to\infty}{\frac{\log\left(\alpha_n\beta_{k_n,n}\right)}{n\lambda_n}}=0
\end{equation}
as well as
\begin{itemize}
\item[(1)] for large scales, i.e. $\liminf_{n>0}{\lambda_n}>0$, the limit $n\lambda_n\Delta_n^2\log(1/(\alpha_n\beta_{k_n,n}))^{-1}\rightarrow \infty$,
\item[(2)] for small scales, i.e. $\lambda_n\rightarrow 0$, the inequality
\begin{equation}\label{eq:ratesmallscales}
\sqrt{n\lambda_n}\Delta_n\geq \left(\sqrt{512}+C+\epsilon_n\right)\sqrt{-\log(\lambda_n)}
\end{equation}
with possibly $\epsilon_n\to 0$, but such that $\epsilon_n\sqrt{-\log(\lambda_n)}\rightarrow \infty$ and 
\begin{equation*}
\limsup_{n\to\infty}\frac{\sqrt{\log(8/(\alpha_n\beta_{k_n,n}))}}{\epsilon_n\sqrt{-\log(\lambda_n)})} < \frac{1}{\sqrt{512}},
\end{equation*}
with $C=0$ for $K_n$ bounded and $C=16\sqrt{6}$ for $K_n$ unbounded.
\end{itemize}
Then,
\begin{equation*}
\lim_{n\rightarrow \infty}{\sup_{(\mu_n,\sigma_n^2)\in\mathcal{S}_{\Delta_n,\lambda_n}}\Pj_{(\mu_n, \sigma_n^2)}\left(\hat{K}_n < K_n\right)}=0.
\end{equation*}
\end{Theorem}
Theorems \ref{theorem:onesignalrate} and \ref{theorem:multiscalesignalsrate} state conditions on the tuning parameters $\alpha_n$ and $\beta_{k_n,n}$ as well as on the length of the minimal scale $\vert I_n\vert=:\lambda_n$ (to simplify notations we only write $\lambda_n$ in the following) and the standardized jump size $\Delta_n$ to detect the vanishing signals uniformly over $\mathcal{S}_{\Delta_n,\lambda_n}$. If, in addition, $\lim_{n\rightarrow \infty}{\alpha_n}=0$ holds, then we control also the probability to overestimate the number of change-points and therefore the estimation of the number of change-points is still consistent in the case of a vanishing signal. The main condition in both theorems is that $\sqrt{n\lambda_n}\Delta_n$ has to be at least of order $\sqrt{-\log(\lambda_n)}$, see \eqref{eq:onesignalrate} and \eqref{eq:ratesmallscales}. This is optimal in the sense that no signal with a smaller rate can be detected asymptotically with probability one, see \citep{{DuembgenSpokoiny01},{ChanWalther13},{SMUCE}} for the case of homogeneous observations, and note that this is a sub-model of our model. But different to the homogeneous case we need, in addition, that $\lambda_n$ is at least of order $\log(n)/n$, see \eqref{eq:onesignalrateIn} and \eqref{conditionmultiscalesignalsrate}. Such a restriction appears reasonable, since for the additional variance estimation only the number of observation on the segment is relevant and not the size of the change. Finally, we observe that the constants encountered in the lower detection bound for $\HSMUCE$ in \eqref{eq:onesignalrate} and \eqref{eq:ratesmallscales} increase with the difficulty of the estimation problem, where the difficulty is represented by the number of vanishing segments. All of these constants are a little bit larger as the analogue constants for $\SMUCE$ in \citep[Theorem 5 and 6]{SMUCE} reflecting the additional difficulty encountered by the heterogeneous noise. More precisely, we have $4$ instead of the optimal $\sqrt{2}$ for one vanishing segment, $\sqrt{512}$ instead of $4$ for a bounded number of vanishing segments and $\sqrt{512}+16\sqrt{6}$ instead of $12$ for an unbounded number of vanishing segments. Note again, that the optimal constants for the heterogeneous case are unknown to us.

\subsection{Choice of the tuning parameters}\label{sec:tuningparameters}
In this section we discuss the choice of the tuning parameters $\alpha$ and $\beta_1,\ldots,\beta_{\sn}$. 
\paragraph*{\textbf{Choice of $\alpha$}}
As illustrated in Figure \ref{fig:example1000} the choice depends on the application. If a strict overestimation control of the number of change-points $K$ is desirable $\alpha$ should be chosen small, e.g. $0.05$ or $0.1$, recall Theorems \ref{theorem:boundoverestimate} and \ref{theorem:overestimationk}. This might come at the expense of missing change-points but with large probability not detecting too many (recall Figure \ref{fig:example1000} and see also the simulations in Section \ref{sec:simulations}). If change-point screening is the primarily goal, i.e. we aim to avoid missing of change-points, $\alpha$ should be increased, e.g. $\alpha = 0.5$ or even higher, since Theorem \ref{theorem:boundunderestimate} shows that the error probability to underestimate the number of change-points decreases with increasing $\alpha$. If model selection, i.e. $\hat{K}=K$, is the major aim, an intermediate level that balances the over- and underestimation error should be chosen, e.g. $\alpha$ between $0.1$ and $0.5$. Both errors vanish super polynomially for the asymptotic choice $\alpha=\alpha_n\in\exp(-o(n))$, see Remark \ref{remark:errorbounds}. A finite sample approach is to weight these error probabilities $\gamma\Pj(\hat{K}>K)+(1-\gamma)\Pj(\hat{K}<K)$, with $\gamma\in(0,1)$, and to choose $\alpha$ such that its upper bound
\begin{equation*}
\gamma\alpha + (1-\gamma)\left(1-\left[1-3\exp\left(-\frac{1}{48}\left(\sqrt{\frac{n\lambda\Delta^2}{32}}-\sqrt{16\log\left(\frac{8}{\lambda\alpha\beta_k}\right)}\right)_+^2\right)\right]_+^{2K}\right)
\end{equation*}
is minimized. This also allows to incorporate prior information on $(\lambda,\Delta)$. Alternatively, the bound on the expectation $\E[\vert\hat{K}-K\vert]$ by combining Theorems \ref{theorem:overestimationk} and \ref{theorem:boundunderestimate} can be minimized to take the size of the missestimation into account. Despite of all possibilities to choose the 'best' $\alpha$ for a given application, comparing estimates at different $\alpha$ can be helpful to trace the "stability of evidence" of the estimated change-points at different significance levels. Of course, the interpretation of such a "significance screening" does not allow for a frequentist interpretation of a significance level anymore as $\alpha$ has to be fixed in advance, see e.g.~\citep{Schervish96}. Nevertheless, it might give for instance some indication whether to perform further experiments. Despite of this, for a fixed $\alpha$ the confidence statements of $\HSMUCE$ can also be used to support findings by other estimators. This is illustrated in Section \ref{sec:ion} for the ion channel application. 
\paragraph*{\textbf{Choice of $\beta_1,\ldots,\beta_{\sn}$}} 
As a default choice we recommend equal weights $\beta_1=\cdots=\beta_\sn=1/\sn$. This choice fulfils (together with many other choices) the conditions of the Theorems \ref{theorem:consistency} and \ref{theorem:confidencefunction}. Unlike as for the significance level $\alpha$ only the bound for the underestimation of the number of change-points depends on these weights. Note, that this gives the user the possibility to incorporate prior information on the scales without violating the overestimation control in Theorems \ref{theorem:boundoverestimate} and \ref{theorem:overestimationk}: If for instance changes are expected to occur only on small segments then the detection power on these scales can be increased if the first weights are chosen large and the other ones small (or even zero). In contrast, if the general signal to noise ratio is expected to be very small then it is nearly impossible to detect changes on small scales and larger scales should be weighted more to detect at least the changes on these scales. A quantitative influence of the weights on the detection power can be seen in the underestimation bound in Theorem \ref{theorem:boundunderestimatesharp} in Supplement \ref{sec:proofs} which is a refinement of Theorem \ref{theorem:boundunderestimate}. We also investigate such choices quantitatively in simulations in Section \ref{sec:simulations}. 

\section{Simulations}\label{sec:simulations}
In this section we compare $\HSMUCE$\footnote{\url{http://www.stochastik.math.uni-goettingen.de/hsmuce}, v. 0.0.0.9000, 2015-04-15} in simulations with $\CBS$ \citep{CBS07}, $\cumSeg$ \citep{MuggeoAdelfio11} and $\LOOVF$ \citep{ARLOT} as they are also designed to be robust against heterogeneous noise. Moreover, we include $\SMUCE$ \citep{SMUCE} in simulations with a constant variance as a benchmark to examine how much the detection power of $\HSMUCE$ decreases in this case, which may be regarded as the price for adaptation to heterogeneous noise. We fix the weights $\beta_1,\ldots,\beta_\sn=1/\sn$ and vary the significance level $\alpha$. A simulation with tuned weights can be found in Section \ref{sec:simprior} in the Supplement. For circular binary segmentation ($\CBS$) we call the function \textit{segmentByCBS}\footnote{\url{http://cran.r-project.org/web/packages/PSCBS/}, v. 0.40.4, 2014-02-04} with the standard parameters. 
For the cross-validation method $\LOOVF$ we use the Matlab function \textit{proc\_LOOVF}\footnote{\url{http://www.di.ens.fr/~arlot/code/CHPTCV.htm}, v. 1.0, 2010-10-27} with the parameter choice of the demo file. For $\cumSeg$ we call the method \textit{jumpoints}\footnote{\url{http://cran.r-project.org/web/packages/cumSeg/}, v. 1.1, 2011-10-14} with the parameter $k$ large enough such that the estimation is not influenced by this choice. For $\SMUCE$ we call the function \textit{smuceR}\footnote{\url{http://cran.r-project.org/web/packages/stepR}, v. 1.0-3, 2015-06-18} with the standard parameters, in particular the interval set of all intervals is used if $n\leq 1\,000$.\\
To avoid specific interactions between the signal and the dyadic partition we generate in each repetition a random pair $(\mu_R,\sigma_R^2)\in\mathcal{S}$ (all random variables are independent from each other).
\begin{enumerate}
\item[(a)] We fix the number of observations $n$, the number of change-points $K$, a constant $C$ and a minimum value for the smallest scale $\lambda_{\min}$.
\item[(b)] We draw the locations of the change-points $\tau_0:=0<\tau_1< \cdots <\tau_K<1=:\tau_{K+1}$ uniformly distributed with the restriction that $\lambda:=\min_{k=0,\ldots,K}{\vert \tau_{k+1}-\tau_{k}\vert} \geq \lambda_{\min}$.
\item[(c)] We choose the function values $\valsd_0,\ldots,\valsd_K$ of the standard deviation function $\sigma_R$ by $\valsd_k:=2^{U_k}$, where $U_0,\ldots,U_K$ are uniform distributed on $[-2,2]$.
\item[(d)] We determine the function values $\valmu_0,\ldots,\valmu_K$ of the signal $\mu_R$ such that
\begin{equation}\label{eq:C}
\vert \valmu_k - \valmu_{k-1}\vert = \sqrt{\frac{C}{n}\min\left(\frac{\tau_{k+1}-\tau_{k}}{\valsd_k^2}, \frac{\tau_k-\tau_{k-1}}{\valsd_{k-1}^2}\right)^{-1}}\ \forall\ k = 1,\ldots,K.
\end{equation}
Thereby, we start with $\valmu_0 = 0$ and choose randomly with probability $1/2$ whether the expectation increases or decreases. 
\end{enumerate}
By \eqref{eq:C} we provide a situation where all change-points are similarly hard to find, recall the minimax detection boundary from Section \ref{sec:rates}. An example has been displayed in Figure \ref{fig:example1000} in the introduction, where $\HSMUCE$ misses at $\alpha = 0.1$ one change-point and detects for $\alpha$ between $0.15$ and $0.99$ (only displayed for $\alpha = 0.3$ and $\alpha = 0.5$) the correct number of change-points. In Figure \ref{fig:example1000other} (Supplement \ref{sec:add}) we see that $\CBS$ \citep{CBS07} finds also all change-points, but detects further changes. Less good is the performance of $\cumSeg$ \citep{MuggeoAdelfio11} and $\LOOVF$ \citep{ARLOT} which both miss several changes and $\LOOVF$ adds also a false positive. We examine these methods now more extensively. All simulations are repeated $10\,000$ times.\\
In the following we report the difference between the estimated $\hat{K}$ and the true number $K$ of change-points as well as the mean of the absolute value of this difference. Additionally, we use the false positive sensitive location error
\begin{equation*}
\operatorname{FPSLE}=\frac{n}{2\hat{K}}\sum_{k=1}^{\hat{K}+1}{\vert \tau_{l_k-1} - \hat{\tau}_{k-1}\vert + \vert \tau_{l_k} - \hat{\tau}_{k}\vert},
\end{equation*}
with $l_k\in\{1,\ldots,K+1\}$ such that $(\hat{\tau}_{k-1}+\hat{\tau}_k)/2\in (\tau_{l_k-1},\tau_{l_k}]$, i.e. the left and right neighbouring change-points to the middle point of $(\hat{\tau}_{k-1},\hat{\tau}_k]$, and the false negative sensitive location error \begin{equation*}
\operatorname{FNSLE}=\frac{n}{2K}\sum_{k=1}^{K+1}{\vert \tau_{k-1}-\hat{\tau}_{l_k-1}\vert + \vert \tau_{k}-\hat{\tau}_{l_k}\vert},
\end{equation*}
with $l_k\in\{1,\ldots,\hat{K}+1\}$ such that $(\tau_{k-1}+\tau_k)/2\in (\hat{\tau}_{l_k-1},\hat{\tau}_{l_k}]$, see \citep[Section 3.1]{Futschiketal14}, to rate the estimation of the locations of the change-points. We also show the mean integrated squared (absolute) error $\operatorname{MISE}$ ($\operatorname{MIAE}$) for all methods.

\subsection{Simulation results}\label{sec:simresults}
In this section we discuss the results of the simulations for model \eqref{eq:model} and \eqref{eq:S}. We start in Table \ref{tab:fixed} (Supplement \ref{sec:add}) with the simple setting of a single change at the midpoint, where we vary the variances on the adjoining segments. In Table \ref{tab:constvar} (Supplement \ref{sec:add}) we display results for a constant variance and in Table \ref{tab:random10000} (Supplement \ref{sec:add}) for heterogeneous errors. We excluded $\LOOVF$ from simulations for larger $n$ due to its large computation time, confer the run time simulations in Section \ref{sec:computationtime} in the supplement.\\
All simulations confirm the overestimation control $\alpha$ for $\HSMUCE$ from Theorem \ref{theorem:boundoverestimate} and the exponential decay of the overestimation in Theorem \ref{theorem:overestimationk}. The simulations with a single change-point confirm that the size of the variance change has no influence, rather the size of the variances matters. We found that $\HSMUCE$ performs well compared to all other methods. A small $\alpha$ avoids overestimation, but risks to miss changes that are harder to detect. Thus, the comparison of the estimates of $\HSMUCE$ for different $\alpha$ shows in accordance with our theory that it is reasonable to relax $\alpha$ if changes are expected to be harder to detect (recall the discussion in Section \ref{sec:tuningparameters}). From the other methods $\cumSeg$ performs best in the easier and $\LOOVF$ in the difficult scenarios, whereby $\CBS$ and in particular $\LOOVF$ shows a tendency to overestimate the number of change-points.\\
For a constant signal (corresponding to $K=0$ in Table \ref{tab:constvar}) $\HSMUCE$ overestimates the number of change-points even slightly less than $\SMUCE$, whereas $\CBS$ and $\cumSeg$ overestimate hardly ever. In the case of a constant variance we found that the detection power of $\HSMUCE$ is only slightly worse than $\SMUCE$ for $K=2$, although $\SMUCE$ used instead of the dyadic partition $\mathcal{D}$ the system of all intervals. The difference is larger for $K=10$ and in this case also $\CBS$ and $\cumSeg$ performs better than $\HSMUCE$, since the detection power of $\HSMUCE$ depends strongly on the lengths of the constant segments. Moreover, $\lambda_{\min}$ plays a similar role as the number of change-points $K$, since the average constant segments length decreases if $\lambda_{\min}$ decreases or $K$ increases. Worse results for smaller lengths are due to the family\-wise error control $\alpha$ of $\HSMUCE$ as it guarantees a strict control of overestimating the number of change-points.\\
Similar results can be observed for $n=100$ with heterogeneous errors. $\CBS$ performs better than $\cumSeg$ and $\LOOVF$, and in particular better than in the single change-point setting. $\CBS$ outperforms $\HSMUCE$ for $K=5$, although $\HSMUCE$ has a much smaller tendency to overestimate the number of change-points, whereas in particular $\CBS$ and $\LOOVF$ tend to overestimation. This can also be seen for the $\operatorname{MISE}$ and $\operatorname{MIAE}$ as these measures are much more affected by underestimation than by overestimation. These findings are also supported by the $\operatorname{FPSLE}$ and the $\operatorname{FNSLE}$, the $\operatorname{FPSLE}$ is heavily affected by overestimation, whereas the $\operatorname{FNSLE}$ is larger in case of underestimation.\\ 
In all simulations with heterogeneous errors and $1\,000$ observations $\HSMUCE$ outperforms the other methods, for $10\,000$ observations this becomes even more pronounced. In comparison to the simulation with $100$ observations the tendency of $\CBS$ to overestimate the number of change-points becomes then also more prominent. Finally, in further simulations (not displayed) we found that the detection power of all methods decreases for smaller $C$ in (d), but all results remain qualitatively the same. All in all, we found that $\HSMUCE$ performs well as sample size becomes larger, in particular if the constant segments are not too short as indicated by assumption \eqref{conditionmultiscalesignalsrate} in Theorem \ref{theorem:multiscalesignalsrate}.\\
A comparison of Table \ref{tab:random10000} and \ref{tab:beta10000} (Supplement \ref{sec:add}) shows that tuned weights increase the detection power of $\HSMUCE$ for all significance levels, so we encourage the user to adapt the weights if prior information on the scales where changes occur is available. Details how the weights are chosen can be found in Section \ref{sec:simprior} in the supplement.

\subsection{Robustness against model violations}\label{sec:simrobust}
We begin by investigating how robust the methods are against a violation of the assumption that the standard deviation changes only at the same locations as the mean changes. We consider continuous changes as well as abrupt changes. The exact functions for the standard deviation can be seen in Figure \ref{fig:std} (Supplement \ref{sec:add}). In Table \ref{tab:robustvar} (Supplement \ref{sec:add}) we see that $\HSMUCE$ and $\CBS$ perform very robust against heterogeneous noise on the constant segments, whereas, remarkably, the detection power of $\cumSeg$ is even improved. Moreover, in additional simulations (not displayed) with less observations we found that $\LOOVF$ is very robust, too.\\
Moreover, we examine robustness against small periodic trends in the mean in simulations similar to those in \citep{CBS04}, also adapted to the inhomogeneous variance. The exact simulation setting can be found in Section \ref{sec:addrobustness} (Supplement B). We obtain from Table \ref{tab:robustnesstrend} that $\HSMUCE$ shows similar results for small trends compared to the simulation without trend for small trends, but is affected by larger trends, in particular if these are not scaled by the standard deviation. $\CBS$ overestimates heavily in all cases, whereas $\cumSeg$ (although not affected by the trend) shows over- and underestimation.\\
Furthermore, we investigate robustness against heavy tails of the error distribution. In Table \ref{tab:robustnonnormal} (Supplement \ref{sec:add}) we consider $t_3$-distributed errors which are scaled such that the expectation and the standard deviation are the same as in Section \ref{sec:simresults}. As expected $\SMUCE$ is not robust against heavy tails (as it misinterprets extreme values as a change in the signal, whereas $\HSMUCE$ provides reasonable results. In comparison to gaussian errors $\HSMUCE$ is not influenced for $K=0$, underestimation is more distinct in the constant variance scenario and detection power is even increased in the scenario with heterogeneous errors. In comparison, $\CBS$ is not influenced for $K=0$, too, underestimates and overestimates in the constant variance scenario and is slightly worse with a tendency to underestimation in the scenario with heterogeneous errors, whereas $\cumSeg$ overestimates rarely, but heavily for $K=0$, underestimates and overestimates in the constant variance scenario and is robust in the last scenario.\\
In summary, $\HSMUCE$ seems to be robust against a wide range of variance changes on constant segments and seems to be only slightly affected by larger tails than gaussian, in particular no tendency to overestimation was visible in our simulations. This may be explained by the fact that the local likelihood tests of $\HSMUCE$ are quite robust against heterogeneous noise, see for instance \citep{{BakirovandSzekely05},{IbragimovandMueller10}}, and against non-normal errors, see \citep{LehmannRomano05} and the references therein. Unlike the number of change-points, the locations are sometimes miss-estimated, since the restricted maximum likelihood estimator is influenced by changes of the variance. Instead, more robust estimators, for instance local median and MAD estimators, could be used.

\section{Application to ion channel recordings}\label{sec:ion}
In this section we apply $\HSMUCE$ to current recordings of a porin in planar lipid bilayers performed in the Steinem lab (Institute of Organic and Biomolecular Chemistry, University of G{\"o}ttingen). Porins are $\beta$-barrel proteins present in the outer membrane of bacteria and in the outer mitochondrial membrane of eukaryotes \citep{Benz94,Schirmer98}. Due to their large pore diameter they enable passive diffusion of small solutes like ions or sugars. The partial blockade of the pore by an internal loop results in gating that can be detected using the voltage clamp technique \citep{SakmannNeher95}. We aim to detect the gating automatically, since in many ion channel applications hundred or more datasets each with several hundredthousands data points have to be analysed. For noise reduction the data was automatically preprocessed in the amplifier with an analogue four-pole Bessel low-pass filter of 1 kHz. Hence, the noise is coloured, but the correlation is less than $10^{-3}$ if the trace is subsampled by eleven or more observations, see \citep[(6)]{Hotzetal13}, which has been done in the following. Finally, we apply to $882$ subsampled observations $\HSMUCE$, $\CBS$, $\cumSeg$ and $\LOOVF$.

\begin{figure}[!htp]
\centering
\begin{subfigure}{\textwidth}
\includegraphics[width=0.98\textwidth]{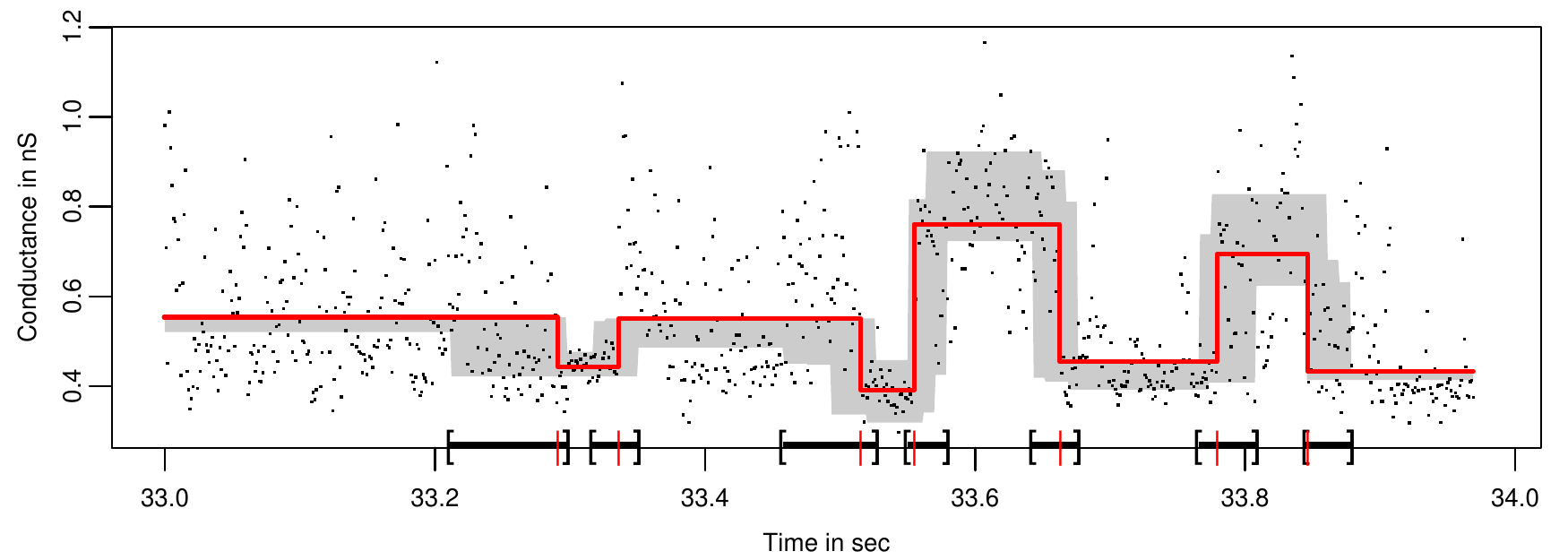}
\subcaption{$\alpha = 0.05$.}
\label{subfig:ion2}
\end{subfigure}

\begin{subfigure}{\textwidth}
\includegraphics[width=0.98\textwidth]{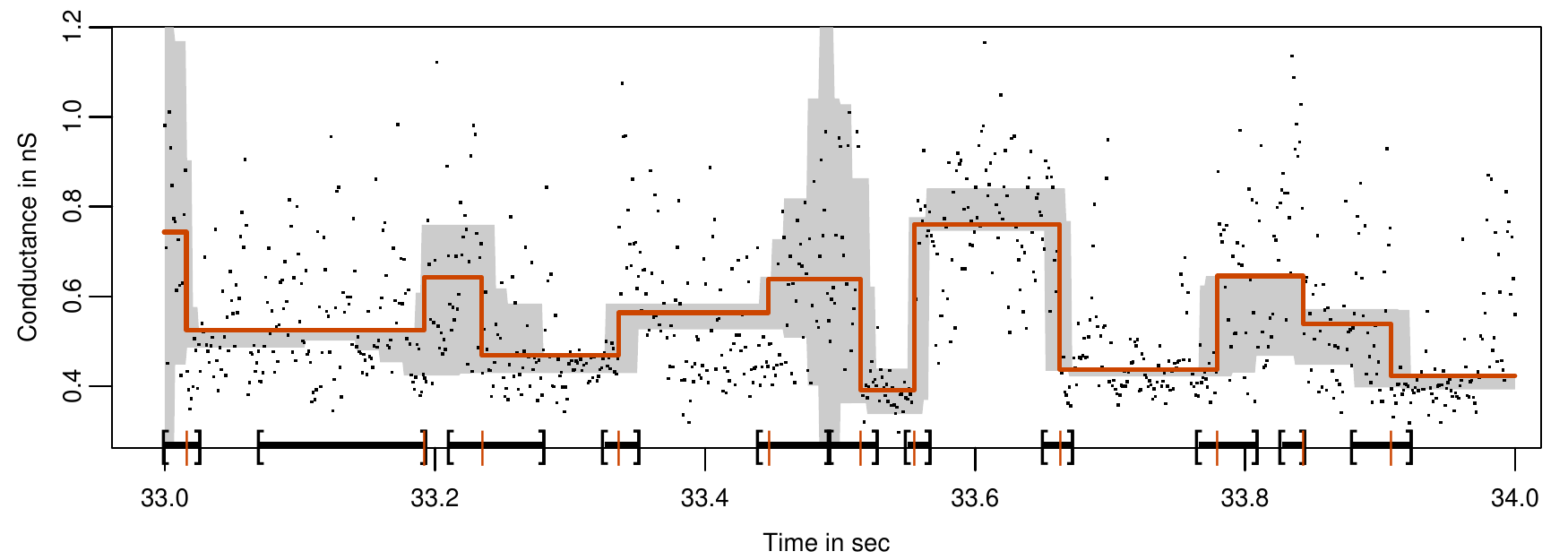}
\subcaption{$\alpha = 0.5$.}
\label{subfig:ionalpha}
\end{subfigure}

\begin{subfigure}{\textwidth}
\includegraphics[width=0.98\textwidth]{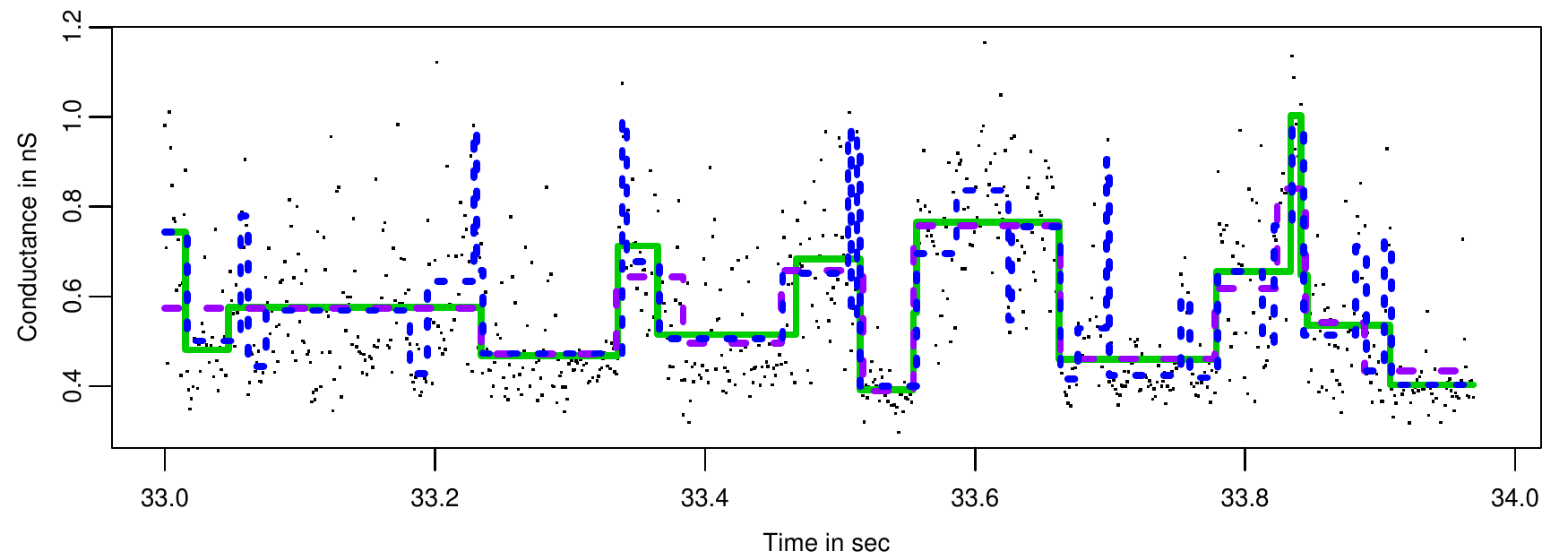}
\subcaption{Estimates by $\CBS$ (green), $\cumSeg$ (purple) and $\LOOVF$ (blue).}
\label{subfig:other2}
\end{subfigure}

\caption{\subref{subfig:ion2}, \subref{subfig:ionalpha}: Subsampled observations (black points) together with the confidence band (grey), the confidence intervals for the change-point locations (brackets and thick lines), the estimated change-points locations (red dashes) and the estimate (red line) by $\HSMUCE$ at different $\alpha$. \subref{subfig:other2}: other estimates.}
\label{fig:ion2all}

\end{figure}

In Figure \ref{subfig:ion2} we see that the signal fluctuates around two or more levels, the so called open (higher conductivity, larger current measurements) and closed (lower conductivity, smaller current measurements) states. Moreover, the variance in the open states is larger than in the closed state, a well known phenomenon denoted as open channel noise (\citet[Section 3.4.4]{SakmannNeher95} and the references therein) which arises for larger ion channels such as porins from conformational fluctuations in the channel protein \citep{Sigworth85}. Due to the pronounced heterogeneity in the variance, methods which assume a constant variance fail to reconstruct the gating, see Figure \ref{fig:example} in the introduction for an illustration. In contrast, $\HSMUCE$ at $\alpha=0.05$ provides a reasonable fit that covers the main features of the data. Additional smaller changes are found by $\CBS$, $\cumSeg$ and $\LOOVF$, see Figure \ref{subfig:other2}. These changes might be explained by some uncontrollable base line fluctuations caused for instance by small holes in the membrane due to movements of the lipids. On the other hand, we found in the simulations, see Table \ref{tab:random10000} and the example in Figure \ref{fig:example1000other} (both Supplement \ref{sec:add}), that $\CBS$ and $\LOOVF$ tend to include small artificial changes, whereas we saw in Table \ref{tab:robustnesstrend} (Supplement \ref{sec:add}) that $\HSMUCE$ is quite robust against small periodic trends in the signal. For illustrative purposes, in order to examine these changes further we increase in Figure \ref{subfig:ionalpha} the significance level $\alpha=0.5$ and detect for instance changes around 33.0s and 33.2s, too. Taking also the confidence regions of $\HSMUCE$ into account confirms several changes with high "significance" (e.g. the reconstruction of $\CBS$ between 33.5 and 33.8) and further changes with less "significance" (e.g. the changes around 33.0s and 33.2s). Other changes could not be confirmed by $\HSMUCE$ at any reasonable significance level (e.g. the peaks of $\CBS$ and $\LOOVF$ at 33.85). In this spirit $\HSMUCE$ can always be used to accompany any segmentation method to help to identify its significant changes. Recall, that of course a frequentist statistical error control is only given when $\alpha$ is fixed in advance.

\section*{Acknowledgement}
Support of DFG CRC 803 - Z2 and FOR 916 - B3 is gratefully acknowledged. We are grateful to various comments by two referees and an associate editor which lead to a significant improvement of the paper. We thank A. Bartsch, O. Sch{\"u}tte and C. Steinem (Institute of Organic and Biomolecular Chemistry, University of G{\"o}ttingen) for providing the ion channel data and helpful discussions. We also thank T. Aspelmeier, M. Behr, A. Hartmann, H. Li and I. Tecuapetla for helpful comments to the paper and the R-package.

\newpage
\section*{Supplement to\\ Heterogeneous Change Point Inference\\ by Florian Pein, Hannes Sieling and Axel Munk}
\appendix

\section{Computation}\label{sec:computation}
In this section we detail the computation of the estimator $\HSMUCE$ (Section \ref{sec:comphsmuce}) and of the critical values $q_1,\ldots,q_\sn$ (Section \ref{sec:detcriticalvalues}). We also examine the computation time (Section \ref{sec:computationtime}) theoretically and empirically. An R-package is available online\footnote{\url{http://www.stochastik.math.uni-goettingen.de/hsmuce}}.

\subsection{Computation of the estimator}\label{sec:comphsmuce}
First of all, we obtain from the multiscale test the bounds
\begin{equation}\label{eq:bounds}
[\underline{b}_{i,j},\overline{b}_{i,j}]:=\left[\overline{Y}_{ij} - \sqrt{\frac{q_{ij}\hat{\valsd}^2_{ij}}{j-i+1}},\overline{Y}_{ij} + \sqrt{\frac{q_{ij}\hat{\valsd}^2_{ij}}{j-i+1}}\right]
\end{equation}
for $\mu$ on the interval $[i/n,j/n]\in\mathcal{D}$. Therefore, $\HSMUCE$ can be computed as in \citep[Section 3]{SMUCE} for $\SMUCE$ described. However, in what follows  we give a modification of the algorithm which reduces the computation time remarkably due to the small number of intervals $\OO(n)$ in the dyadic partition $\mathcal{D}$. Here, we compute first left and right limits for the location of the change-points and then start the dynamic program restricted to these intervals. A notable difference to \citep{{PELT},{SMUCE}} is that this approach leads also to pruning in the forward step of the dynamic program. More precisely, we define the intersected bounds as
\begin{equation*}
\underline{B}_{i,j}:=\max_{\substack{i \leq s < t \leq j\\ [s/n,t/n]\in\mathcal{D}}}{\underline{b}_{s,t}} \text{ and } \overline{B}_{i,j}:=\min_{\substack{i \leq s < t \leq j\\ [s/n,t/n]\in\mathcal{D}}}{\overline{b}_{s,t}}
\end{equation*}
and set recursively
\begin{equation*}
L_k:=\min\left\{1 < r \leq L_{k+1}-1\ :\ \underline{B}_{r,L_{k+1}-1}\leq \overline{B}_{r,L_{k+1}-1}\right\},
\end{equation*}
for $k=\hat{K},\ldots,1$, with $L_{\hat{K}+1}:=n+1$. The right limits are defined as
\begin{equation*}
R_k:=\min\left\{R_{k-1} < r \leq n\ :\ \underline{B}_{R_{k-1},r}> \overline{B}_{R_{k-1},r}\right\},
\end{equation*}
for $k=1,\ldots,\hat{K}$, with $R_0:=1$. In other words, the left limit for the $k$-th change-point $L_k$ is the smallest number $1 < r \leq n$ such that between $Y_r$ and $Y_n$ a piecewise constant solution with $\hat{K}-k$ change-points exists which respects the bounds \eqref{eq:bounds}. Analogously, the right limit $R_k$ is the smallest number $1 < r \leq n$ such that between $Y_1$ and $Y_r$ no piecewise constant solution with $k-1$ change-points exists which fulfils the bounds \eqref{eq:bounds}. Note, that we do not have to compute the right limits separately, since we can just start the dynamic program at $L_k$ and stop if another change-point has to be included. It follows that the $k$-th change-point $\hat{\tau}_k$ has to be in the confidence interval $[L_k/n,R_k/n]$, since otherwise an additional change-point would be necessary to fulfil the multiscale constraints.

\subsection{Computation of the critical values}\label{sec:detcriticalvalues}
In this section we show how the critical values can be computed by Monte-Carlo simulations. Note first that the following method uses only the continuity and the monotonicity of the cumulative distribution functions of the statistics $T_1,\ldots,T_\sn$ and therefore the methodology can also be used for other multiscale tests, see for instance the extension to other interval sets in Remark \ref{remark:otherintervalsets}.\\
Let $M$ be the number of simulations and $(T_{1,1},\ldots,T_{\sn,1}),\ldots,(T_{1,M},\ldots,T_{\sn,M})$ be i.i.d. copies of the vector $(T_1,\ldots,T_\sn)$. Moreover, we denote by $F_M(\cdot)$ the empirical distribution function of $(T_1,\ldots,T_\sn)$ and by $F_{M,k}(\cdot)$ the empirical distribution function of the random variable $T_k$. Then, we aim to find a vector of critical values $\widehat{\textbf{q}}_M=(\widehat{q}_{M,1},\ldots,\widehat{q}_{M,\sn})$ which satisfies with
\begin{equation}\label{eq:empsignificancelevel}
\alpha-\frac{1}{M}< 1-F_M\left(\widehat{\textbf{q}}_M\right)\leq \alpha,
\end{equation}
an empirical version of condition \eqref{eq:significancelevel}, and with
\begin{equation}\label{eq:empbalancing}
\frac{1-F_{M,j_1}(\widehat{q}_{M,j_1})}{\beta_{j_1}}\leq\frac{1-F_{M,j_2}(\widehat{q}_{M,j_2})+\frac{1}{M}}{\beta_{j_2}}\quad \text{for all}\quad j_1,j_2\in\{1,\ldots,\sn\},
\end{equation}
an empirical version of condition \eqref{eq:balancing}. In the following we propose an iterative method to determine such a vector and show afterwards that this vector converges almost surely to the vector of critical values defined by \eqref{eq:significancelevel} and \eqref{eq:balancing}. As the $k$-th entry of the starting vector we choose the empirical $(1-\alpha\beta_k)$-quantile of the statistic $T_k$, since the vector with these values satisfies condition \eqref{eq:empbalancing} and the inequality
\begin{equation*}
1-F_M\left(\cdot\right)\leq\alpha.
\end{equation*}
Afterwards, we reduce the entries until the lower bound from condition \eqref{eq:empsignificancelevel} is satisfied, too. To ensure condition \eqref{eq:empbalancing} in every iteration, we always reduce the entry which has the smallest ratio 
\begin{equation*}
\frac{1-F_{M,k}(\widehat{q}_{M,k})}{\beta_{k}}.
\end{equation*}
In Algorithm \ref{alg:balancingq} the determination of the critical values is summarized in pseudocode.
\begin{algorithm}[ht]
\caption{Determination of the critical values.}
\label{alg:balancingq}
\begin{algorithmic}[1]
\Require The statistics $T_1,\ldots,T_\sn$ as well as the significance level $\alpha\in(0,1)$, the weights $\beta_1,\ldots,\beta_\sn>0$, with $\sum_{k=1}^{\sn}{\beta_k}=1$, and the number of simulations $M\in\N$.
\Ensure The vector of critical values $\widehat{\textbf{q}}_M=(\widehat{q}_{M,1},\ldots,\widehat{q}_{M,\sn})$ which fulfils the conditions \eqref{eq:empsignificancelevel} and \eqref{eq:empbalancing}.
\For{$i=1,\ldots,M$}
\State $(T_{1,i},\ldots,T_{\sn,i}) \leftarrow \text{ realisation of } (T_1,\ldots,T_\sn)$
\EndFor
\For{$k=1,\ldots,\sn$}
\State $(S_{k,1},\ldots,S_{k,M}) \leftarrow \operatorname{sort}\left((T_{k,1},\ldots,T_{k,M})\right)$
\State $w_k \leftarrow M-\lfloor \alpha\beta_k M\rfloor$
\EndFor
\Repeat
\State $\hat{k}\leftarrow \argmin_{k=1,\ldots,\sn}{\beta_k^{-1}\left(1-F_{M,k}(S_{k,w_k})\right)}$
\State $w_{\hat{k}}\leftarrow w_{\hat{k}} - 1$
\Until{$1-F_M\left(S_{1,w_1},\ldots,S_{m,w_\sn}\right) > \alpha$}
\State $w_{\hat{k}}\leftarrow w_{\hat{k}} + 1$
\Return $S_{1,w_1},\ldots,S_{m,w_\sn}$
\end{algorithmic}
\end{algorithm}

The method has the advantage that we do not need specific assumptions on the distribution of the vector $(T_1,\ldots,T_\sn)$ and still get critical values which are adapted to the exact finite sample distribution of $(T_1,\ldots,T_\sn)$ and ensure therefore even for a finite number of observations the significance level $\alpha$.\\
The following theorem shows the convergence of this algorithm to $\textbf{q}=(q_1,\ldots,q_\sn)$.
\begin{Theorem}[Consitency of Monte-Carlo critical values]\label{theorem:convergencecritval}
The empirical vector of critical values $\widehat{\textbf{q}}_M=(\widehat{q}_{M,1},\ldots,\widehat{q}_{M,\sn})$ converges almost surely in the number of simulations $M$ to the vector of critical values $\textbf{q}=(q_1,\ldots,q_\sn)$ defined by \eqref{eq:significancelevel} and \eqref{eq:balancing}.
\end{Theorem}
The computation time is dominated by the generation of the $M$ i.i.d. copies of the vector $(T_1,\ldots,T_\sn)$. Therefore, we store the generated realizations and recycle them. To avoid memory problems we only store the realizations for every dyadic number, because the significance level $\alpha$ is still satisfied if we determine the critical values based on realizations with a larger number of observations, since then the maxima in $(T_1,\ldots,T_\sn)$ are taken over more intervals. To this end, the choice $M=10\,000$ seems to be a good trade-off between computation time and approximation accuracy. 

\subsection{Computation time}\label{sec:computationtime}
In this section we discuss the theoretical computation time of $\HSMUCE$ and compare it later in simulations with $\CBS$, $\cumSeg$ and $\LOOVF$. We stress that the computation time for the bounds, for the limits $L_1,\ldots,L_{\hat{K}}$ (and so for $\hat{K}$) and for the optimization problem \eqref{eq:optproblemhsmuce}, and therefore of all confidence sets, is always $\OO(n)$. Hence, the computation time is dominated by the determination of the restricted maximum likelihood estimator by dynamic programming. 
\begin{Lemma}[Computation time]\label{lemma:computationtime}
The algorithm has data depended computation time 
\begin{equation}\label{eq:computationtime}
\OO\left(n+\sum_{k=1}^{\hat{K}-1}{(R_k-L_k+1)(R_{k+1}-L_{k+1}+1)}\right).
\end{equation}
\end{Lemma}
This can be bounded by $\OO(n^2)$ in the worst case, but the computation time is in many cases much smaller. In particular, if the signal to noise ratios are large enough such that the change-points are easy to detect, i.e. $R_k-L_k$ is small. This is for instance the case for a fixed signal, where $R_k-L_k$ stays more or less constant. More precisely, by combining \eqref{eq:computationtime} with equation \eqref{eq:confidenceinterval} we see that with probability tending to one the computation time of $\HSMUCE$ is even linear, if $\alpha_n \to 0$, but $n^{-\frac{1}{2}}\log((\alpha_n\beta_{k_n,n})^ {-1})\to 0$. In comparison to the computation time of $\SMUCE$, see \citep[(4.3)]{Sieling14}, which is dominated by the term 
\begin{equation*}
\OO\left(\sum_{k=1}^{\hat{K}-1}{(R_k-R_{k-1})(R_{k+1}-R_{k})}\right),
\end{equation*}
we see that the computation time is further reduced. In particular, if no change-point is present the computation time is $\OO(n)$ instead of $\OO(n^2)$. The computation time is also $\OO(n)$ if the number of change-points increases linear in the number of observations and the change-points are evenly enough distributed.\\
In the following we examine the computation time empirically in a similar simulation study as in \citep{DiscussionMaidstonePickering}. More precisely, we generate data with varying number of observations $n$ and equidistant change-points. Thereby, we consider $K = 10$, $K = \sqrt{n}$ and $K = n/100$. In all scenarios we choose the values of the mean and the standard deviation function randomly like in Section \ref{sec:simulations}, once again with $C=200$. All simulations are repeated $100$ times and terminated after ten seconds. The simulations were performed on a single core system with $1.8$ GHz and $8$ GB RAM in a $64$-bit OS.\\
We fix the significance level $\alpha = 0.1$ as well as the weights $\beta_1=\cdots=\beta_\sn=1/\sn$ and compare $\HSMUCE$ with $\CBS$, $\LOOVF$ and $\cumSeg$. Note, that we restore the Monte-Carlo simulations at the first use to reduce further loading times, here we only take the already restored simulations into account. Furthermore, we set for $\cumSeg$ the maximal number of change-points $k = \max(2K,10)$, since for the default parameter $k=\min(30,n/10)$ the program requires manual increase of $k$ for many simulations runs. Note, that the choice above already incorporates prior knowledge about the true signal. We stress (not displayed) that the computation time (and the required memory space) increases severely in the parameter $k$.

\setcounter{figure}{7}

\begin{figure}[!htp]
\centering
\begin{subfigure}{\textwidth}
\includegraphics[width=0.98\textwidth]{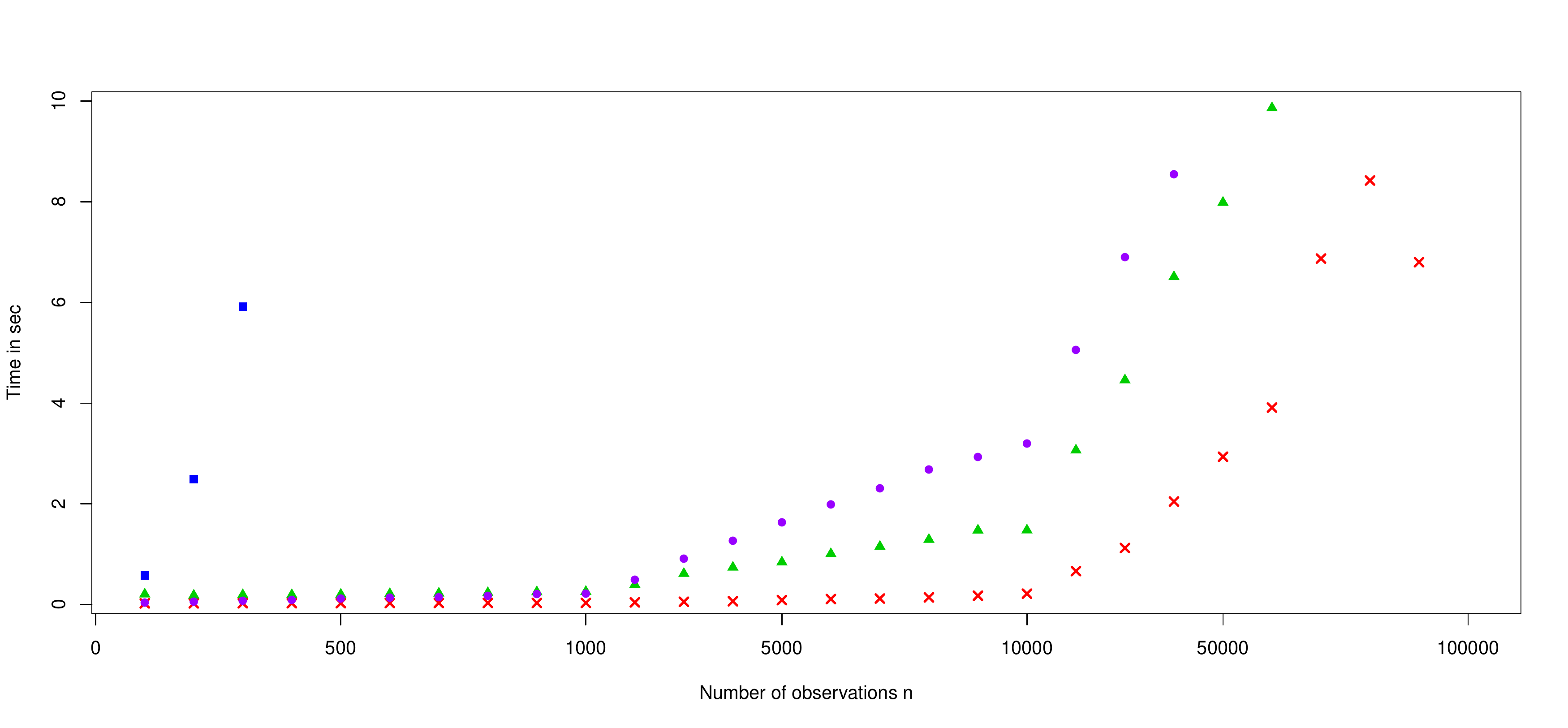}
\subcaption{$K = 10$.}
\end{subfigure}
\begin{subfigure}{\textwidth}
\includegraphics[width=0.98\textwidth]{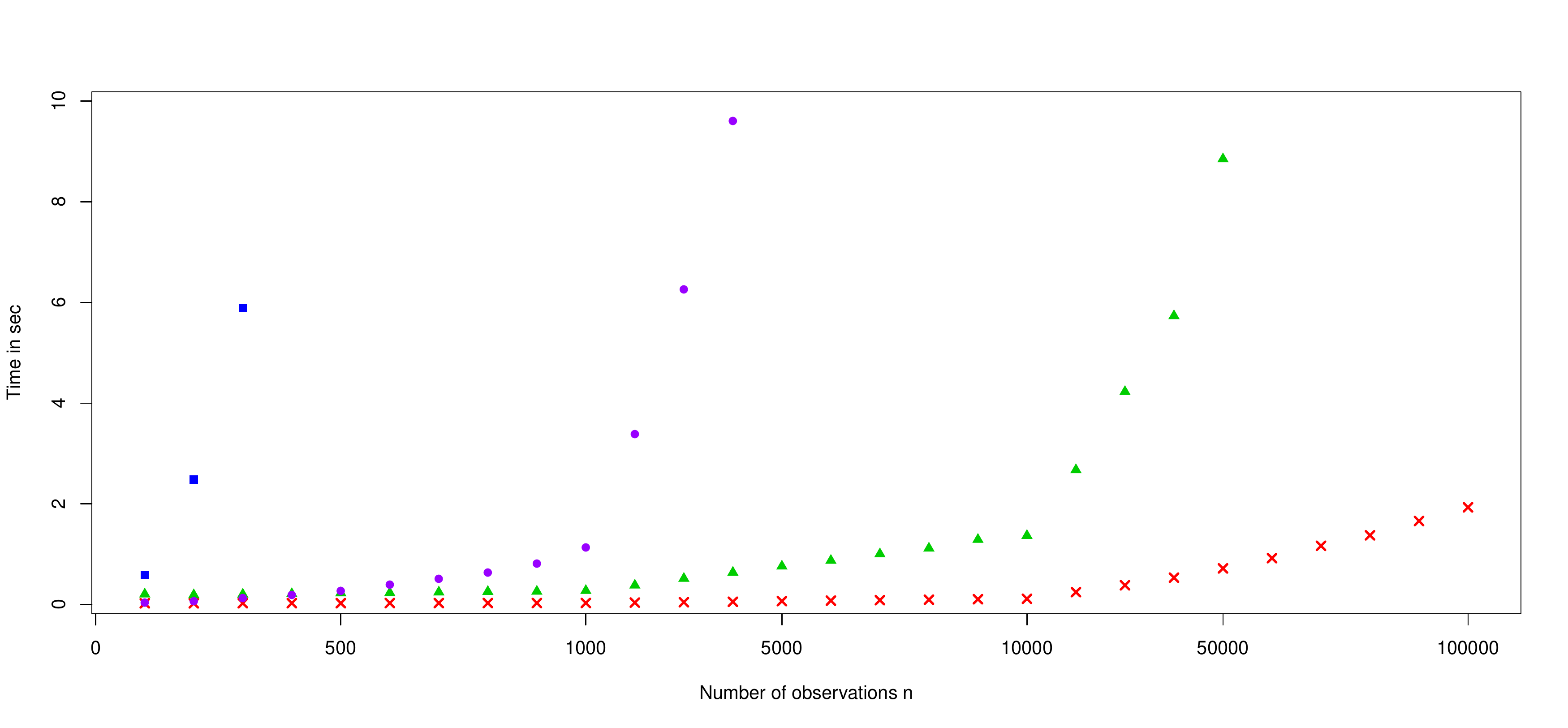}
\subcaption{$K = \sqrt{n}$.}
\end{subfigure}
\begin{subfigure}{\textwidth}
\includegraphics[width=0.98\textwidth]{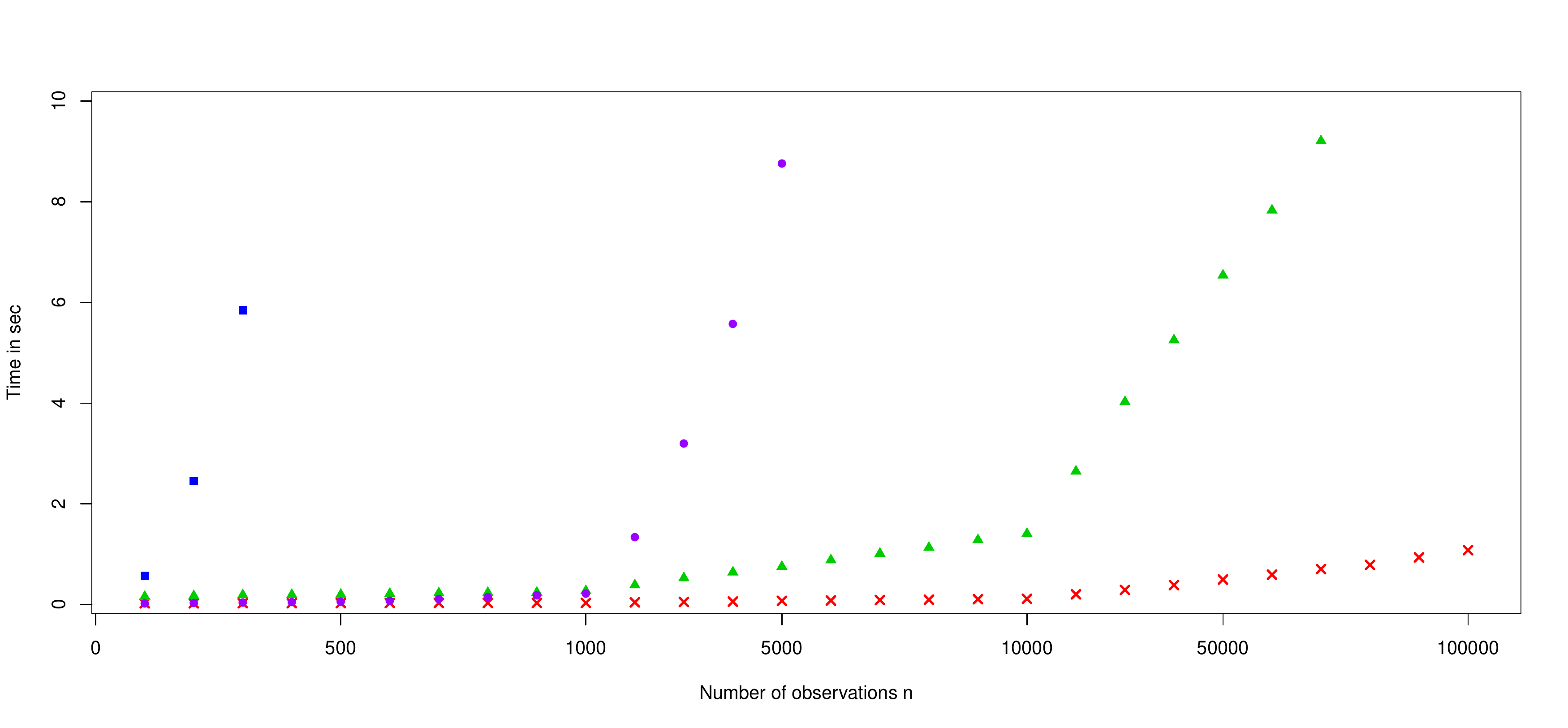}
\subcaption{$K = n/100$.}
\end{subfigure}
\caption{Mean computation time of $\HSMUCE$ (red crosses), $\CBS$ (green triangles), $\cumSeg$ (purple circles) and $\LOOVF$ (blue squares) for different number of observations $n$ and different number of change-points $K$. Note that for purposes of visualization the x-axis is displayed non-equidistantly.}
\label{fig:simulationtime}
\end{figure}

From Figure \ref{fig:simulationtime} we draw that $\HSMUCE$ is much faster than the other methods, in particular if the number of change-points increases. For $K=n/100$ the computation time increases almost linearly in the number of observations. For example, when $n=10^7$ it is still less than a minute. The second shortest computation time has $\CBS$ for larger numbers of observations, whereas $\cumSeg$ is superior for smaller numbers of observations. The computation time of $\CBS$ for $n=10^5$ observations is still less than a minute in all scenarios, whereas $\cumSeg$ has a similar computation time for $K=10$, but lasts several minutes in the other cases. Lastly, $\LOOVF$ exceeds ten seconds already for $n=400$ observations and is always found to be the slowest method.

\section{Additional Figures and Tables}\label{sec:add}
In this section we collect additional figures and tables.
 
\subsection{Simulations}
We start with estimates by $\CBS$, $\cumSeg$ and $\LOOVF$ for the data from Figure \ref{fig:example1000}.
\begin{figure}[!htp]
\centering
\begin{subfigure}{\textwidth}
\includegraphics[width=0.98\textwidth]{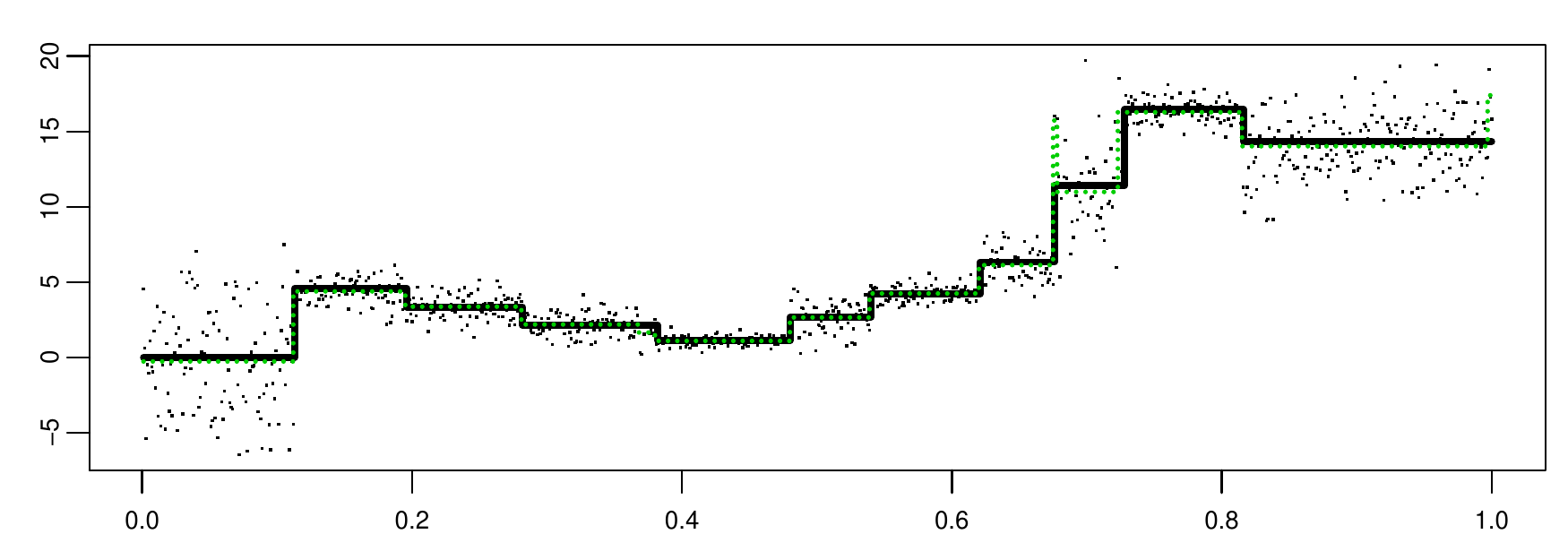}
\subcaption{$\CBS$.}
\end{subfigure}
\begin{subfigure}{\textwidth}
\includegraphics[width=0.98\textwidth]{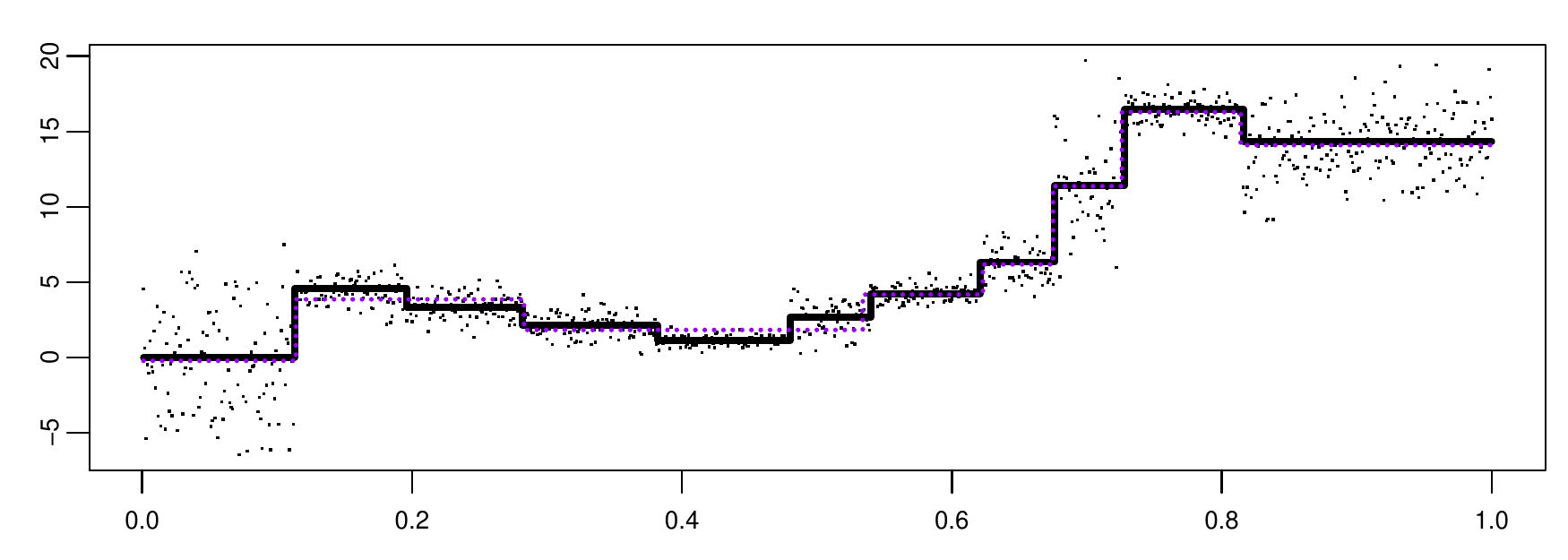}
\subcaption{$\cumSeg$.}
\end{subfigure}
\begin{subfigure}{\textwidth}
\includegraphics[width=0.98\textwidth]{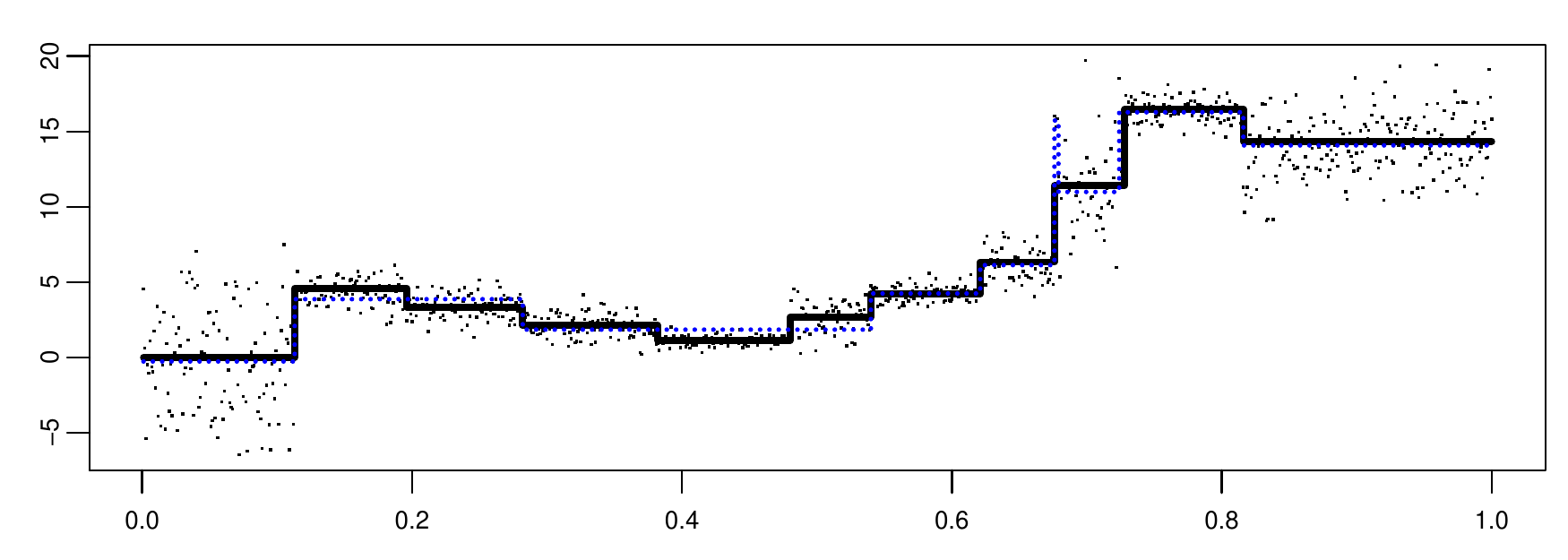}
\subcaption{$\LOOVF$.}
\end{subfigure}
\caption{Observations (black points) and true signal (black line) together with estimates by $\CBS$, $\cumSeg$ and $\LOOVF$ for the data from Figure \ref{fig:example1000}. All parameters are chosen as described in Section \ref{sec:simulations}.}
\label{fig:example1000other}
\end{figure}

The following three tables collect the results of the simulations in Section \ref{sec:simresults}. Recall the random pair $(\mu_R,\sigma_R^2)\in\mathcal{S}$ (all random variables are independent from each other):
\begin{enumerate}
\item[(a)] We fix the number of observations $n$, the number of change-points $K$, a constant $C$ and a minimum value for the smallest scale $\lambda_{\min}$.
\item[(b)] We draw the locations of the change-points $\tau_0:=0<\tau_1< \cdots <\tau_K<1=:\tau_{K+1}$ uniformly distributed with the restriction that $\lambda:=\min_{k=0,\ldots,K}{\vert \tau_{k+1}-\tau_{k}\vert} \geq \lambda_{\min}$.
\item[(c)] We choose the function values $\valsd_0,\ldots,\valsd_K$ of the standard deviation function $\sigma_R$ by $\valsd_k:=2^{U_k}$, where $U_0,\ldots,U_K$ are uniform distributed on $[-2,2]$.
\item[(d)] We determine the function values $\valmu_0,\ldots,\valmu_K$ of the signal $\mu_R$ such that
\[\vert \valmu_k - \valmu_{k-1}\vert = \sqrt{\frac{C}{n}\min\left(\frac{\tau_{k+1}-\tau_{k}}{\valsd_k^2}, \frac{\tau_k-\tau_{k-1}}{\valsd_{k-1}^2}\right)^{-1}}\ \forall\ k = 1,\ldots,K.\]
Thereby, we start with $\valmu_0 = 0$ and choose randomly with probability $1/2$ whether the expectation increases or decreases. 
\end{enumerate}
All simulations are repeated $10\,000$ times.

\begin{table}[!ht]
\scriptsize
\centering
\begin{tabular}
{l||l||c
ccc||c|cc|cc}
Setting & Method & -1 & 0 & +1 & $\geq +2$ & $\vert\hat{K}-K\vert$ & $\operatorname{FPSLE}$ & $\operatorname{FNSLE}$ & $\operatorname{MISE}$ & $\operatorname{MIAE}$\\
\hline
$\sigma_0=0.5$, & $\operatorname{HS}(0.1)$ & 0.000 & 0.995 & 0.004 & 0.000 & 0.005 & 0.82 & 0.74 & 0.0119 & 0.0644\\
$\sigma_1=0.5$, & $\operatorname{HS}(0.3)$ & 0.000 & 0.975 & 0.025 & 0.000 & 0.025 & 1.38 & 0.95 & 0.0129 & 0.0672\\
& $\operatorname{HS}(0.5)$ & 0.000 & 0.929 & 0.070 & 0.001 & 0.072 & 2.67 & 1.40 & 0.0144 & 0.0706\\
& $\CBS$ & 0.000 & 0.949 & 0.036 & 0.015 & 0.066 & 2.31 & 0.94 & 0.0128 & 0.0660\\
& $\cumSeg$ & 0.000 & 0.995 & 0.005 & 0.000 & 0.005 & 1.37 & 1.28 & 0.0172 & 0.0707\\
& $\LOOVF$ & 0.000 & 0.774 & 0.142 & 0.084 & 0.378 & 10.10 & 2.42 & 0.1402 & 0.2897\\
\hline
$\sigma_0=0.5$, & $\operatorname{HS}(0.1)$ & 0.112 & 0.886 & 0.002 & 0.000 & 0.114 & 3.99 & 6.77 & 0.0543 & 0.1405\\
$\sigma_1=1$, & $\operatorname{HS}(0.3)$ & 0.020 & 0.961 & 0.019 & 0.000 & 0.039 & 2.38 & 2.56 & 0.0321 & 0.1086\\
& $\operatorname{HS}(0.5)$ & 0.005 & 0.940 & 0.054 & 0.001 & 0.061 & 3.12 & 2.30 & 0.0314 & 0.1090\\
& $\CBS$ & 0.042 & 0.873 & 0.068 & 0.017 & 0.147 & 5.87 & 4.93 & 0.0496 & 0.1315\\
& $\cumSeg$ & 0.008 & 0.969 & 0.021 & 0.003 & 0.034 & 3.09 & 2.78 & 0.0375 & 0.1126\\
& $\LOOVF$ & 0.006 & 0.791 & 0.112 & 0.091 & 0.373 & 11.30 & 3.90 & 0.1720 & 0.3004\\
\hline
$\sigma_0=0.5$, & $\operatorname{HS}(0.1)$ & 0.484 & 0.515 & 0.001 & 0.000 & 0.485 & 12.77 & 24.89 & 0.1736 & 0.3110\\
$\sigma_1=1.5$, & $\operatorname{HS}(0.3)$ & 0.209 & 0.778 & 0.012 & 0.000 & 0.222 & 6.81 & 11.92 & 0.1025 & 0.2075\\
& $\operatorname{HS}(0.5)$ & 0.089 & 0.872 & 0.039 & 0.000 & 0.129 & 4.92 & 6.54 & 0.0725 & 0.1690\\
& $\CBS$ & 0.417 & 0.454 & 0.105 & 0.024 & 0.577 & 17.63 & 25.40 & 0.1845 & 0.3385\\
& $\cumSeg$ & 0.231 & 0.731 & 0.032 & 0.006 & 0.276 & 9.60 & 14.62 & 0.1149 & 0.2317\\
& $\LOOVF$ & 0.135 & 0.683 & 0.098 & 0.085 & 0.490 & 15.78 & 11.92 & 0.2307 & 0.3322\\
\hline
$\sigma_0=1$, & $\operatorname{HS}(0.1)$ & 0.453 & 0.547 & 0.001 & 0.000 & 0.453 & 13.49 & 24.97 & 0.1514 & 0.3140\\
$\sigma_1=1$, & $\operatorname{HS}(0.3)$ & 0.171 & 0.818 & 0.011 & 0.000 & 0.182 & 8.11 & 12.53 & 0.0942 & 0.2170\\
& $\operatorname{HS}(0.5)$ & 0.062 & 0.900 & 0.038 & 0.000 & 0.101 & 6.75 & 7.99 & 0.0745 & 0.1847\\
& $\CBS$ & 0.156 & 0.744 & 0.091 & 0.008 & 0.265 & 9.51 & 11.41 & 0.0943 & 0.2127\\
& $\cumSeg$ & 0.120 & 0.876 & 0.004 & 0.000 & 0.124 & 5.93 & 8.88 & 0.0748 & 0.1839\\
& $\LOOVF$ & 0.039 & 0.749 & 0.132 & 0.081 & 0.405 & 13.29 & 6.87 & 0.1947 & 0.3472\\
\hline
$\sigma_0=1$, & $\operatorname{HS}(0.1)$ & 0.727 & 0.272 & 0.000 & 0.000 & 0.728 & 19.44 & 37.71 & 0.2237 & 0.4244\\
$\sigma_1=1.5$, & $\operatorname{HS}(0.3)$ & 0.410 & 0.584 & 0.006 & 0.000 & 0.416 & 13.35 & 23.77 & 0.1644 & 0.3256\\
& $\operatorname{HS}(0.5)$ & 0.218 & 0.753 & 0.028 & 0.000 & 0.247 & 10.25 & 15.64 & 0.1283 & 0.2669\\
& $\CBS$ & 0.491 & 0.406 & 0.096 & 0.008 & 0.604 & 18.00 & 28.42 & 0.2013 & 0.3741\\
& $\cumSeg$ & 0.409 & 0.580 & 0.010 & 0.000 & 0.420 & 12.91 & 22.99 & 0.1571 & 0.3155\\
& $\LOOVF$ & 0.184 & 0.638 & 0.105 & 0.072 & 0.501 & 16.55 & 14.92 & 0.2410 & 0.3626\\
\hline 
$\sigma_0=1.5$, & $\operatorname{HS}(0.1)$ & 0.844 & 0.156 & 0.000 & 0.000 & 0.844 & 22.41 & 43.65 & 0.2581 & 0.4713\\
$\sigma_1=1.5$, & $\operatorname{HS}(0.3)$ & 0.574 & 0.423 & 0.003 & 0.000 & 0.577 & 18.21 & 33.12 & 0.2219 & 0.4101\\
& $\operatorname{HS}(0.5)$ & 0.352 & 0.629 & 0.018 & 0.000 & 0.371 & 15.47 & 25.01 & 0.1915 & 0.3582\\
& $\CBS$ & 0.659 & 0.258 & 0.079 & 0.003 & 0.746 & 20.73 & 35.81 & 0.2449 & 0.4379\\
& $\cumSeg$ & 0.629 & 0.369 & 0.002 & 0.000 & 0.631 & 17.56 & 33.32 & 0.2147 & 0.4067\\
& $\LOOVF$ & 0.297 & 0.534 & 0.104 & 0.066 & 0.589 & 19.34 & 21.26 & 0.2715 & 0.4046\\
\hline
\end{tabular}
\caption{Simulations with a single change (fixed signal and variances): $n=100$ observations and a single change at $0.5$, from $0$ to $1$ for different standard deviations changing from $\sigma_0$ to $\sigma_1$ at $0.5$, too. Columns from left to right: setting, method, proportions of $\hat{K}-K$ and averages of the corresponding error criteria. $\operatorname{HS}(\alpha)$ denotes $\HSMUCE$ at significance level $\alpha$.}
\label{tab:fixed}
\end{table}

\begin{table}[!htp]
\scriptsize
\centering
\begin{tabular}
{l||l||cc
ccc||c|cc|cc}
Setting & Method & $\leq -2$ & -1 & 0 & +1 & $\geq +2$ & $\vert\hat{K}-K\vert$ & $\operatorname{FPSLE}$ & $\operatorname{FNSLE}$ & $\operatorname{MISE}$ & $\operatorname{MIAE}$\\
\hline
n = 1000, & $\operatorname{HS}(0.1)$ & - & - & 0.965 & 0.035 & 0.000 & 0.035 & 17.75 & 4.73 & 0.0035 & 0.0365\\
K = 0, & $\operatorname{HS}(0.3)$ & - & - & 0.867 & 0.128 & 0.005 & 0.138 & 68.95 & 18.42 & 0.0045 & 0.0401\\
$\mu=\mu_R\equiv 0$, & $\operatorname{HS}(0.5)$ & - & - & 0.719 & 0.256 & 0.025 & 0.307 & 153.45 & 41.25 & 0.0061 & 0.0454\\
$\sigma=\sigma_R$& $\operatorname{S}(0.1)$ & - & - & 0.965 & 0.034 & 0.001 & 0.036 & 17.90 & 5.03 & 0.0039 & 0.0371\\
$\equiv \operatorname{const}$ & $\operatorname{S}(0.3)$ & - & - & 0.832 & 0.160 & 0.008 & 0.177 & 88.45 & 24.80 & 0.0059 & 0.0435\\
& $\operatorname{S}(0.5)$ & - & - & 0.667 & 0.298 & 0.035 & 0.370 & 184.90 & 50.94 & 0.0082 & 0.0499\\
& $\CBS$ & - & - & 0.991 & 0.000 & 0.009 & 0.018 & 8.90 & 1.26 & 0.0037 & 0.0351\\
& $\cumSeg$ & - & - & 0.999 & 0.001 & 0.000 & 0.001 & 0.30 & 0.06 & 0.0029 & 0.0345\\
\hline
n = 1000, & $\operatorname{HS}(0.1)$ & 0.010 & 0.174 & 0.802 & 0.014 & 0.000 & 0.208 & 26.32 & 72.66 & 0.0132 & 0.0613\\
K = 2, & $\operatorname{HS}(0.3)$ & 0.004 & 0.108 & 0.819 & 0.067 & 0.002 & 0.187 & 38.10 & 52.90 & 0.0114 & 0.0571\\
$\lambda_{\min} = 30$, & $\operatorname{HS}(0.5)$ & 0.002 & 0.070 & 0.768 & 0.150 & 0.010 & 0.244 & 64.14 & 48.50 & 0.0111 & 0.0573\\
$\mu=\mu_R$, & $\operatorname{S}(0.1)$ & 0.003 & 0.074 & 0.912 & 0.011 & 0.000 & 0.092 & 16.96 & 34.03 & 0.0092 & 0.0513\\
$\sigma\equiv 1$ & $\operatorname{S}(0.3)$ & 0.001 & 0.040 & 0.892 & 0.065 & 0.002 & 0.112 & 32.24 & 27.39 & 0.0090 & 0.0513\\
& $\operatorname{S}(0.5)$ & 0.001 & 0.025 & 0.806 & 0.155 & 0.013 & 0.209 & 63.30 & 32.33 & 0.0095 & 0.0536\\
& $\CBS$ & 0.005 & 0.060 & 0.821 & 0.082 & 0.033 & 0.221 & 37.55 & 37.57 & 0.0111 & 0.0527\\
& $\cumSeg$ & 0.025 & 0.116 & 0.749 & 0.099 & 0.011 & 0.289 & 65.32 & 82.63 & 0.0364 & 0.0738\\
  \hline
n = 1000, & $\operatorname{HS}(0.1)$ & 0.009 & 0.160 & 0.815 & 0.015 & 0.000 & 0.194 & 27.14 & 68.91 & 0.0127 & 0.0611\\
K = 2, & $\operatorname{HS}(0.3)$ & 0.004 & 0.098 & 0.829 & 0.067 & 0.001 & 0.176 & 37.77 & 49.63 & 0.0111 & 0.0572\\
$\lambda_{\min} = 50$, & $\operatorname{HS}(0.5)$ & 0.002 & 0.063 & 0.774 & 0.152 & 0.009 & 0.237 & 63.46 & 46.06 & 0.0109 & 0.0573\\
$\mu=\mu_R$, & $\operatorname{S}(0.1)$ & 0.003 & 0.068 & 0.919 & 0.009 & 0.000 & 0.084 & 16.82 & 31.94 & 0.0091 & 0.0515\\
$\sigma\equiv 1$ & $\operatorname{S}(0.3)$ & 0.001 & 0.035 & 0.899 & 0.063 & 0.002 & 0.104 & 31.19 & 25.81 & 0.0090 & 0.0515\\
& $\operatorname{S}(0.5)$ & 0.001 & 0.020 & 0.819 & 0.147 & 0.013 & 0.195 & 59.86 & 30.23 & 0.0095 & 0.0537\\
& $\CBS$ & 0.005 & 0.058 & 0.824 & 0.083 & 0.031 & 0.215 & 37.50 & 36.27 & 0.0112 & 0.0532\\
& $\cumSeg$ & 0.023 & 0.110 & 0.769 & 0.090 & 0.008 & 0.262 & 59.74 & 79.25 & 0.0336 & 0.0741\\  
\hline
n = 1000, & $\operatorname{HS}(0.1)$ & 0.508 & 0.330 & 0.161 & 0.001 & 0.000 & 1.634 & 54.37 & 172.66 & 0.1112 & 0.1842\\
K = 10, & $\operatorname{HS}(0.3)$ & 0.354 & 0.377 & 0.263 & 0.006 & 0.000 & 1.233 & 44.53 & 127.81 & 0.0817 & 0.1561\\
$\lambda_{\min} = 30$, & $\operatorname{HS}(0.5)$ & 0.253 & 0.384 & 0.346 & 0.017 & 0.000 & 0.987 & 40.88 & 102.88 & 0.0679 & 0.1419\\
$\mu=\mu_R$, & $\operatorname{S}(0.1)$ & 0.163 & 0.352 & 0.485 & 0.001 & 0.000 & 0.721 & 29.14 & 77.49 & 0.0424 & 0.1193\\
$\sigma\equiv 1$ & $\operatorname{S}(0.3)$ & 0.093 & 0.301 & 0.598 & 0.007 & 0.000 & 0.513 & 24.23 & 56.17 & 0.0366 & 0.1099\\
& $\operatorname{S}(0.5)$ & 0.062 & 0.258 & 0.657 & 0.022 & 0.001 & 0.415 & 23.34 & 46.37 & 0.0342 & 0.1060\\
& $\CBS$ & 0.033 & 0.129 & 0.531 & 0.204 & 0.102 & 0.644 & 42.69 & 45.08 & 0.0417 & 0.1078\\
& $\cumSeg$ & 0.163 & 0.216 & 0.403 & 0.165 & 0.053 & 0.904 & 65.16 & 105.59 & 0.1107 & 0.1492\\
\hline
n = 1000, & $\operatorname{HS}(0.1)$ & 0.445 & 0.356 & 0.198 & 0.001 & 0.000 & 1.474 & 59.32 & 162.03 & 0.0913 & 0.1801\\
K = 10, & $\operatorname{HS}(0.3)$ & 0.303 & 0.384 & 0.307 & 0.005 & 0.000 & 1.104 & 47.34 & 120.10 & 0.0682 & 0.1532\\
$\lambda_{\min} = 50$, & $\operatorname{HS}(0.5)$ & 0.213 & 0.379 & 0.390 & 0.018 & 0.001 & 0.881 & 41.98 & 96.70 & 0.0577 & 0.1398\\
$\mu=\mu_R$, & $\operatorname{S}(0.1)$ & 0.155 & 0.351 & 0.494 & 0.000 & 0.000 & 0.697 & 32.51 & 77.29 & 0.0426 & 0.1235\\
$\sigma\equiv 1$ & $\operatorname{S}(0.3)$ & 0.085 & 0.299 & 0.612 & 0.004 & 0.000 & 0.485 & 26.14 & 55.78 & 0.0368 & 0.1131\\
& $\operatorname{S}(0.5)$ & 0.054 & 0.252 & 0.680 & 0.014 & 0.000 & 0.381 & 23.81 & 45.39 & 0.0344 & 0.1086\\
& $\CBS$ & 0.027 & 0.135 & 0.524 & 0.203 & 0.111 & 0.653 & 45.64 & 44.88 & 0.0425 & 0.1116\\
& $\cumSeg$ & 0.165 & 0.217 & 0.389 & 0.179 & 0.050 & 0.904 & 63.73 & 104.37 & 0.1037 & 0.1522\\
 \hline     
\end{tabular}
\caption{Simulations with constant variance and $C=200$. Columns from left to right: setting, method, proportions of $\hat{K}-K$ and averages of the corresponding error criteria. $\operatorname{HS}(\alpha)$ and $\operatorname{S}(\alpha)$ denote $\HSMUCE$ and $\SMUCE$ at significance level $\alpha$, respectively.}
\label{tab:constvar}
\end{table}

\begin{table}[!htp]
\scriptsize
\centering
\begin{tabular}
{l||l||cc
ccc||c|cc|cc}
Setting & Method & $\leq -2$ & -1 & 0 & +1 & $\geq +2$ & $\vert\hat{K}-K\vert$ & $\operatorname{FPSLE}$ & $\operatorname{FNSLE}$ & $\operatorname{MISE}$ & $\operatorname{MIAE}$\\
\hline
n = 100, & $\operatorname{HS}(0.1)$ & 0.000 & 0.125 & 0.873 & 0.002 & 0.000 & 0.128 & 1.51 & 4.07 & 0.8182 & 0.3308\\
K = 2, & $\operatorname{HS}(0.3)$ & 0.000 & 0.042 & 0.945 & 0.013 & 0.000 & 0.055 & 1.04 & 1.70 & 0.4217 & 0.2482\\
$\lambda_{\min} = 15$, & $\operatorname{HS}(0.5)$ & 0.000 & 0.016 & 0.940 & 0.043 & 0.000 & 0.060 & 1.63 & 1.26 & 0.2776 & 0.2291\\
$\mu=\mu_R$, & $\CBS$ & 0.000 & 0.001 & 0.925 & 0.058 & 0.016 & 0.092 & 2.03 & 0.79 & 0.2220 & 0.2143\\
$\sigma=\sigma_R$ & $\cumSeg$ & 0.000 & 0.066 & 0.720 & 0.167 & 0.047 & 0.343 & 6.50 & 4.39 & 0.4898 & 0.3053\\
 & $\LOOVF$ & 0.000 & 0.031 & 0.700 & 0.163 & 0.106 & 0.683 & 12.83 & 3.36 & 0.3167 & 0.2639\\
 \hline
n = 100, & $\operatorname{HS}(0.1)$ & 0.608 & 0.364 & 0.028 & 0.000 & 0.000 & 1.610 & 13.51 & 32.33 & 9.5104 & 1.8626\\
K = 5, & $\operatorname{HS}(0.3)$ & 0.212 & 0.577 & 0.211 & 0.000 & 0.000 & 1.003 & 8.63 & 19.80 & 6.5362 & 1.3263\\
$\lambda_{\min} = 15$, & $\operatorname{HS}(0.5)$ & 0.061 & 0.466 & 0.473 & 0.001 & 0.000 & 0.588 & 5.27 & 11.65 & 3.9992 & 0.9047\\
$\mu=\mu_R$, & $\CBS$ & 0.001 & 0.008 & 0.884 & 0.089 & 0.018 & 0.137 & 1.65 & 1.02 & 0.4539 & 0.3130\\
$\sigma=\sigma_R$ & $\cumSeg$ & 0.098 & 0.230 & 0.544 & 0.117 & 0.012 & 0.588 & 6.93 & 12.13 & 1.2454 & 0.5441\\
 & $\LOOVF$ & 0.031 & 0.112 & 0.520 & 0.152 & 0.184 & 1.648 & 14.61 & 6.92 & 0.5887 & 0.4042\\
 \hline
n = 1000, & $\operatorname{HS}(0.1)$ & 0.000 & 0.007 & 0.974 & 0.018 & 0.000 & 0.026 & 8.42 & 5.83 & 0.0195 & 0.0617\\
K = 2, & $\operatorname{HS}(0.3)$ & 0.000 & 0.001 & 0.921 & 0.075 & 0.002 & 0.080 & 24.72 & 9.57 & 0.0193 & 0.0636\\
$\lambda_{\min} = 30$, & $\operatorname{HS}(0.5)$ & 0.000 & 0.000 & 0.827 & 0.162 & 0.012 & 0.185 & 53.23 & 17.09 & 0.0204 & 0.0668\\
$\mu=\mu_R$, & $\CBS$ & 0.005 & 0.019 & 0.774 & 0.146 & 0.056 & 0.298 & 52.95 & 21.17 & 0.0347 & 0.0711\\
$\sigma=\sigma_R$ & $\cumSeg$ & 0.022 & 0.161 & 0.683 & 0.103 & 0.030 & 0.387 & 64.04 & 92.66 & 0.0765 & 0.1112\\
 \hline
n = 1000, & $\operatorname{HS}(0.1)$ & 0.000 & 0.002 & 0.982 & 0.017 & 0.000 & 0.018 & 7.25 & 4.35 & 0.0182 & 0.0630\\
K = 2, & $\operatorname{HS}(0.3)$ & 0.000 & 0.000 & 0.926 & 0.071 & 0.002 & 0.076 & 22.64 & 8.49 & 0.0196 & 0.0657\\
$\lambda_{\min} = 50$, & $\operatorname{HS}(0.5)$ & 0.000 & 0.000 & 0.830 & 0.160 & 0.010 & 0.181 & 50.02 & 16.22 & 0.0214 & 0.0692\\
$\mu=\mu_R$, & $\CBS$ & 0.003 & 0.011 & 0.776 & 0.153 & 0.057 & 0.296 & 53.69 & 15.85 & 0.0355 & 0.0730\\
$\sigma=\sigma_R$ & $\cumSeg$ & 0.016 & 0.155 & 0.699 & 0.098 & 0.031 & 0.370 & 60.63 & 84.69 & 0.0739 & 0.1132\\ 
 \hline
n = 1000, & $\operatorname{HS}(0.1)$ & 0.123 & 0.429 & 0.446 & 0.002 & 0.000 & 0.686 & 22.83 & 55.06 & 0.4045 & 0.2402\\
K = 10, & $\operatorname{HS}(0.3)$ & 0.016 & 0.199 & 0.770 & 0.015 & 0.000 & 0.245 & 11.98 & 21.12 & 0.1863 & 0.1618\\
$\lambda_{\min} = 30$, & $\operatorname{HS}(0.5)$ & 0.002 & 0.088 & 0.863 & 0.045 & 0.001 & 0.140 & 11.84 & 12.71 & 0.1220 & 0.1404\\
$\mu=\mu_R$, & $\CBS$ & 0.002 & 0.008 & 0.463 & 0.316 & 0.211 & 0.843 & 47.26 & 15.20 & 0.1274 & 0.1435\\
$\sigma=\sigma_R$ & $\cumSeg$ & 0.439 & 0.243 & 0.187 & 0.085 & 0.046 & 1.674 & 94.91 & 228.44 & 0.3120 & 0.2806\\
 \hline
n = 1000, & $\operatorname{HS}(0.1)$ & 0.025 & 0.262 & 0.711 & 0.002 & 0.000 & 0.315 & 16.94 & 32.39 & 0.2102 & 0.1866\\
K = 10, & $\operatorname{HS}(0.3)$ & 0.002 & 0.058 & 0.925 & 0.015 & 0.000 & 0.076 & 8.46 & 10.58 & 0.1009 & 0.1372\\
$\lambda_{\min} = 50$, & $\operatorname{HS}(0.5)$ & 0.000 & 0.017 & 0.940 & 0.043 & 0.001 & 0.061 & 9.03 & 7.72 & 0.0860 & 0.1307\\
$\mu=\mu_R$, & $\CBS$ & 0.001 & 0.007 & 0.451 & 0.319 & 0.222 & 0.868 & 47.81 & 15.10 & 0.1293 & 0.1463\\
$\sigma=\sigma_R$ & $\cumSeg$ & 0.433 & 0.254 & 0.197 & 0.082 & 0.035 & 1.601 & 97.00 & 223.47 & 0.2771 & 0.2794\\
 \hline
n = 10000, & $\operatorname{HS}(0.1)$ & 0.000 & 0.004 & 0.983 & 0.013 & 0.000 & 0.017 & 50.65 & 30.94 & 0.0016 & 0.0183\\
K = 2, & $\operatorname{HS}(0.3)$ & 0.000 & 0.002 & 0.936 & 0.061 & 0.001 & 0.065 & 188.73 & 63.72 & 0.0016 & 0.0188\\
$\lambda_{\min} = 30$, & $\operatorname{HS}(0.5)$ & 0.000 & 0.001 & 0.865 & 0.128 & 0.006 & 0.142 & 407.41 & 125.46 & 0.0016 & 0.0197\\
$\mu=\mu_R$, & $\CBS$ & 0.012 & 0.036 & 0.532 & 0.200 & 0.220 & 0.886 & 1548.96 & 373.22 & 0.0057 & 0.0235\\
$\sigma=\sigma_R$ & $\cumSeg$ & 0.054 & 0.245 & 0.600 & 0.084 & 0.017 & 0.477 & 682.64 & 1457.08 & 0.0090 & 0.0379\\ 
 \hline
n = 10000, & $\operatorname{HS}(0.1)$ & 0.000 & 0.001 & 0.984 & 0.015 & 0.000 & 0.016 & 53.23 & 24.89 & 0.0014 & 0.0182\\
K = 2, & $\operatorname{HS}(0.3)$ & 0.000 & 0.000 & 0.941 & 0.057 & 0.002 & 0.060 & 181.06 & 59.83 & 0.0014 & 0.0188\\
$\lambda_{\min} = 50$, & $\operatorname{HS}(0.5)$ & 0.000 & 0.000 & 0.870 & 0.124 & 0.007 & 0.137 & 394.16 & 115.62 & 0.0016 & 0.0197\\
$\mu=\mu_R$, & $\CBS$ & 0.012 & 0.035 & 0.521 & 0.208 & 0.225 & 0.917 & 1601.54 & 366.42 & 0.0058 & 0.0238\\
$\sigma=\sigma_R$ & $\cumSeg$ & 0.052 & 0.241 & 0.603 & 0.087 & 0.016 & 0.473 & 673.81 & 1430.47 & 0.0084 & 0.0377\\ 
 \hline
n = 10000, & $\operatorname{HS}(0.1)$ & 0.023 & 0.231 & 0.741 & 0.005 & 0.000 & 0.282 & 58.42 & 165.72 & 0.0178 & 0.0431\\
K = 10, & $\operatorname{HS}(0.3)$ & 0.006 & 0.123 & 0.844 & 0.027 & 0.000 & 0.162 & 68.27 & 98.25 & 0.0122 & 0.0385\\
$\lambda_{\min} = 30$, & $\operatorname{HS}(0.5)$ & 0.003 & 0.079 & 0.854 & 0.064 & 0.002 & 0.151 & 108.19 & 87.63 & 0.0103 & 0.0377\\
$\mu=\mu_R$, & $\CBS$ & 0.024 & 0.043 & 0.180 & 0.222 & 0.531 & 2.088 & 1286.59 & 525.95 & 0.0198 & 0.0475\\
$\sigma=\sigma_R$ & $\cumSeg$ & 0.619 & 0.169 & 0.130 & 0.059 & 0.024 & 2.345 & 1000.55 & 3122.28 & 0.0433 & 0.0917\\ 
 \hline
n = 10000, & $\operatorname{HS}(0.1)$ & 0.009 & 0.165 & 0.819 & 0.007 & 0.000 & 0.190 & 59.11 & 124.05 & 0.0132 & 0.0418\\
K = 10, & $\operatorname{HS}(0.3)$ & 0.001 & 0.064 & 0.905 & 0.029 & 0.001 & 0.097 & 67.32 & 65.54 & 0.0089 & 0.0375\\
$\lambda_{\min} = 50$, & $\operatorname{HS}(0.5)$ & 0.000 & 0.029 & 0.900 & 0.067 & 0.003 & 0.102 & 103.42 & 60.04 & 0.0078 & 0.0368\\
$\mu=\mu_R$, & $\CBS$ & 0.019 & 0.034 & 0.162 & 0.228 & 0.557 & 2.203 & 1317.31 & 467.47 & 0.0198 & 0.0475\\
$\sigma=\sigma_R$ & $\cumSeg$ & 0.607 & 0.188 & 0.131 & 0.051 & 0.023 & 2.277 & 997.64 & 3105.88 & 0.0405 & 0.0925\\ 
 \hline
n = 10000, & $\operatorname{HS}(0.1)$ & 0.609 & 0.284 & 0.107 & 0.001 & 0.000 & 1.908 & 155.65 & 504.02 & 0.1016 & 0.1031\\
K = 25, & $\operatorname{HS}(0.3)$ & 0.278 & 0.399 & 0.318 & 0.006 & 0.000 & 1.044 & 94.53 & 263.30 & 0.0640 & 0.0789\\
$\lambda_{\min} = 30$, & $\operatorname{HS}(0.5)$ & 0.140 & 0.371 & 0.470 & 0.019 & 0.000 & 0.696 & 84.07 & 182.54 & 0.0483 & 0.0703\\
$\mu=\mu_R$, & $\CBS$ & 0.015 & 0.024 & 0.069 & 0.128 & 0.765 & 3.348 & 921.91 & 409.98 & 0.0411 & 0.0723\\
$\sigma=\sigma_R$ & $\cumSeg$ & 0.934 & 0.036 & 0.018 & 0.009 & 0.003 & 6.028 & 1043.82 & 3488.43 & 0.1159 & 0.1540\\  
 \hline
n = 10000, & $\operatorname{HS}(0.1)$ & 0.396 & 0.383 & 0.220 & 0.001 & 0.000 & 1.334 & 146.74 & 387.66 & 0.0699 & 0.0945\\
K = 25, & $\operatorname{HS}(0.3)$ & 0.103 & 0.359 & 0.528 & 0.010 & 0.000 & 0.591 & 85.33 & 175.03 & 0.0390 & 0.0715\\
$\lambda_{\min} = 50$, & $\operatorname{HS}(0.5)$ & 0.038 & 0.241 & 0.690 & 0.030 & 0.001 & 0.352 & 78.74 & 114.01 & 0.0291 & 0.0647\\
$\mu=\mu_R$, & $\CBS$ & 0.010 & 0.017 & 0.055 & 0.120 & 0.799 & 3.529 & 934.29 & 346.33 & 0.0405 & 0.0726\\
$\sigma=\sigma_R$ & $\cumSeg$ & 0.934 & 0.036 & 0.019 & 0.008 & 0.003 & 5.849 & 1053.35 & 3462.62 & 0.1022 & 0.1547\\ 
 \hline     
\end{tabular}
\caption{Simulations with heterogeneous errors and $C=200$. Columns from left to right: setting, method, proportions of $\hat{K}-K$ and averages of the corresponding error criteria. $\operatorname{HS}(\alpha)$ denotes $\HSMUCE$ at significance level $\alpha$.}
\label{tab:random10000}
\end{table}

\subsection{Prior information on scales}\label{sec:simprior}
To demonstrate the effect of incorporating prior knowledge about those scales where change-points are likely to happen we consider again the observations from Table \ref{tab:random10000} with $n=10\,000$, $K=10$ and $\lambda_{\min}=50$. To this end, we use the adapted weights, where we eliminate the smallest three scales $k=1,2,3$, since all constant segments contain at least $50$ observations and therefore these small scales are not needed for detection. Moreover, we choose $\tilde{\beta}_4=1/4$, $\tilde{\beta}_5=1/4$, $\tilde{\beta}_6=1/6$, $\tilde{\beta}_7=1/6$, $\tilde{\beta}_8=1/12$, $\tilde{\beta}_9=1/12$ in decreasing order, since change-points on smaller scales are more likely and harder to detect. For the same reasons we eliminate the four largest scales $k=10, 11, 12, 13$, too.

\begin{table}[!ht]
\scriptsize
\centering
\begin{tabular}
{l||cc
ccc||c|cc|cc}
Method & $\leq -2$ & -1 & 0 & +1 & $\geq +2$ & $\vert\hat{K}-K\vert$ & $\operatorname{FPSLE}$ & $\operatorname{FNSLE}$ & $\operatorname{MISE}$ & $\operatorname{MIAE}$\\
   \hline
$\operatorname{HS}(0.1)$ & 0.005 & 0.117 & 0.876 & 0.002 & 0.000 & 0.130 & 50.82 & 113.50 & 0.0107 & 0.0406\\
$\operatorname{HS}(0.3)$ & 0.000 & 0.032 & 0.952 & 0.016 & 0.000 & 0.049 & 48.39 & 49.84 & 0.0075 & 0.0368\\
$\operatorname{HS}(0.5)$ & 0.000 & 0.013 & 0.940 & 0.045 & 0.001 & 0.061 & 78.86 & 48.19 & 0.0072 & 0.0368\\
 \hline
\end{tabular}
\caption{$n=10\,000$ observations, $K=10$ change-points, $C=200$ and $\lambda_{\min}=50$ from Table \ref{tab:random10000}. Columns from left to right: setting, method, proportions of $\hat{K}-K$ and averages of the corresponding error criteria. $\operatorname{HS}(\alpha)$ denotes $\HSMUCE$ at significance level $\alpha$, but with weights $\tilde{\beta}_4,\ldots,\tilde{\beta}_{9}$.}
\label{tab:beta10000}
\end{table}

A comparison of Table \ref{tab:random10000} and \ref{tab:beta10000} shows that the modified weights increase the detection power of $\HSMUCE$ for all significance levels, so we encourage the user to adapt the weights if prior information on the scales where changes occur is available.\\ 

\subsection{Robustness}\label{sec:addrobustness}
Figure \ref{fig:std} shows the standard deviation functions in Table \ref{tab:robustvar} to examine robustness against variance changes on constant segments. We consider the sinus-shaped standard deviation $\sigma_1$ (continuous changes), the piecewise linear standard deviation $\sigma_2$ (continuous and abrupt changes at the same time) and the piecewise constant standard deviation $\sigma_3$ (abrupt changes). Moreover, we analyse in Table \ref{tab:robustnesstrend} robustness against small periodic trends in simulations similar to those in \citep{CBS04}. More precisely, we generate the random pairs $(\mu_R,\sigma_R^2)\in\mathcal{S}$ as in (a)-(d) described, but replace the signal $\mu_R$ by
\begin{equation*}
\begin{split}
&\mu_T(i/n)=\mu_R + b\sin(a\pi i)\\
\text{ and }&\mu_{T_\sigma}(i/n)=\mu_R + b\sigma_R(i/n)\sin(a\pi i)+b(\sigma_R(i/n)-\sigma_R((i-1)/n)\sin(a\pi i),\\
&i=1,\ldots,n,
\end{split}
\end{equation*}
respectively. The signal $\mu_T$ reflects the situation of a fixed periodic trend, whereas in $\mu_{T_\sigma}$ the trend is scaled by the local standard deviation. The last term corrects the size of changes such that still $\sigma_R$ determines the changes. We consider as in \citep{CBS04} long ($a=0.01$) and short ($a=0.025$) trends. Finally, Table \ref{tab:robustnonnormal} reports result of $t_3$ distributed errors.

\begin{figure}[!ht]
\centering
\begin{subfigure}{0.32\textwidth}
\includegraphics[width=\textwidth]{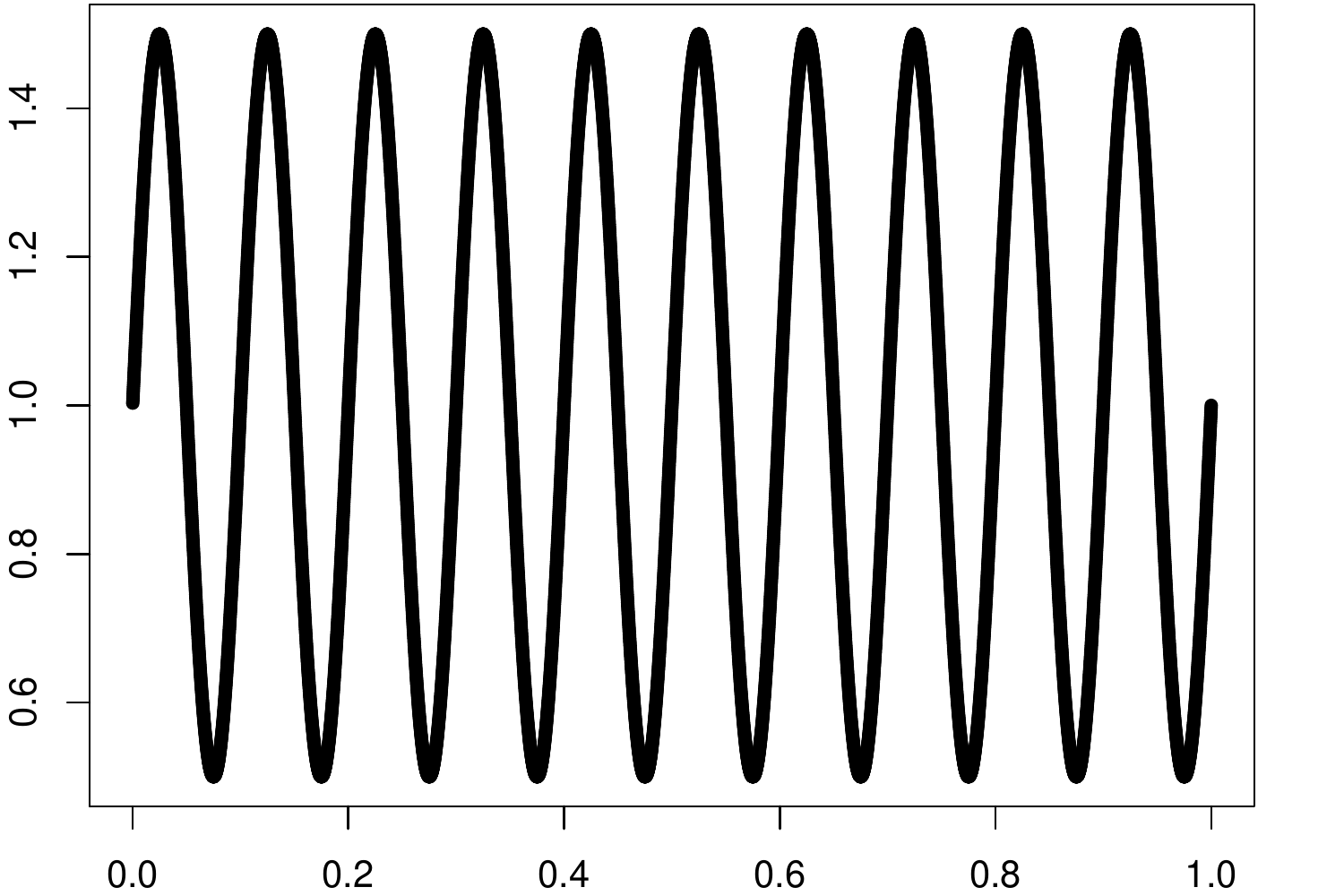}
\subcaption{}\label{subfig:sinus}
\end{subfigure}
\begin{subfigure}{0.32\textwidth}
\includegraphics[width=\textwidth]{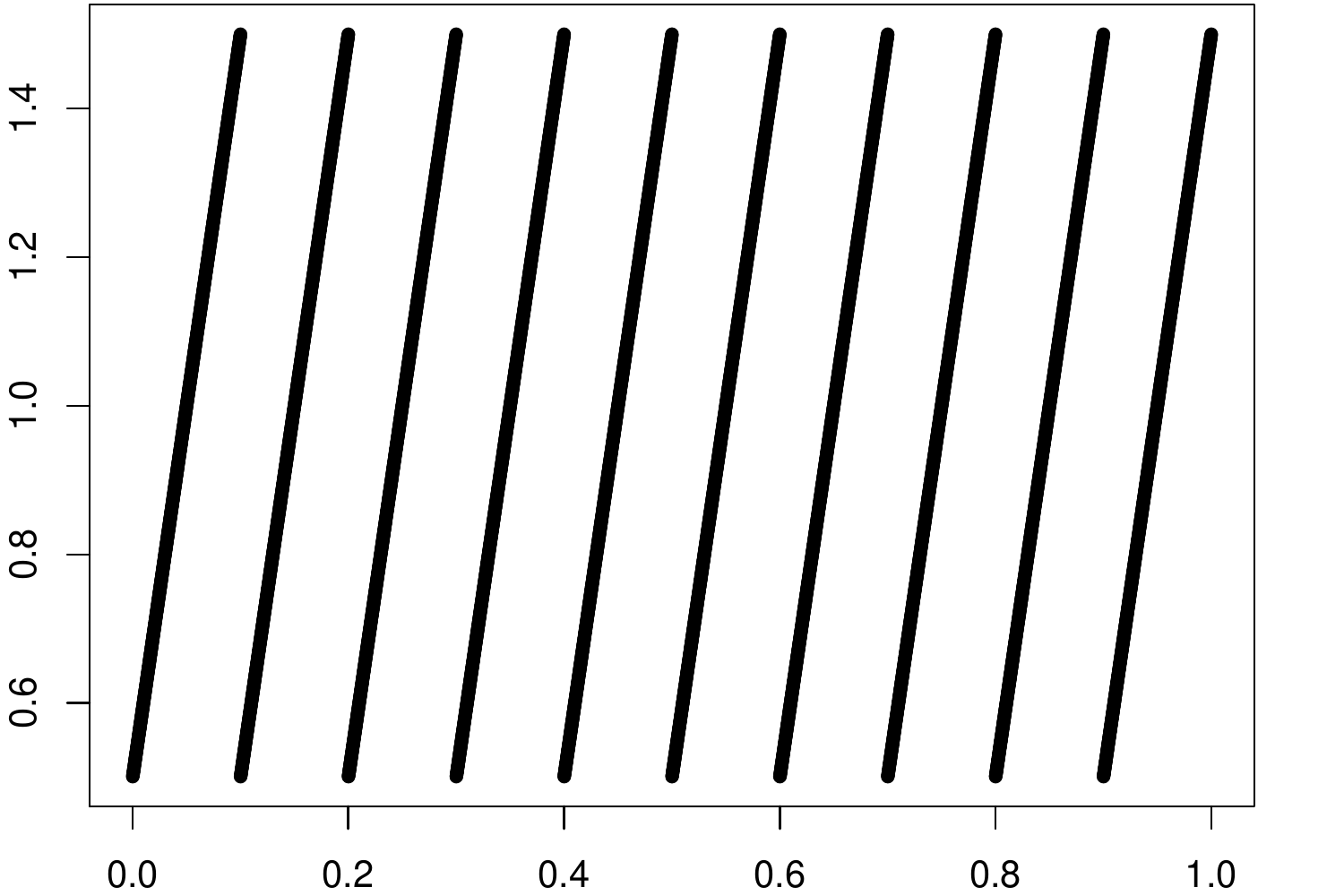}
\subcaption{}\label{subfig:lin}
\end{subfigure}
\begin{subfigure}{0.32\textwidth}
\includegraphics[width=\textwidth]{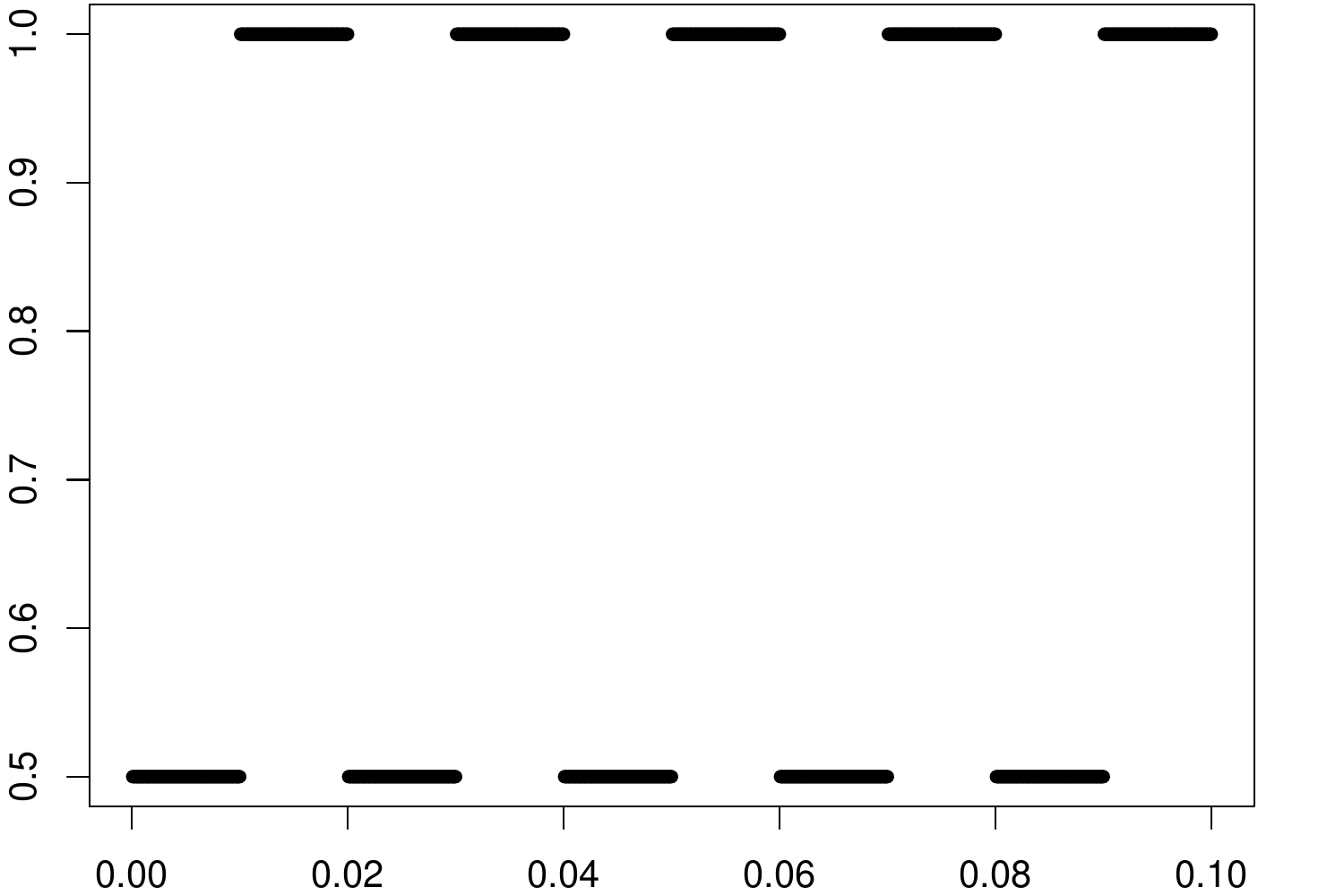}
\subcaption{}\label{subfig:const}
\end{subfigure}
\caption{\subref{subfig:sinus}: Continuous sinus-shaped standard deviation $\sigma_1(t):=1+0.5\sin(20\pi t)$. \subref{subfig:lin}: Piecewise linear standard deviation $\sigma_2(t):=0.5+\sum_{i=0}^9{(10t-i)\EINS_{(0.1i,0.1(i+1)]}(t)}$. \subref{subfig:const}: Piecewise constant standard deviation $\sigma_3(t):=\sum_{i=1}^{n/200}{0.5\EINS_{(200(i-1)/n,\ 200(i-1)/n+100/n]}(t)+\EINS_{(200(i-1)/n+100/n,\ 200i/n]}(t)}$, exemplary for $n=1\,000$.}
\label{fig:std}
\end{figure}

\begin{table}[!ht]
\scriptsize
\centering
\begin{tabular}
{l||l||cc
ccc||c|cc|cc}
Setting & Method & $\leq -2$ & -1 & 0 & +1 & $\geq +2$ & $\vert\hat{K}-K\vert$ & $\operatorname{FPSLE}$ & $\operatorname{FNSLE}$ & $\operatorname{MISE}$ & $\operatorname{MIAE}$\\
   \hline
n = 1000, & $\operatorname{HS}(0.1)$ & - & - & 0.968 & 0.032 & 0.000 & 0.033 & 16.30 & 4.12 & 0.0013 & 0.0277\\
K = 0, & $\operatorname{HS}(0.3)$ & - & - & 0.876 & 0.118 & 0.005 & 0.129 & 64.60 & 15.91 & 0.0018 & 0.0306\\
$\mu=\mu_R\equiv 0$, & $\operatorname{HS}(0.5)$ & - & - & 0.734 & 0.239 & 0.027 & 0.293 & 146.75 & 36.45 & 0.0023 & 0.0338\\
$\sigma=\sigma_1$ & $\CBS$ & - & - & 0.916 & 0.001 & 0.083 & 0.186 & 93.25 & 11.21 & 0.0045 & 0.0288\\
& $\cumSeg$ & - & - & 1.000 & 0.000 & 0.000 & 0.000 & 0.20 & 0.04 & 0.0011 & 0.0264\\
\hline
n = 1000, & $\operatorname{HS}(0.1)$ & - & - & 0.968 & 0.031 & 0.001 & 0.032 & 16.10 & 4.12 & 0.0013 & 0.0278\\
K = 0, & $\operatorname{HS}(0.3)$ & - & - & 0.876 & 0.118 & 0.005 & 0.129 & 64.55 & 15.73 & 0.0017 & 0.0306\\
$\mu=\mu_R\equiv 0$, & $\operatorname{HS}(0.5)$ & - & - & 0.734 & 0.241 & 0.024 & 0.292 & 145.80 & 35.28 & 0.0022 & 0.0340\\
$\sigma=\sigma_2$ & $\CBS$ & - & - & 0.937 & 0.004 & 0.060 & 0.135 & 67.70 & 8.96 & 0.0034 & 0.0281\\
& $\cumSeg$ & - & - & 0.999 & 0.001 & 0.000 & 0.001 & 0.40 & 0.12 & 0.0011 & 0.0264\\
\hline
n = 1000, & $\operatorname{HS}(0.1)$ & - & - & 0.969 & 0.030 & 0.001 & 0.032 & 15.75 & 3.91 & 0.0007 & 0.0210\\
K = 0, & $\operatorname{HS}(0.3)$ & - & - & 0.875 & 0.119 & 0.006 & 0.130 & 65.10 & 16.31 & 0.0009 & 0.0227\\
$\mu=\mu_R\equiv 0$, & $\operatorname{HS}(0.5)$ & - & - & 0.737 & 0.236 & 0.026 & 0.290 & 145.15 & 36.09 & 0.0012 & 0.0250\\
$\sigma=\sigma_3$ & $\CBS$ & - & - & 0.937 & 0.002 & 0.061 & 0.134 & 67.10 & 8.64 & 0.0019 & 0.0213\\
& $\cumSeg$ & - & - & 0.999 & 0.001 & 0.000 & 0.001 & 0.35 & 0.10 & 0.0006 & 0.0199\\
\hline
n = 10000, & $\operatorname{HS}(0.1)$ & 0.013 & 0.185 & 0.796 & 0.005 & 0.000 & 0.218 & 661.16 & 755.83 & 0.0212 & 0.0684\\
K = 10, & $\operatorname{HS}(0.3)$ & 0.003 & 0.076 & 0.890 & 0.031 & 0.001 & 0.113 & 543.91 & 548.21 & 0.0167 & 0.0585\\
$\lambda_{\min}=50$, & $\operatorname{HS}(0.5)$ & 0.001 & 0.041 & 0.886 & 0.069 & 0.003 & 0.117 & 513.55 & 468.37 & 0.0147 & 0.0542\\
$\mu=\mu_R$, & $\CBS$ & 0.000 & 0.001 & 0.191 & 0.155 & 0.653 & 2.636 & 1590.35 & 276.51 & 0.0092 & 0.0358\\
$\sigma=\sigma_1$ & $\cumSeg$ & 0.206 & 0.118 & 0.413 & 0.193 & 0.070 & 0.984 & 790.10 & 1054.73 & 0.0146 & 0.0502\\
 \hline
n = 10000, & $\operatorname{HS}(0.1)$ & 0.014 & 0.205 & 0.776 & 0.006 & 0.000 & 0.238 & 421.19 & 513.32 & 0.0156 & 0.0556\\
K = 10, & $\operatorname{HS}(0.3)$ & 0.001 & 0.077 & 0.894 & 0.027 & 0.001 & 0.108 & 348.50 & 358.14 & 0.0119 & 0.0475\\
$\lambda_{\min}=50$, & $\operatorname{HS}(0.5)$ & 0.000 & 0.038 & 0.897 & 0.062 & 0.002 & 0.105 & 344.93 & 311.35 & 0.0106 & 0.0446\\
$\mu=\mu_R$, & $\CBS$ & 0.000 & 0.000 & 0.215 & 0.174 & 0.611 & 2.362 & 1454.85 & 247.26 & 0.0085 & 0.0346\\
$\sigma=\sigma_2$ & $\cumSeg$ & 0.114 & 0.102 & 0.467 & 0.236 & 0.082 & 0.795 & 756.12 & 720.95 & 0.0136 & 0.0478\\
 \hline
n = 10000, & $\operatorname{HS}(0.1)$ & 0.019 & 0.233 & 0.744 & 0.004 & 0.000 & 0.276 & 161.27 & 251.06 & 0.0053 & 0.0301\\
K = 10, & $\operatorname{HS}(0.3)$ & 0.002 & 0.069 & 0.904 & 0.025 & 0.000 & 0.099 & 137.86 & 136.79 & 0.0036 & 0.0254\\
$\lambda_{\min}=50$, & $\operatorname{HS}(0.5)$ & 0.000 & 0.029 & 0.906 & 0.062 & 0.003 & 0.096 & 170.29 & 128.56 & 0.0033 & 0.0248\\
$\mu=\mu_R$, & $\CBS$ & 0.000 & 0.000 & 0.246 & 0.173 & 0.582 & 2.189 & 1134.85 & 214.71 & 0.0047 & 0.0263\\
$\sigma=\sigma_3$ & $\cumSeg$ & 0.054 & 0.051 & 0.516 & 0.279 & 0.101 & 0.669 & 749.33 & 499.10 & 0.0070 & 0.0346\\
\hline
\end{tabular}
\caption{Simulations with standard deviations $\sigma_1(\cdot)$-$\sigma_3(\cdot)$ from Figure \ref{fig:std} and $C=200$. Columns from left to right: setting, method, proportions of $\hat{K}-K$ and averages of the corresponding error criteria. $\operatorname{HS}(\alpha)$ denotes $\HSMUCE$ at significance level $\alpha$.}
\label{tab:robustvar}
\end{table}

\begin{table}[!ht]
\scriptsize
\centering
\begin{tabular}
{l||l||cc
ccc||c|cc|cc}
Setting & Method & $\leq -2$ & -1 & 0 & +1 & $\geq +2$ & $\vert\hat{K}-K\vert$ & $\operatorname{FPSLE}$ & $\operatorname{FNSLE}$ & $\operatorname{MISE}$ & $\operatorname{MIAE}$\\
   \hline
$\mu=\mu_T$, & $\operatorname{HS}(0.1)$ & 0.137 & 0.421 & 0.439 & 0.003 & 0.000 & 0.709 & 24.74 & 58.06 & 0.3976 & 0.2449\\
$a=0.01$, & $\operatorname{HS}(0.3)$ & 0.019 & 0.209 & 0.741 & 0.032 & 0.000 & 0.279 & 15.65 & 24.13 & 0.1814 & 0.1681\\
$b=0.1$, & $\operatorname{HS}(0.5)$ & 0.003 & 0.091 & 0.822 & 0.081 & 0.003 & 0.184 & 17.32 & 15.63 & 0.1206 & 0.1474\\
$\sigma=\sigma_R$ & $\CBS$ & 0.001 & 0.011 & 0.383 & 0.290 & 0.315 & 1.141 & 72.93 & 22.14 & 0.1321 & 0.1535\\
& $\cumSeg$ & 0.443 & 0.237 & 0.192 & 0.085 & 0.043 & 1.663 & 95.08 & 225.57 & 0.3080 & 0.2823\\
\hline
$\mu=\mu_T$, & $\operatorname{HS}(0.1)$ & 0.149 & 0.360 & 0.370 & 0.104 & 0.016 & 0.821 & 73.52 & 94.89 & 0.4410 & 0.3070\\
$a=0.01$, & $\operatorname{HS}(0.3)$ & 0.029 & 0.181 & 0.496 & 0.226 & 0.067 & 0.611 & 78.00 & 62.38 & 0.2304 & 0.2376\\
$b=0.3$, & $\operatorname{HS}(0.5)$ & 0.007 & 0.092 & 0.466 & 0.306 & 0.129 & 0.697 & 88.73 & 53.42 & 0.1595 & 0.2169\\
$\sigma=\sigma_R$ & $\CBS$ & 0.001 & 0.006 & 0.082 & 0.135 & 0.776 & 3.243 & 249.67 & 74.97 & 0.1646 & 0.2325\\
& $\cumSeg$ & 0.439 & 0.233 & 0.200 & 0.086 & 0.043 & 1.652 & 107.57 & 237.10 & 0.3394 & 0.3222\\
\hline
$\mu=\mu_T$, & $\operatorname{HS}(0.1)$ & 0.140 & 0.287 & 0.323 & 0.176 & 0.075 & 0.936 & 134.92 & 135.90 & 0.5279 & 0.4052\\
$a=0.01$, & $\operatorname{HS}(0.3)$ & 0.032 & 0.146 & 0.327 & 0.298 & 0.197 & 0.970 & 146.40 & 103.03 & 0.3106 & 0.3393\\
$b=0.5$, & $\operatorname{HS}(0.5)$ & 0.009 & 0.076 & 0.258 & 0.329 & 0.328 & 1.223 & 162.18 & 92.24 & 0.2410 & 0.3207\\
$\sigma=\sigma_R$ & $\CBS$ & 0.002 & 0.004 & 0.020 & 0.043 & 0.932 & 5.623 & 435.32 & 124.66 & 0.2462 & 0.3440\\
& $\cumSeg$ & 0.420 & 0.244 & 0.190 & 0.093 & 0.054 & 1.641 & 127.88 & 248.73 & 0.3921 & 0.3823\\
\hline
$\mu=\mu_T$, & $\operatorname{HS}(0.1)$ & 0.128 & 0.424 & 0.446 & 0.002 & 0.000 & 0.693 & 23.50 & 56.36 & 0.4066 & 0.2415\\
$a=0.025$, & $\operatorname{HS}(0.3)$ & 0.017 & 0.201 & 0.759 & 0.023 & 0.001 & 0.259 & 13.43 & 22.21 & 0.1854 & 0.1628\\
$b=0.1$, & $\operatorname{HS}(0.5)$ & 0.003 & 0.086 & 0.843 & 0.066 & 0.002 & 0.162 & 14.51 & 13.77 & 0.1218 & 0.1416\\
$\sigma=\sigma_R$ & $\CBS$ & 0.002 & 0.008 & 0.395 & 0.287 & 0.308 & 1.135 & 64.07 & 19.55 & 0.1304 & 0.1471\\
& $\cumSeg$ & 0.440 & 0.240 & 0.188 & 0.086 & 0.046 & 1.672 & 95.36 & 229.25 & 0.3058 & 0.2796\\
\hline
$\mu=\mu_T$, & $\operatorname{HS}(0.1)$ & 0.108 & 0.344 & 0.411 & 0.111 & 0.027 & 0.738 & 58.57 & 73.26 & 0.4223 & 0.2606\\
$a=0.025$, & $\operatorname{HS}(0.3)$ & 0.016 & 0.138 & 0.468 & 0.252 & 0.126 & 0.715 & 77.74 & 50.19 & 0.2143 & 0.1929\\
$b=0.3$, & $\operatorname{HS}(0.5)$ & 0.003 & 0.058 & 0.382 & 0.315 & 0.243 & 0.989 & 101.66 & 48.03 & 0.1503 & 0.1749\\
$\sigma=\sigma_R$ & $\CBS$ & 0.002 & 0.003 & 0.050 & 0.065 & 0.880 & 5.370 & 356.54 & 85.59 & 0.1585 & 0.2027\\
& $\cumSeg$ & 0.438 & 0.241 & 0.184 & 0.091 & 0.046 & 1.678 & 101.67 & 234.23 & 0.3243 & 0.2958\\
\hline
$\mu=\mu_T$, & $\operatorname{HS}(0.1)$ & 0.054 & 0.180 & 0.276 & 0.226 & 0.264 & 1.247 & 164.83 & 114.26 & 0.4732 & 0.3127\\
$a=0.025$, & $\operatorname{HS}(0.3)$ & 0.007 & 0.060 & 0.195 & 0.229 & 0.509 & 1.945 & 214.89 & 100.21 & 0.2748 & 0.2586\\
$b=0.5$, & $\operatorname{HS}(0.5)$ & 0.001 & 0.027 & 0.127 & 0.191 & 0.654 & 2.591 & 256.78 & 101.30 & 0.2115 & 0.2465\\
$\sigma=\sigma_R$ & $\CBS$ & 0.000 & 0.001 & 0.006 & 0.011 & 0.982 & 10.383 & 709.91 & 149.02 & 0.2261 & 0.2993\\
& $\cumSeg$ & 0.439 & 0.238 & 0.184 & 0.088 & 0.050 & 1.698 & 113.20 & 245.71 & 0.3520 & 0.3206\\
\hline
$\mu=\mu_{T_\sigma}$, & $\operatorname{HS}(0.1)$ & 0.151 & 0.435 & 0.411 & 0.002 & 0.000 & 0.755 & 26.56 & 64.22 & 0.4469 & 0.3000\\
$a=0.01$, & $\operatorname{HS}(0.3)$ & 0.021 & 0.230 & 0.725 & 0.023 & 0.000 & 0.297 & 15.80 & 27.08 & 0.2289 & 0.2245\\
$b=0.2$, & $\operatorname{HS}(0.5)$ & 0.004 & 0.103 & 0.819 & 0.071 & 0.003 & 0.188 & 17.53 & 17.34 & 0.1622 & 0.2019\\
$\sigma=\sigma_R$ & $\CBS$ & 0.004 & 0.012 & 0.342 & 0.305 & 0.338 & 1.225 & 80.55 & 26.89 & 0.1689 & 0.2056\\
& $\cumSeg$ & 0.422 & 0.233 & 0.193 & 0.095 & 0.056 & 1.653 & 107.36 & 234.27 & 0.3431 & 0.3292\\
\hline
$\mu=\mu_{T_\sigma}$, & $\operatorname{HS}(0.1)$ & 0.254 & 0.410 & 0.298 & 0.036 & 0.001 & 1.012 & 68.27 & 115.70 & 0.6346 & 0.4582\\
$a=0.01$, & $\operatorname{HS}(0.3)$ & 0.055 & 0.261 & 0.488 & 0.176 & 0.019 & 0.591 & 81.26 & 80.78 & 0.4483 & 0.3953\\
$b=0.5$, & $\operatorname{HS}(0.5)$ & 0.015 & 0.129 & 0.466 & 0.321 & 0.070 & 0.624 & 100.70 & 72.34 & 0.3826 & 0.3694\\
$\sigma=\sigma_R$ & $\CBS$ & 0.002 & 0.004 & 0.022 & 0.059 & 0.914 & 4.231 & 357.42 & 117.42 & 0.2907 & 0.3140\\
& $\cumSeg$ & 0.332 & 0.211 & 0.198 & 0.139 & 0.121 & 1.575 & 179.79 & 280.65 & 0.4522 & 0.4317\\
\hline
$\mu=\mu_{T_\sigma}$, & $\operatorname{HS}(0.1)$ & 0.136 & 0.440 & 0.422 & 0.002 & 0.000 & 0.726 & 24.82 & 59.14 & 0.4567 & 0.3165\\
$a=0.025$, & $\operatorname{HS}(0.3)$ & 0.017 & 0.219 & 0.746 & 0.018 & 0.000 & 0.273 & 14.12 & 24.13 & 0.2432 & 0.2435\\
$b=0.2$, & $\operatorname{HS}(0.5)$ & 0.003 & 0.096 & 0.843 & 0.056 & 0.002 & 0.162 & 14.44 & 14.77 & 0.1756 & 0.2222\\
$\sigma=\sigma_R$ & $\CBS$ & 0.003 & 0.011 & 0.353 & 0.295 & 0.338 & 1.231 & 71.51 & 23.32 & 0.1892 & 0.2287\\
& $\cumSeg$ & 0.432 & 0.238 & 0.182 & 0.094 & 0.054 & 1.688 & 103.19 & 236.34 & 0.3604 & 0.3501\\
\hline
$\mu=\mu_{T_\sigma}$, & $\operatorname{HS}(0.1)$ & 0.181 & 0.433 & 0.370 & 0.016 & 0.000 & 0.831 & 37.89 & 75.31 & 0.7365 & 0.5244\\
$a=0.025$, & $\operatorname{HS}(0.3)$ & 0.033 & 0.240 & 0.594 & 0.125 & 0.009 & 0.450 & 42.28 & 44.37 & 0.5518 & 0.4736\\
$b=0.5$, & $\operatorname{HS}(0.5)$ & 0.007 & 0.110 & 0.541 & 0.281 & 0.061 & 0.534 & 64.52 & 41.39 & 0.4981 & 0.4582\\
$\sigma=\sigma_R$ & $\CBS$ & 0.002 & 0.002 & 0.023 & 0.043 & 0.929 & 5.589 & 365.42 & 103.32 & 0.4594 & 0.4362\\
& $\cumSeg$ & 0.316 & 0.184 & 0.179 & 0.139 & 0.182 & 1.735 & 158.60 & 254.38 & 0.6153 & 0.5317\\
\hline
\end{tabular}
\caption{Simulations with small periodic trends in the mean and $n=1\,000$, $K=10$, $\lambda_{\min}=30$ and $C=200$. Columns from left to right: setting, method, proportions of $\hat{K}-K$ and averages of the corresponding error criteria. $\operatorname{HS}(\alpha)$ denotes $\HSMUCE$ at significance level $\alpha$.}
\label{tab:robustnesstrend}
\end{table}

\begin{table}[!ht]
\scriptsize
\centering
\begin{tabular}
{l||l||cc
ccc||c|cc|cc}
Setting & Method & $\leq -2$ & -1 & 0 & +1 & $\geq +2$ & $\vert\hat{K}-K\vert$ & $\operatorname{FPSLE}$ & $\operatorname{FNSLE}$ & $\operatorname{MISE}$ & $\operatorname{MIAE}$\\
\hline
n = 1000, & $\operatorname{HS}(0.1)$ & - & - & 0.982 & 0.018 & 0.000 & 0.018 & 9.05 & 2.53 & 0.0031 & 0.0347\\
K = 0, & $\operatorname{HS}(0.3)$ & - & - & 0.927 & 0.071 & 0.001 & 0.074 & 37.05 & 10.31 & 0.0034 & 0.0361\\
$\mu=\mu_R\equiv 0$, & $\operatorname{HS}(0.5)$ & - & - & 0.824 & 0.167 & 0.009 & 0.185 & 92.40 & 26.35 & 0.0043 & 0.0392\\
$\sigma=\sigma_R$, & $\operatorname{S}(0.1)$ & - & - & 0.001 & 0.001 & 0.999 & 11.859 & 5929.70 & 369.55 & 0.8710 & 0.1491\\
& $\operatorname{S}(0.3)$ & - & - & 0.000 & 0.000 & 1.000 & 14.803 & 7401.65 & 397.77 & 0.9338 & 0.1674\\
& $\operatorname{S}(0.5)$ & - & - & 0.000 & 0.000 & 1.000 & 16.862 & 8431.00 & 411.30 & 0.9730 & 0.1787\\
& $\CBS$ & - & - & 0.991 & 0.000 & 0.009 & 0.018 & 9.05 & 1.13 & 0.0058 & 0.0340\\
& $\cumSeg$ & - & - & 0.955 & 0.001 & 0.044 & 0.188 & 93.90 & 11.98 & 0.0682 & 0.0375\\
\hline
n = 1000, & $\operatorname{HS}(0.1)$ & 0.008 & 0.136 & 0.848 & 0.007 & 0.000 & 0.160 & 25.70 & 62.95 & 0.0120 & 0.0578\\
K = 2, & $\operatorname{HS}(0.3)$ & 0.003 & 0.086 & 0.876 & 0.035 & 0.000 & 0.127 & 29.74 & 44.61 & 0.0103 & 0.0537\\
$\lambda_{\min} = 30$, & $\operatorname{HS}(0.5)$ & 0.001 & 0.055 & 0.851 & 0.090 & 0.003 & 0.152 & 44.21 & 38.62 & 0.0097 & 0.0524\\
$\mu=\mu_R$, & $\operatorname{S}(0.1)$ & 0.000 & 0.000 & 0.001 & 0.001 & 0.998 & 11.104 & 2683.40 & 250.21 & 0.3046 & 0.1232\\
$\sigma\equiv 1$, & $\operatorname{S}(0.3)$ & 0.000 & 0.000 & 0.000 & 0.000 & 1.000 & 13.984 & 3361.80 & 283.17 & 0.3264 & 0.1340\\
& $\operatorname{S}(0.5)$ & 0.000 & 0.000 & 0.000 & 0.000 & 1.000 & 15.991 & 3836.28 & 302.43 & 0.3400 & 0.1419\\
& $\CBS$ & 0.053 & 0.161 & 0.726 & 0.043 & 0.018 & 0.346 & 46.69 & 119.74 & 0.0241 & 0.0712\\
& $\cumSeg$ & 0.025 & 0.097 & 0.722 & 0.093 & 0.063 & 0.456 & 108.11 & 81.86 & 0.0557 & 0.0707\\
\hline 
n = 10000, & $\operatorname{HS}(0.1)$ & 0.002 & 0.079 & 0.916 & 0.004 & 0.000 & 0.086 & 93.09 & 119.69 & 0.0130 & 0.0425\\
K = 10, & $\operatorname{HS}(0.3)$ & 0.000 & 0.025 & 0.957 & 0.017 & 0.000 & 0.043 & 86.32 & 81.78 & 0.0105 & 0.0397\\
$\lambda_{\min}=50$, & $\operatorname{HS}(0.5)$ & 0.000 & 0.012 & 0.950 & 0.038 & 0.000 & 0.050 & 99.93 & 76.24 & 0.0097 & 0.0389\\
$\mu=\mu_R$, & $\CBS$ & 0.467 & 0.148 & 0.167 & 0.107 & 0.111 & 2.516 & 1356.25 & 6254.20 & 0.0877 & 0.1308\\
$\sigma=\sigma_R$ & $\cumSeg$ & 0.586 & 0.192 & 0.136 & 0.055 & 0.032 & 2.242 & 997.13 & 3005.71 & 0.0433 & 0.0906\\
\hline
\end{tabular}
\caption{Simulations with $t_3$ distributed errors and $C=200$. Columns from left to right: setting, method, proportions of $\hat{K}-K$ and averages of the corresponding error criteria. $\operatorname{HS}(\alpha)$ and $\operatorname{S}(\alpha)$ denote $\HSMUCE$ and $\SMUCE$ at significance level $\alpha$, respectively.}
\label{tab:robustnonnormal}
\end{table}

\section{Proofs}\label{sec:proofs}
In this section we collect the proofs together with some auxiliary statements.
\subsection{Proof of Lemma \ref{lemma:uniqunesscritval}}
\begin{proof}[Proof of Lemma \ref{lemma:uniqunesscritval}]
A single statistic $T_i^j(Z,0)$ has the c.d.f. $F_{1,j-i}(\cdot)$ of an F-distribution with $(1,j-i)$ degrees of freedom. 
Thus, $F_k(\cdot)=F_{1,2^k-1}(\cdot)^{\vert \mathcal{D}_k\vert}$ is continuous and strictly monotonically increasing for positive arguments. 
Now, it follows from equation \eqref{eq:balancing} that
\begin{equation}\label{eq:qkq1}
q_k=F_k^{-1}\left(1-\frac{\beta_k}{\beta_1}\big(1-F_1(q_1)\big)\right)\text{ for } k=2,\ldots,\sn.
\end{equation}
This together with equation \eqref{eq:significancelevel} yields
\begin{equation*}
G(q_1):=F\left(q_1,F_2^{-1}\left(1-\frac{\beta_2}{\beta_1}\big(1-F_1(q_1)\big)\right),\ldots,F_\sn^{-1}\left(1-\frac{\beta_\sn}{\beta_1}\big(1-F_1(q_1)\big)\right)\right)=1-\alpha.
\end{equation*}
Note, that $F$ is continuous and $\lim_{q_k\rightarrow 0}{F\left(q_1,\ldots,q_\sn\right)}=0$ for all $k=1,\ldots,\sn$ as well as $\lim_{q_1,\ldots,q_\sn\rightarrow \infty}{F\left(q_1,\ldots,q_\sn\right)}=1$. 
Thus, the function $G$ is continuous, strictly monotonically increasing on $[0,\infty)$ and attains all values in $[0,1)$. Therefore, the existence of the vector of critical values follows from the intermediate value theorem and the vector is also unique.
\end{proof}

\subsection{Proof of Lemma \ref{lemma:boundcritval}}
First of all, recall from the proof of Lemma \ref{lemma:uniqunesscritval} that the statistic $T_k$ has c.d.f. $F_{1,2^k-1}(\cdot)^{\vert \mathcal{D}_k\vert}$. For every $k=1,\ldots,\sn$ we use the transformation 
\begin{equation*}
U_k:=F_{1,2^k-1}\left(T_k\right)^{\vert \mathcal{D}_k\vert}
\end{equation*}
and the identity 
\begin{equation*}
T_k=F_{1,2^k-1}^{-1}\left(U_k^{\vert \mathcal{D}_k\vert^{-1}}\right).
\end{equation*}
Here, $F^{-1}_{1,2^k-1}\left(\cdot\right)$ denotes the quantile function of an F-distribution with $(1,2^k-1)$ degrees of freedom. 
Analogously, we define 
\begin{equation*}
q_{k,U}:=F_{1,2^k-1}\left(q_k\right)^{\vert \mathcal{D}_k\vert}
\end{equation*}
and have the identity
\begin{equation}\label{eq:identityqU}
q_k=F_{1,2^k-1}^{-1}\left(q_{k,U}^{\vert \mathcal{D}_k\vert^{-1}}\right).
\end{equation}
Then, the events $U_k>q_{k,\operatorname{U}}$ and $T_k>q_{k}$ are equivalent and therefore the vector $\textbf{q}_U=(q_{1,\operatorname{U}}, \ldots, q_{\sn,\operatorname{U}})$ satisfies similar conditions to the equations \eqref{eq:significancelevel} and \eqref{eq:balancing}, i.e.
\begin{equation}\label{eq:significancelevelqU}
1-\Pj\left(U_1\leq q_{1,U},\ldots,U_\sn \leq q_{\sn,U}\right)=\alpha
\end{equation}
and
\begin{equation}\label{eq:balancingqU}
\frac{1-\Pj\left(U_1\leq q_{1,U}\right)}{\beta_1}=\cdots=\frac{1-\Pj\left(U_\sn\leq q_{\sn,U}\right)}{\beta_\sn}.
\end{equation}
The following bounds can be interpreted as a weighted version of the Bonferroni-inequality.
\begin{Lemma}\label{lemma:boundqU}
$q_{k,U}\leq 1-\alpha\beta_k$ for $k=1,\ldots,\sn$.
\end{Lemma}
\begin{proof}
We have
$\Pj\left(U_j\leq q_{j,U}\right)=q_{j,U}$ for $j=1,\ldots,\sn$, since $U_j$ is uniformly distributed. Moreover, it follows from condition \eqref{eq:balancingqU} that $1-q_{j,U}=(1-q_{k,U})\beta_j/\beta_k$. Combining this with equation \eqref{eq:significancelevelqU} and $\sum_{j=1}^{\sn}{\beta_j}=1$ yields
\begin{equation*}
\begin{split}
\alpha = & 1-\Pj\left(U_1\leq q_{1,U},\ldots,U_\sn \leq q_{\sn,U}\right)\\
\leq & \sum_{j=1}^{\sn}{\Pj\left(U_j>q_{j,U}\right)}=\sum_{j=1}^{\sn}{(1-q_{j,U})}
= \sum_{j=1}^{\sn}{(1-q_{k,U})\frac{\beta_j}{\beta_k}}=\frac{1-q_{k,U}}{\beta_k},
\end{split}
\end{equation*}
which proves the assertion.
\end{proof}
Lemma \ref{lemma:boundFquantile} bounds the quantile function of an F-distribution with $(1,c)$ degrees of freedom.
\begin{Lemma}[Bounds on the F-quantiles]\label{lemma:boundFquantile}
Let $F^{-1}_{1,c}(y)$ be the quantile function of an F-distribution with $(1,c)$ degrees of freedom, then
\begin{equation*}
c\left[\left(1-y^2\right)^{-\frac{1}{c}}-1\right]\leq F^{-1}_{1,c}(y)\leq c\left[\left(1-y^2\right)^{-\frac{2}{c-\frac{1}{2}}}-1\right].
\end{equation*}
\end{Lemma}
\begin{proof}
We have from \citep[Theorem 4.2]{FujikoshiMukaihata93} that
\begin{equation*}
c\left[\exp\left(\frac{\left({\chi^2_1}\right)^{-1}(y)}{c}\right)-1\right]\leq F^{-1}_{1,c}(y) \leq c\left[\exp\left(\frac{\left({\chi^2_1}\right)^{-1}(y)}{c-\frac{1}{2}}\right)-1\right],
\end{equation*}
with $\left({\chi^2_1}\right)^{-1}(y)$ the quantile function of the chi-squared distribution with one degree of freedom. Moreover, we obtain for all $y\geq 0$
\begin{equation*}
\begin{split}
& \Pj\left(\chi^2_1\leq y\right)=\Pj\left(-\sqrt{y}\leq Z \leq \sqrt{y}\right)=2\Phi\left(\sqrt{y}\right)-1\\
\Longleftrightarrow\ & \left({\chi^2_1}\right)^{-1}(y) = \Phi^{-1}\left(\frac{y+1}{2}\right)^2,
\end{split}
\end{equation*}
where $\Phi^{-1}(y)$ is the quantile of the standard gaussian distribution. Furthermore, we have from \citep[(13.48), p. 115]{Johnsonetal} that
\begin{equation*}
\frac{1}{2}\left[1+\left(1-\exp\left(-\frac{x^2}{2}\right)\right)^{\frac{1}{2}}\right] \leq \Phi\left(x\right) \leq \frac{1}{2}\Bigg[1+\bigg(1-\exp\bigg(-x^2\bigg)\bigg)^{\frac{1}{2}}\Bigg]
\end{equation*}
and so for the quantile function one finds
\begin{equation*}
\sqrt{-\log\big(1-(2y-1)^2\big)}\leq \Phi^{-1}\left(y\right) \leq \sqrt{-2\log\big(1-(2y-1)^2\big)}.
\end{equation*}
Combining the formulas proves the assertion.
\end{proof}
\begin{proof}[Proof of Lemma \ref{lemma:boundcritval}]
First of all, \eqref{eq:identityqU} and the equation $\vert \mathcal{D}_k\vert=\lfloor n2^{-k}\rfloor$ yields
\begin{equation*}
q_{k}= F_{1,2^k-1}^{-1}\left(q_{k,U}^{\vert \mathcal{D}_k\vert^{-1}}\right)
= F_{1,2^k-1}^{-1}\left(q_{k,U}^{\left\lfloor n 2^{-k}\right\rfloor^{-1}}\right)
\leq F_{1,2^k-1}^{-1}\left(q_{k,U}^{2^k/n}\right).
\end{equation*}
Moreover, it follows from the Lemmas \ref{lemma:boundqU} and \ref{lemma:boundFquantile} that
\begin{equation*}
\begin{split}
q_{k}\leq F_{1,2^k-1}^{-1}\left(q_{k,U}^{2^k/n}\right) \leq &F_{1,2^k-1}^{-1}\left((1-\alpha\beta_k)^{2^k/n}\right)\\
\leq &\left(2^k-1\right)\left[\left(1-\left((1-\alpha\beta_k)^{2^k/n}\right)^2\right)^{-\frac{2}{\left(2^k-1\right)-\frac{1}{2}}}-1\right]\\
\leq &2^k\left[\left(1-(1-\alpha\beta_k)^{2^k/n}\right)^{-\frac{4}{2^{k+1}-3}}-1\right].
\end{split}
\end{equation*}
Applying Bernoulli's inequality $\left(1-x\right)^c\leq 1-cx$ gives
\begin{equation*}
q_{k}
\leq 2^k\left[\left(1-(1-\alpha\beta_k)^{2^k/n}\right)^{-\frac{4}{2^{k+1}-3}}-1\right]
\leq 2^k\left[\left(\frac{2^k\alpha\beta_k}{n}\right)^{-\frac{4}{2^{k+1}-3}}-1\right].
\end{equation*}
Moreover, for $x,c>0$ the inequality $c^x\leq 1+2x\log(c)$ holds whenever $x\log(c)\leq 1$. Together with the assumption $k \geq 2$ we finally obtain
\begin{equation*}
q_{k}
\leq 2^k\left[\left(\frac{2^k\alpha\beta_k}{n}\right)^{-\frac{4}{2^{k+1}-3}}-1\right]
\leq 4\frac{2^{k+1}}{2^{k+1}-3}\log\left(\frac{n}{2^k\alpha\beta_k}\right)
\leq 8\log\left(\frac{n}{2^k\alpha\beta_k}\right)
\end{equation*}
if 
\begin{equation*}
2^{-k}\log\left(\frac{n}{2^k\alpha\beta_k}\right)\leq \frac{1}{2}\frac{2^{k+1}-3}{2^{k+1}}\leq \frac{1}{2}.
\end{equation*}
\end{proof}

\subsection{Exponential deviation bounds}\label{sec:prooflargedeviation}
For the subsequent proofs we need a bound for the distribution function of a single test statistic $T_i^j$ \eqref{eq:localtest} which is in our setting a bound for the c.d.f. of a non-central F-distribution.
\begin{Lemma}\label{lemma:boundcdfstat}
Let $Y_1,\ldots,Y_n$ be i.i.d. gaussian random variables with expectation $\valmu\in\R$ and variance $\valsd^2>0$. Let $x_+:=\max(x,0)$. Then, for any $\delta\neq 0$, $q>0$
\begin{equation}\label{eq:boundcdfstat}
\begin{split}
&\Pj\left(T_1^n(Y,\valmu+\delta)\leq q\right)\\
\leq & \min_{z\geq 0}\left\{\exp\left(-\frac{1}{2}\left(\frac{\Delta\sqrt{n}}{2}-\frac{q(1+z)}{\Delta\sqrt{n}}\right)_+^2\right)+\exp\left(-(n-1)\frac{z-\log(1+z)}{2}\right)\right\},
\end{split}
\end{equation}
where $\Delta:=\vert \delta\vert/\valsd$.
\end{Lemma}
\begin{proof}
Let $\widetilde{T}_i^j(Y,\valmu):=\left(j-i+1\right)\left(\overline{Y}_{ij}-\valmu\right)^2/\valsd^2$. Then,
\begin{equation*}
T_1^n(Y,\valmu+\delta)=\frac{\widetilde{T}_1^n(Y,\valmu+\delta)}{\hat{\valsd}_{1n}^2/\valsd^2}.
\end{equation*}
The statistics $\hat{\valsd}_{1n}^2/\valsd^2$ and $\widetilde{T}_1^n(Y,\valmu+\delta)$ are independent, since $\widetilde{T}_1^n(Y,\valmu+\delta)$ depends only on the mean $\overline{Y}_{1n}$. Hence, for all $z\geq 0$
\begin{equation*}
\begin{split}
\Pj\left(T_1^n(Y,\valmu+\delta)\leq q\right)
=&\Pj\left(\widetilde{T}_1^n(Y,\valmu+\delta)\leq q \frac{\hat{\valsd}_{1n}^2}{\valsd^2}\right)\\
=&\Pj\left(\widetilde{T}_1^n(Y,\valmu+\delta)\leq q \frac{\hat{\valsd}_{1n}^2}{\valsd^2}\bigg\vert \frac{\hat{\valsd}_{1n}^2}{\valsd^2}\leq 1+z\right)\Pj\left(\frac{\hat{\valsd}_{1n}^2}{\valsd^2}\leq 1+z\right)\\
&+\Pj\left(\widetilde{T}_1^n(Y,\valmu+\delta)\leq q \frac{\hat{\valsd}_{1n}^2}{\valsd^2}\bigg\vert \frac{\hat{\valsd}_{1n}^2}{\valsd^2}> 1+z\right)\Pj\left(\frac{\hat{\valsd}_{1n}^2}{\valsd^2}> 1+z\right)\\
\leq & \Pj\left(\widetilde{T}_1^n(Y,\valmu+\delta)\leq q(1+z)\right)+\Pj\left(\frac{\hat{\valsd}_{1n}^2}{\valsd^2}> 1+z\right)\\
\leq & \exp\left(-\frac{1}{2}\left(\frac{\Delta\sqrt{n}}{2}-\frac{q(1+z)}{\Delta\sqrt{n}}\right)_+^2\right)+\exp\left(-(n-1)\frac{z-\log(1+z)}{2}\right).
\end{split}
\end{equation*}
The first term of the last inequality follows from \citep[Lemma 7.3 and the proof]{SMUCE} and the second from \citep[Theorem 2.1]{SpokoinyZhilova13}, since $(n-1)\hat{\valsd}_{1n}^2/\valsd^2\sim \chi^2_{n-1}$.\\
It remains to show that the minimum in \eqref{eq:boundcdfstat} is attained for some $z\geq 0$. The function $(\Delta\sqrt{n}/2-q(1+z)/(\Delta\sqrt{n}))_+^2$ is strictly monotonically decreasing for $z > 0$ until the function value zero is attained for some finite $z$. The function $(n-1)(z-\log(1+z))$ is zero for $z=0$ and strictly monotonically increasing on $[0,\infty)$. Therefore, the two continuous functions intersect and the minimum is attained for some $z\geq 0$.
\end{proof}
The minimum in the last lemma cannot be determined analytically, but it can be computed numerically. In Lemma \ref{lemma:boundcdfstatfurther} we estimate the right hand side further to obtain an explicit exponential bound.
\begin{Lemma}\label{lemma:boundcdfstatfurther}
Let $Y_1,\ldots,Y_n$, $n\geq 4$, be i.i.d. gaussian random variables with expectation $\valmu\in\R$ and variance $\valsd^2>0$, then we have for all $q>0$ with
\begin{equation}\label{eq:conditionboundcdf}
\frac{q}{n}\leq \frac{1}{8}
\end{equation}
as well as for all $\delta\neq 0$ and $\Delta:=\vert \delta\vert/\valsd$ the bound
\begin{equation}\label{eq:largedeviationboundF}
\Pj\left(T_1^n(Y,\valmu+\delta)\leq q\right)
\leq 2\exp\left(-\frac{1}{48}\left(\sqrt{n}\Delta-\sqrt{2q}\right)_+^2\right).
\end{equation}
\end{Lemma}
\begin{proof}
Let $z>0$ be arbitrary, but fixed. Then, it follows from Lemma \ref{lemma:boundcdfstat} that
\begin{equation*}
\begin{split}
\Pj\left(T_1^n(Y,\valmu+\delta)\leq q\right)
\leq & \exp\left(-\frac{1}{2}\left(\frac{\Delta\sqrt{n}}{2}-\frac{q(1+z)}{\Delta\sqrt{n}}\right)_+^2\right)+\exp\left(-(n-1)\frac{z-\log(1+z)}{2}\right)\\
\leq & 2\exp\left(-\min\left[\frac{1}{2}\left(\frac{\Delta\sqrt{n}}{2}-\frac{q(1+z)}{\Delta\sqrt{n}}\right)_+^2,(n-1)\frac{z-\log(1+z)}{2}\right]\right).
\end{split}
\end{equation*}
The inequality
\begin{equation*}
z-\log(1+z)\geq\frac{1}{2}\frac{z^2}{1+z}
\geq 
\frac{1}{4}\min\left(z^2,z\right)
\end{equation*}
yields
\begin{equation*}
\begin{split}
&\Pj\left(T_1^n(Y,\valmu+\delta)\leq q\right)\\
\leq & 2\exp\left(-\min\left[\frac{1}{8}n\left(\Delta-\frac{2q(1+z)}{\Delta n}\right)_+^2,\frac{1}{8}(n-1)\min\left(z^2,z\right)\right]\right)\\
\leq 
& 2\exp\left(-\frac{1}{8}(n-1)\min\left[\min\left[\left(\Delta-\frac{2q(1+z)}{\Delta n}\right)_+^2,z^2\right],\min\left[\left(\Delta-\frac{2q(1+z)}{\Delta n}\right)_+^2,z\right]\right]\right).
\end{split}
\end{equation*}
Now, we minimize the r.h.s in $z\geq 0$. The functions $z$ and $z^2$ are both increasing, the function $(\Delta-2q(1+z)/(\Delta n))_+^2$ in contrast is decreasing in $z$. Therefore, both inner minima are attained and by solving the corresponding quadratic equations (note that we have to take the solution with $\Delta-2q(1+z)/(\Delta n)\geq 0$) we get
\begin{equation*}
\begin{split}
&\Pj\left(T_1^n(Y,\valmu+\delta)\leq q\right)\\
\leq & 2\exp\left(-\frac{1}{8}(n-1)\min\left[\left(\frac{\Delta-\frac{2q}{\Delta n}}{1+\frac{2q}{\Delta n}}\right)_+^2,\frac{1+2\frac{2q}{\Delta n}\left(\Delta-\frac{2q}{\Delta n}\right)_+ -\sqrt{1+4\frac{2q}{\Delta n}\left(\Delta-\frac{2q}{\Delta n}\right)_+}}{2\left(\frac{2q}{\Delta n}\right)^2}\right]\right).
\end{split}
\end{equation*}
Using the inequality $\sqrt{1+4x}\leq 1+2x-2x^2+4x^3$ for all $x>-1/4$ with $x=2q/(\Delta n)(\Delta-2q/(\Delta n))_+$ we find
\begin{equation*}
\begin{split}
&\Pj\left(T_1^n(Y,\valmu+\delta)\leq q\right)\\
\leq & 2\exp\left(-\frac{1}{8}\frac{n-1}{n}n\min\left[\left(\frac{\Delta-\frac{2q}{\Delta n}}{1+\frac{2q}{\Delta n}}\right)_+^2,\left(\Delta-\frac{2q}{\Delta n}\right)_+^2\left[1-2\frac{2q}{\Delta n}\left(\Delta-\frac{2q}{\Delta n}\right)_+\right]\right]\right).
\end{split}
\end{equation*}
Next, we consider the two terms in the minimum separately. We assume w.l.o.g. that $\sqrt{2q}/(\Delta \sqrt{n})\leq 1$, since otherwise the r.h.s. in \eqref{eq:largedeviationboundF} is two. For the first term we distinguish the cases $2q>\Delta n$ and $2q\leq\Delta n$. If $2q\leq \Delta n$ is satisfied, then
\begin{equation*}
n\left(\frac{\Delta-\frac{2q}{\Delta n}}{1+\frac{2q}{\Delta n}}\right)_+^2\geq \frac{1}{4}n\left(\Delta-\frac{2q}{\Delta n}\right)_+^2=\frac{1}{4}\left(\sqrt{n}\Delta-\frac{2q}{\Delta \sqrt{n}}\right)_+^2\geq \frac{1}{4}\left(\sqrt{n}\Delta-\sqrt{2q}\right)_+^2.
\end{equation*}
For the other case, when $2q> \Delta n$ holds, we obtain with $q/n\leq 1/8$
\begin{equation*}
\begin{split}
n\left(\frac{\Delta-\frac{2q}{\Delta n}}{1+\frac{2q}{\Delta n}}\right)_+^2
\geq & \frac{1}{4}n\left(\frac{\Delta-\frac{2q}{\Delta n}}{\frac{2q}{\Delta n}}\right)_+^2
=\frac{1}{4}n\left(\frac{n\Delta^2}{2q}-1\right)_+^2\\
=&\frac{1}{4}\frac{n}{2q}\left(\frac{(\sqrt{n}\Delta)^2}{\sqrt{2q}}-\sqrt{2q}\right)_+^2
\geq \left(\sqrt{n}\Delta-\sqrt{2q}\right)_+^2.
\end{split}
\end{equation*}
For the second term it follows with $q/n\leq 1/8$ that
\begin{equation*}
\begin{split}
n\left(\Delta-\frac{2q}{\Delta n}\right)_+^2\left[1-2\frac{2q}{\Delta n}\left(\Delta-\frac{2q}{\Delta n}\right)_+\right]
= &\left(\sqrt{n}\Delta-\frac{2q}{\Delta \sqrt{n}}\right)_+^2\left[1-4\frac{q}{n}\left(1-\frac{2q}{\Delta^2 n}\right)_+\right]\\
\geq &\left(\sqrt{n}\Delta-\sqrt{2q}\right)_+^2\frac{1}{2}.
\end{split}
\end{equation*}
This yields
\begin{equation*}
\begin{split}
\Pj\left(T_1^n(Y,\valmu+\delta)\leq q\right)\leq & 2\exp\left(-\frac{1}{32}\frac{n-1}{n}\left(\sqrt{n}\Delta-\sqrt{2q}\right)_+^2\right)\\
\leq & 2\exp\left(-\frac{1}{48}\left(\sqrt{n}\Delta-\sqrt{2q}\right)_+^2\right).
\end{split}
\end{equation*}
\end{proof}

\subsection{Proofs of Section \ref{sec:theory}}
\begin{proof}[Proof of Theorem \ref{theorem:boundoverestimate}]
The estimated number of change-points $\hat{K}$ is by its definition in \eqref{eq:K} equal to the minimal number of change-points of all feasible functions. Therefore, all functions with the true number of change-points (or less change-points) have to be infeasible, if the number of change-points is overestimated. Hence, by \eqref{eq:significancelevel}
\begin{equation*}
\begin{split}
\sup_{(\mu,\sigma^2)\in\mathcal{S}}\Pj_{(\mu,\sigma^2)}\left(\hat{K}>K\right)
&\leq \sup_{(\mu,\sigma^2)\in\mathcal{S}}\Pj_{(\mu,\sigma^2)}\left(\max_{\left[\frac{i}{n},\frac{j}{n}\right]\in\mathcal{D}(\mu)}\left[T_i^j\left(Y, \mu([i/n, j/n])\right)-q_{ij}\right]> 0\right)\\
&\leq \Pj_{(0,1)}\left(\max_{\left[\frac{i}{n},\frac{j}{n}\right]\in\mathcal{D}}\left[T_i^j\left(Y, 0)\right)-q_{ij}\right]>0\right)
=\alpha, 
\end{split}
\end{equation*}
where the last inequality follows from $\mathcal{D}(\mu)\subset \mathcal{D}$ and the fact that the distribution of $T_i^j\left(Y, \mu([i/n, j/n])\right)$ does not depend on $\mu(\cdot)$ and $\sigma(\cdot)$, as these are constant on intervals in $\mathcal{D}(\mu)$.
\end{proof}
\begin{proof}[Proof of Theorem \ref{theorem:overestimationk}]
First of all, we show that it is enough to prove the result for $\mu\equiv 0$ and $\sigma^2\equiv 1$ and hence $K=0$. We have
\begin{equation*}
\begin{split}
&\sup_{(\mu,\sigma^2)\in\mathcal{S}}\Pj_{(\mu,\sigma^2)}\left(\hat{K}>K+2k\right)\\
= & \sup_{(\mu,\sigma^2)\in\mathcal{S}}\Pj_{(\mu,\sigma^2)}\left(\max_{\left[\frac{i}{n},\frac{j}{n}\right]\in\mathcal{D}(\tilde{\mu})}\left[T_i^j\left(Y, \tilde{\mu}([i/n, j/n])\right)-q_{ij}\right]> 0\ \forall\ \tilde{\mu}\in\mathcal{M}\ s.t.\ \vert \mathcal{I}(\tilde{\mu})\vert \leq K + 2k\right)\\
\leq & \sup_{(\mu,\sigma^2)\in\mathcal{S}}\Pj_{(\mu,\sigma^2)}\Bigg(\max_{\left[\frac{i}{n},\frac{j}{n}\right]\in\mathcal{D}(\tilde{\mu})}\left[T_i^j\left(Y, \tilde{\mu}([i/n, j/n])\right)-q_{ij}\right]> 0\\
&\quad\quad\quad\quad\quad\quad\quad\quad\quad\quad\quad \forall\ \tilde{\mu}\in\mathcal{M}\ s.t.\ \mathcal{I}(\mu)\subset \mathcal{I}(\tilde{\mu}), \vert \mathcal{I}(\tilde{\mu})\vert \leq K + 2k\Bigg)\\
\leq & \Pj_{(0,1)}\left(\max_{\left[\frac{i}{n},\frac{j}{n}\right]\in\mathcal{D}(\tilde{\mu})}\left[T_i^j\left(Y, \tilde{\mu}([i/n, j/n])\right)-q_{ij}\right]> 0\ \forall\ \tilde{\mu}\in\mathcal{M}\ s.t.\ \vert \mathcal{I}(\tilde{\mu})\vert \leq 2k\right)\\
= &\Pj_{(0, 1)}\left(\hat{K}>2k\right),
\end{split} 
\end{equation*}
where the last inequality follows from the same argument as in the proof of Theorem \ref{theorem:boundoverestimate}.
Now, we define $R_0:=0$ and iteratively
\begin{equation*}
R_{k+1}:=\min\{t>R_k:\exists\ s\ s.t.\ R_k<s<t \text{ and } [s/n,t/n]\in\mathcal{D},  T_s^t(Y,0)>q_{\log_2(t-s+1)}\},
\end{equation*}
with the convention $\min \emptyset = \infty$. Then,
\begin{equation*}
\Pj_{0,1}(R_{k+1}\leq n\vert R_1 = t)\leq \Pj_{0,1}(R_k\leq n) \text{ for all }t\in\{1,\ldots,n\},
\end{equation*}
since for the l.h.s. the remaining $k$ rejections $R_2,\ldots,R_{k+1}$ have to be in $\{t+1,\ldots,n\}$ instead of $\{1,\ldots,n\}$. It follows
\begin{equation*}
\begin{split}
&\Pj_{0,1}(\hat{K}>2k)\leq \Pj_{0,1}(R_{k+1}\leq n)=\sum_{t=1}^{n}{\Pj_{0,1}(R_{k+1}\leq n \vert R_1 = t)\Pj_{0,1}(R_1 = t)}\\
\leq & \Pj_{0,1}(R_1\leq n)\Pj_{0,1}(R_k \leq n)\leq \cdots\leq \Pj_{0,1}(R_1\leq n)^{k+1} \leq \alpha^{k+1},
\end{split} 
\end{equation*}
where the last inequality is given by Theorem \ref{theorem:boundoverestimate}. It follows
\begin{equation*}
\begin{split}
\sup_{(\mu,\sigma^2)\in\mathcal{S}}\E_{(\mu,\sigma^2)}\left[(\hat{K}-K)_+\right]
= & \sup_{(\mu,\sigma^2)\in\mathcal{S}}\sum_{k=0}^{\infty}{\Pj_{(\mu,\sigma^2)}\left(\hat{K}-K>k\right)}\\
\leq & \sup_{(\mu,\sigma^2)\in\mathcal{S}}2\sum_{k=0}^{\infty}{\Pj_{(\mu,\sigma^2)}\left(\hat{K}-K>2k\right)}
\leq 2\sum_{k=0}^{\infty}{\alpha^{k+1}}
=\frac{2\alpha}{1-\alpha}.
\end{split}
\end{equation*}
\end{proof}
The following theorem is sharper version of \ref{theorem:boundunderestimate} that shows different probabilities for the detection of the change-points.
\begin{Theorem}[Underestimation control II]\label{theorem:boundunderestimatesharp}
Let $\lambda_j:=\tau_{j+1}-\tau_{j}$ and $k_{n, j}:=\lfloor\log_2(n\lambda_j/4)\rfloor$, $j=0,\ldots,K$, as well as $\delta_j:=\vert \valmu_j - \valmu_{j-1}\vert$ and 
\begin{equation*}
\begin{split}
\eta_j:=&\left[1-3\exp\left(-\frac{1}{48}\left(\sqrt{\frac{n\lambda_{j-1}\delta_j^2}{32\sigma_{j-1}^ 2}}-\sqrt{16\log\left(\frac{8}{\lambda_j\alpha\beta_{k_{n, j-1}}}\right)}\right)_+^2\right)\right]_+\ \times\\
&\left[1-3\exp\left(-\frac{1}{48}\left(\sqrt{\frac{n\lambda_j\delta_j^2}{32\sigma_j^ 2}}-\sqrt{16\log\left(\frac{8}{\lambda_j\alpha\beta_{k_{n,j}}}\right)}\right)_+^2\right)\right]_+,
\end{split}
\end{equation*}
$j=1,\ldots,K$. Under the assumptions of Theorem \ref{theorem:boundoverestimate} and if $n\lambda_j \geq 32$ and
\begin{equation*}
\left(n\lambda_j\right)^{-1}\log\left(\frac{8}{\lambda_j\alpha\beta_{k_{n,j}}}\right)\leq \frac{1}{512}
\end{equation*}
are satisfied for all $j=1,\ldots,K$, then 
\begin{equation*}
\Pj_{(\mu, \sigma^2)}\left(\hat{K} < K\right)\leq 1-\prod_{j=1}^K\eta_j \text{ and }\E_{(\mu, \sigma^2)}\left[\left(K-\hat{K}\right)_+\right]\leq \sum_{j=1}^K\left(1-\eta_j\right).
\end{equation*}
\end{Theorem}
\begin{proof}
For each $j=1,\ldots,K$ we consider the disjoint intervals $I_j:=\left[\tau_j-\lambda_{j-1}/2,\tau_j+\lambda_j/2\right)$ and split them into disjoint intervals $I_j^+ \cup I_j^-=I_j$ such that $\mu(t)=\mu^+\ \forall\ t\in I_j^+$ and $\mu(t)=\mu^-\ \forall\ t\in I_j^-$, with $\mu^+:=\max(\valmu_{j-1},\valmu_j)$ and $\mu^-:=\min(\valmu_{j-1},\valmu_j)$. Without loss of generality we assume $\mu^+=\valmu_{j-1}$ and $\mu^-=\valmu_{j}$ in the following. Then, there exists subintervals $J_j^+\subset I_j^+$ and $J_j^-\subset I_j^-$ with $J_j^+,J_j^-\in \mathcal{D}$ that have length $\lambda_{j-1}^{*}:=n^{-1}2^{\lfloor \log_2(n\lambda_{j-1}/4)\rfloor}=n^{-1}2^{k_{n,j-1}}\geq \lambda_{j-1}/8$, since $n\vert I_j^+\vert = n\lambda_{j-1}/2 \geq 3$, and $\lambda_j^{*}:=n^{-1}2^{\lfloor \log_2(n\lambda_j/4)\rfloor}=n^{-1}2^{k_{n,j}}\geq \lambda_j/8$, since $n\vert I_j^-\vert = n\lambda_j/2 \geq 3$, respectively. It follows
\begin{align*}
&\Pj_{(\mu,\sigma^2)}\left(\hat{K} < K\right)=1-\Pj_{(\mu,\sigma^2)}\left(\hat{K} \geq K\right)\\
& \leq 1-\Pj_{(\mu,\sigma^2)}\left(\nexists\ \hat{\mu}\in C(\textbf{q}),\ j\in\{1,\ldots,K\}:\hat{\mu} \text{ is constant on } I_j\right)\\
&
\begin{aligned}
\leq 1-\Pj_{(\mu,\sigma^2)}\Big(\forall\ j\in\{1,\ldots,K\}:\quad &\nexists\ \hat{\valmu}\leq (\valmu_{j-1}+\valmu_j)/2: T_{J_j^+}(Y,\hat{\valmu})\leq q_{k_{n,j-1}} \text{ and }\\
& \nexists\ \hat{\valmu}\geq (\valmu_{j-1}+\valmu_j)/2: T_{J_j^-}(Y,\hat{\valmu})\leq q_{k_{n,j}}\Big) 
\end{aligned}\\
&
\begin{aligned}
\leq 1-\prod_{j=1}^K\Pj_{(\mu,\sigma^2)}\Big(&\nexists\ \hat{\valmu}\leq (\valmu_{j-1}+\valmu_j)/2: T_{J_j^+}(Y,\hat{\valmu})\leq q_{k_{n,j-1}} \text{ and }\\
& \nexists\ \hat{\valmu}\geq (\valmu_{j-1}+\valmu_j)/2: T_{J_j^-}(Y,\hat{\valmu})\leq q_{k_{n,j}}\Big), 
\end{aligned}
\end{align*}
where we used in the last inequality that the events are independent, since all intervals are disjoint. We denote by $Z_1,\ldots,Z_n$ i.i.d. standard normally distributed random variables. It follows from once again from the independence due to disjoint intervals and from the Lemmas 7.1 in \citep{SMUCE}, \ref{lemma:boundcdfstatfurther} and \ref{lemma:boundcritval} that
\begin{align*}
&
\begin{aligned}
\Pj_{(\mu,\sigma^2)}\Big(&\nexists\ \hat{\valmu}\leq (\valmu_{j-1}+\valmu_j)/2: T_{J_j^+}(Y,\hat{\valmu})\leq q_{k_{n,j-1}} \text{ and }\\
& \nexists\ \hat{\valmu}\geq (\valmu_{j-1}+\valmu_j)/2: T_{J_j^-}(Y,\hat{\valmu})\leq q_{k_{n,j}}\Big)
\end{aligned}\\
&
\begin{aligned}
\geq &\left[1-\Pj_{(\mu,\sigma^2)}\Big(\exists\ \hat{\valmu}\leq (\valmu_{j-1}+\valmu_j)/2: T_{J_j^+}(Y,\hat{\valmu})\leq q_{k_{n,j-1}}\Big)\right]\times\\
&\left[1-\Pj_{(\mu,\sigma^2)}\Big(\exists\ \hat{\valmu}\geq (\valmu_{j-1}+\valmu_j)/2: T_{J_j^-}(Y,\hat{\valmu})\leq q_{k_{n,j}}\Big) \right]\geq \eta_j,
\end{aligned}
\end{align*}
since 
\begin{align*}
&\Pj_{(\mu,\sigma^2)}\Big(\exists\ \hat{\valmu}\leq (\valmu_{j-1}+\valmu_j)/2: T_{J_j^+}(Y,\hat{\valmu})\leq q_{k_{n,j-1}}\Big)\\
&\leq \Pj_{(\mu,\sigma^2)}\Big(\overline{Y}_{J_j^+}\leq (\valmu_{j-1}+\valmu_j)/2 \text{ or } T_{J_j^+}(Y,(\valmu_{j-1}+\valmu_j)/2)\leq q_{k_{n,j-1}}\Big)\\
&\leq \Pj_{(\mu,\sigma^2)}\Big(\overline{Y}_{J_j^+}\leq (\valmu_{j-1}+\valmu_j)/2\Big)+\Pj_{(\mu,\sigma^2)}\Big(T_{J_j^+}(Y,(\valmu_{j-1}+\valmu_j)/2)\leq q_{k_{n,j-1}}\Big)\\
&\leq \Pj\Big(\overline{Z}_{[0,\lambda_{j-1}^*]}\geq \frac{\delta_j}{2\sigma_{j-1}}\Big) + \Pj\Big(T_{[0,\lambda_{j-1}^*]}\left(Z,\frac{\delta_j}{2\sigma_{j-1}}\right)\leq q_{k_{n,j-1}} \Big)\\
&\leq \exp\left(-\frac{1}{64}\frac{n\lambda_{j-1}\delta_j^2}{\sigma_{j-1}^2}\right)+2\exp\left(-\frac{1}{48}\left(\sqrt{\frac{n\lambda_{j-1}\delta_j^2}{32\sigma_{j-1}^2}}-\sqrt{2q_{k_{n,j-1}}}\right)_+^2\right)\\
&\leq 3\exp\left(-\frac{1}{48}\left(\sqrt{\frac{n\lambda_{j-1}\delta_j^2}{32\sigma_{j-1}}}-\sqrt{16\log\left(\frac{8}{\lambda_{j-1}\alpha\beta_{k_{n,j-1}}}\right)}\right)_+^2\right)
\end{align*}
and the second term by symmetry arguments. Moreover, it follows
\begin{align*}
&\E_{(\mu,\sigma^2)}\left[\left(K-\hat{K}\right)_+\right]\\
\leq & \E_{(\mu,\sigma^2)}\left[\sum_{j=1}^K{\EINS_{\exists\, \hat{\valmu}\leq (\valmu_{j-1}+\valmu_j)/2:\, T_{J_j^+}(Y,\hat{\valmu})\leq q_{k_{n,j-1}} \text{ or } \exists\, \hat{\valmu}\geq (\valmu_{j-1}+\valmu_j)/2:\, T_{J_j^-}(Y,\hat{\valmu})\leq q_{k_{n,j}}}}\right]\\
\leq &\sum_{j=1}^K\left(1-\eta_j\right).
\end{align*}
\end{proof}
\begin{proof}[Proof of Theorem \ref{theorem:boundunderestimate}]
The proof is analogue to the proof of Theorem \ref{theorem:boundunderestimatesharp}, but with $I_j=[\tau_j-\lambda/2,\tau_j+\lambda/2)$.
\end{proof}
\begin{proof}[Proof of Theorem \ref{theorem:consistency}]
We prove the theorem with the Borel-Cantelli lemma. It follows from Theorems \ref{theorem:boundoverestimate} and \ref{theorem:boundunderestimate} that
\begin{align*}
&\sup_{(\mu,\sigma^2)\in \mathcal{S}_{\Delta,\lambda}}\Pj_{(\mu,\sigma^2)}\left(\hat{K}_n\neq K\right)\\
= & \sup_{(\mu,\sigma^2)\in \mathcal{S}_{\Delta,\lambda}}\Pj_{(\mu,\sigma^2)}\left(\hat{K}_n > K\right)+\sup_{(\mu,\sigma^2)\in \mathcal{S}_{\Delta,\lambda}}\Pj_{(\mu,\sigma^2)}\left(\hat{K}_n < K\right)\\
\leq &  \alpha_n+1-\left[1-3\exp\left(-\frac{1}{48}\left(\sqrt{\frac{n\lambda\Delta^2}{32}}-\sqrt{16\log\left(\frac{8}{\lambda\alpha_n\beta_{k_n,n}}\right)}\right)_+^2\right)\right]_+^{2K}\\
\leq &  \alpha_n+6K\exp\left(-\frac{1}{48}\left(\sqrt{\frac{n\lambda\Delta^2}{32}}-\sqrt{16\log\left(\frac{8}{\lambda\alpha_n\beta_{k_n,n}}\right)}\right)_+^2\right),
\end{align*}
since under the given assumptions the conditions of Theorem \ref{theorem:boundunderestimate} are satisfied. The upper bounds for the error probabilities are summable if \eqref{eq:conditionconsistency} is satisfied.
\end{proof}
\begin{Lemma}[Confidence set]\label{lemma:confidencefunction}
Assume the setting and assumptions of Theorem \ref{theorem:boundunderestimate} and let $C(\textbf{q})$ be as in \eqref{eq:defC} with significance level $\alpha$ and weights $\beta_1,\ldots,\beta_{\sn}$. Let $\mathcal{S}_{\Delta,\lambda}$ be as in \eqref{eq:subsetSdeltalambda} with $\Delta,\lambda >0$ arbitrary, but fixed, and $k_n:=\lfloor\log_2(n\lambda/4)\rfloor$. If $n\lambda\geq 32$ and 
\begin{equation*}
\frac{\log\left(\frac{8}{\lambda\alpha_n\beta_{k_n}}\right)}{n\lambda}\leq \frac{1}{512}
\end{equation*}
hold, then uniformly in $\mathcal{S}_{\Delta,\lambda}$
\begin{equation*}
\Pj_{(\mu, \sigma^2)}\left(\mu\in C\left(\textbf{q}\right)\right)\geq 
1-\alpha-(1-\eta^K),
\end{equation*}
with $\eta$ like in Theorem \ref{theorem:boundunderestimate}.
\end{Lemma}
\begin{proof}
It follows from the definition of $C\left(\textbf{q}\right)$ in \eqref{eq:defC} as well as from Theorems \ref{theorem:boundoverestimate} and \ref{theorem:boundunderestimate} that
\begin{equation*}
\begin{split}
&\inf_{(\mu,\sigma^2)\in \mathcal{S}_{\Delta,\lambda}}\Pj_{(\mu,\sigma^2)}\left(\mu\in C\left(\textbf{q}\right)\right)\\
= & \inf_{(\mu,\sigma^2)\in \mathcal{S}_{\Delta,\lambda}}\Pj_{(\mu,\sigma^2)}\left(\max_{[\frac{i}{n},\frac{j}{n}]\in \mathcal{D}(\mu)} \left[T_i^j\big(Y,\mu([i/n,j/n])\big) - q_{ij}\right] \leq 0, \hat{K}=K\right)\\
= &\inf_{(\mu,\sigma^2)\in \mathcal{S}_{\Delta,\lambda}}\Pj_{(\mu,\sigma^2)}\left(\max_{[\frac{i}{n},\frac{j}{n}]\in \mathcal{D}(\mu)} \left[T_i^j\big(Y,\mu([i/n,j/n])\big) - q_{ij}\right] \leq 0, \hat{K}\geq K\right)\\
\geq & \inf_{(\mu,\sigma^2)\in \mathcal{S}_{\Delta,\lambda}}\Pj_{(\mu,\sigma^2)}\left(\max_{[\frac{i}{n},\frac{j}{n}]\in \mathcal{D}(\mu)} \left[T_i^j\big(Y,\mu([i/n,j/n])\big) - q_{ij}\right] \leq 0\right)-\sup_{(\mu,\sigma^2)\in \mathcal{S}_{\Delta,\lambda}}\Pj_{(\mu,\sigma^2)}\left(\hat{K} < K\right)\\
\geq & 1-\alpha-(1-\eta^K).
\end{split}
\end{equation*}
\end{proof}
\begin{proof}[Proof of Theorem \ref{theorem:confidencefunction}]
The statement is a direct consequence of Lemma \ref{lemma:confidencefunction}.
\end{proof}
\begin{Lemma}[Change-point locations]\label{lemma:confidencelocation}
Assume the setting of Lemma \ref{lemma:confidencefunction}. If $c_n$ is a sequence with $0 < c_n\leq \lambda/2$ and $k_n:=\lfloor\log_2(nc_n/2)\rfloor$ such that $nc_n\geq 16$ and 
\begin{equation}\label{eq:condconfidencelocation}
\frac{\log\left(\frac{4}{c_n\alpha\beta_{k_n}}\right)}{nc_n}\leq \frac{1}{256}
\end{equation}
hold, then uniformly in $\mathcal{S}_{\Delta,\lambda}$
\begin{equation*}
\begin{split}
&\Pj_{(\mu, \sigma^2)}\left(\sup_{\hat{\mu}\in C(\textbf{q}_n)}\max_{\tau\in\mathcal{I}(\mu)}\min_{\hat{\tau}\in\mathcal{I}(\hat{\mu})}\vert \hat{\tau}-\tau\vert > c_n\right)\\
&\leq  1-\left[1-3\exp\left(-\frac{1}{48}\left(\sqrt{\frac{nc_n\Delta^2}{16}}-\sqrt{16\log\left(\frac{4}{c_n\alpha\beta_{k_n}}\right)}\right)_+^2\right)\right]_+^{2K}.
\end{split}
\end{equation*}
\end{Lemma}
\begin{proof}
Analogously to the proof of Theorem \ref{theorem:boundunderestimatesharp} we have
\begin{equation*}
\begin{split}
&\sup_{(\mu,\sigma^2)\in \mathcal{S}_{\Delta,\lambda}}\Pj_{(\mu,\sigma^2)}\left(\sup_{\hat{\mu}\in C(\textbf{q}_n)}\max_{\tau\in\mathcal{I}(\mu)}\min_{\hat{\tau}\in\mathcal{I}(\hat{\mu})}\vert \hat{\tau}-\tau\vert > c_n\right)\\
\leq &\sup_{(\mu,\sigma^2)\in \mathcal{S}_{\Delta,\lambda}}\Pj_{(\mu,\sigma^2)}\big(\exists\ j\in\{1,\ldots,K\}\text{ and }\hat{\mu}\in C(\textbf{q}_n) : \hat{\mu} \text{ is constant on } [\tau_j-c_n,\tau_j+c_n)\big)\\
\leq &1-\left[1-3\exp\left(-\frac{1}{48}\left(\sqrt{\frac{nc_n\Delta^2}{16}}-\sqrt{16\log\left(\frac{4}{c_n\alpha\beta_{k_n}}\right)}\right)_+^2\right)\right]_+^{2K}.
\end{split}
\end{equation*}
\end{proof}
\begin{proof}[Proof of Theorem \ref{theorem:confidencelocation}]
For $n$ large enough such that \eqref{eq:conditioncnvanishing} guarantees the assumption of Lemma \ref{lemma:confidencelocation} it follows
\begin{align*}
&\Pj_{(\mu,\sigma^2)}\left(\sup_{\hat{\mu}\in C(\textbf{q}_n)}{\max_{j=1,\ldots,K}c_n^{-1}\left\vert \tau_j - \hat{\tau}_j\right\vert} > 1\right)\\
&\leq  \Pj_{(\mu,\sigma^2)}\left(\hat{K}>K \text{ or } \exists\ \hat{\mu}\in C(\textbf{q}_n),\ j\in\{1,\ldots,K\} : \hat{\mu} \text{ is constant on } [\tau_j-c_n,\tau_j+c_n]\right)\\
&\leq\Pj_{(\mu,\sigma^2)}\left(\hat{K}>K\right)+\Pj_{(\mu,\sigma^2)}\Big(\exists\ \hat{\mu}\in C(\textbf{q}_n),\ j\in\{1,\ldots,K\} : \hat{\mu} \text{ is constant on } [\tau_j-c_n,\tau_j+c_n]\Big)\\
&\leq  \alpha_n+\left(1-\left[1-3\exp\left(-\frac{1}{48}\left(\sqrt{\frac{nc_n\Delta^2}{16}}-\sqrt{16\log\left(\frac{4}{c_n\alpha_n\beta_{k_n,n}}\right)}\right)_+^2\right)\right]_+^{2K}\right).
\end{align*}
The assertion follows from $\alpha_n\to 0$ and 
\begin{equation*}
\lim_{n\rightarrow \infty}{\sqrt{\frac{nc_n\Delta^2}{16}}-\sqrt{16\log\left(\frac{4}{c_n\alpha_n\beta_{k_n,n}}\right)}}=\infty,
\end{equation*}
whereby latter one is direct consequence of \eqref{eq:conditioncnvanishing}.
\end{proof}
The following theorem deals with the detection of single vanishing bump against a noisy background.
\begin{Theorem}[Single vanishing bump]\label{theorem:onesignalrate}
Assume the heterogeneous gaussian change-point model \eqref{eq:model} with sequences of bump signals $\mu_n(t):=\valmu_0+\delta_n\EINS_{I_n}(t)$ and $\sigma_n(t):=\EINS_{I_n^C}(t)+\valsd_n\EINS_{I_n}(t)$, where $\delta_n\neq 0$ is a sequence of change-point sizes, $\valsd_n>0$ a sequence of standard deviations on $I_n\in\mathcal{D}$, which is a sequence of intervals with $\vert I_n\vert \to 0$. Let $k_n:=\lfloor\log_2(n\vert I_n\vert)\rfloor$ and $\Delta_n:=\vert \delta_n\vert/\valsd_n$ be the sequence of the signal to noise ratios. Let $(\hat{K}_n)_n$, $\alpha_n$ and $\beta_{1,n},\ldots,\beta_{\snn,n}$ be as in Theorem \ref{theorem:multiscalesignalsrate}. We further assume
\begin{equation}\label{eq:onesignalrate}
\sqrt{n\vert I_n\vert}\Delta_n\geq (4+\epsilon_n)\sqrt{-\log(\vert I_n\vert)},
\end{equation}
with possibly $\epsilon_n\to 0$, but such that $\epsilon_n\sqrt{-\log(\vert I_n\vert)}\rightarrow \infty$ and
\begin{equation*}
\limsup_{n\to\infty}\frac{\sqrt{-\log\left(\alpha_n\beta_{k_n,n}\right)}}{\epsilon_n\sqrt{-\log(\vert I_n\vert)}}< \frac{1}{4},
\end{equation*}
\begin{equation}\label{eq:onesignalrateIn}
\liminf_{n\to \infty}{\frac{n\vert I_n\vert}{\log(n)}} > 64 \text{ and } \lim_{n\to\infty}{\frac{\log\left(\alpha_n\beta_{k_n,n}\right)}{n\vert I_n\vert}}=0,
\end{equation}
\begin{equation}\label{eq:onesignalratecond3}
\lim_{n\rightarrow \infty}{\valsd_n\frac{\sqrt{\vert I_n^C\vert}}{\sqrt{\vert I_n\vert}}}=\infty \text{ and}
\end{equation}
\begin{equation}\label{eq:onesignalratecond4}
\liminf_{n\rightarrow \infty}{\frac{\log(\beta_{k_n,n})}{\log(\beta_{\min,n})}}>0 \text{, with } \beta_{\min,n}:=\min\{\beta_{1,n},\ldots,\beta_{\snn,n}\}.
\end{equation}
Then,
\begin{equation}\label{eq:detectionsinglesignal}
\lim_{n\rightarrow \infty}{\Pj_{(\mu_n,\sigma^2_n)}\left(\hat{K}_n > 0\right)}=1.
\end{equation}
\end{Theorem}
Conditions \eqref{eq:onesignalrate} and \eqref{eq:onesignalrateIn} are the main assumptions of the theorem to detect the vanishing signal on $I_n$. We discussed them together with the conditions of Theorem \ref{theorem:multiscalesignalsrate} in Section \ref{sec:rates}. We also need the weak technical conditions \eqref{eq:onesignalratecond3} and \eqref{eq:onesignalratecond4} on the size of $\vert I_n^C\vert$ and the minimal weight $\beta_{\min,n}$ to ensure that the detection power on the complement $I_n^C$ is large enough, too. Condition \eqref{eq:onesignalratecond4} is for instance fulfilled by uniform weights $\beta_{1,n}=\cdots=\beta_{\snn,n}=1/\snn$, but many other choices are possible, too. We further assumed $I_n\in\mathcal{D}$, otherwise we have to replace $I_n$ by the largest subinterval which is an element of the dyadic partition. Such an interval exists always and has at least length $n^{-1}2^{\lfloor\log_2(n\vert I_n\vert/2)\rfloor}>\vert I_n\vert/4$. Therefore, omitting the condition $I_n\in\mathcal{D}$ would not change the rate. It is possible to strengthen \eqref{eq:detectionsinglesignal} further to $\lim_{n\rightarrow \infty}{\Pj_{(\mu_n, \sigma_n^2)}(\hat{K}_n \geq K)}=1$ if we increase all constants a little bit.
\begin{proof}[Proof of Theorem \ref{theorem:onesignalrate}]
We denote by $J_n$ the longest subinterval $J_n\subset I_n^ C$ which is part of the dyadic partition. Such an interval exists (at least for $n$ large enough) always, since $\vert I_n\vert\to 0$, and has at least length $\vert I_n^ C \vert / 8$. Moreover, let $k_n:=\log_2(n\vert I_n\vert)$ and $l_n:=\log_2(n\vert J_n\vert)$. Then, the Lemmas 7.1 in \citep{SMUCE} and \ref{lemma:boundcdfstatfurther} yield for any $\theta_n>0$
\begin{align*}
&\lim_{n\rightarrow \infty}{\Pj_{(\mu_n,\sigma^2_n)}\left(\hat{K}_n > 0\right)}\\
&=  \lim_{n\rightarrow \infty}{1 - \Pj_{(\mu_n,\sigma^2_n)}\left(\hat{\mu} \text{ is constant}\right)}\\
&\geq  \lim_{n\rightarrow \infty}{1 - \Pj_{(\mu_n,\sigma^2_n)}\left(\exists\ \hat{\valmu} \leq \valmu_0 + \theta_n : T_{I_n}(Y,\hat{\valmu}) \leq q_{k_n} \text{ or } \exists\ \hat{\valmu} \geq \valmu_0 + \theta_n : T_{J_n}(Y,\hat{\valmu}) \leq q_{l_n}\right)}\\
&
\begin{aligned}
\geq  \lim_{n\rightarrow \infty}1 &- \Pj_{(\mu_n,\sigma^2_n)}\left(\exists\ \hat{\valmu} \leq \valmu_0 + \theta_n : T_{I_n}(Y,\hat{\valmu}) \leq q_{k_n}\right)\\
& - \Pj_{(\mu_n,\sigma^2_n)}\left(\exists\ \hat{\valmu} \geq \valmu_0 + \theta_n : T_{J_n}(Y,\hat{\valmu}) \leq q_{l_n}\right)
\end{aligned}\\
&
\begin{aligned}
\geq  \lim_{n\rightarrow \infty}1 &- \Pj_{(\mu_n,\sigma^2_n)}\left(\overline{Y}_{I_n}\leq \valmu_0 + \theta_n\right) - \Pj_{(\mu_n,\sigma^2_n)}\left(T_{I_n}(Y,\valmu_0 + \theta_n) \leq q_{k_n}\right)\\
& - \Pj_{(\mu_n,\sigma^2_n)}\left(\overline{Y}_{J_n}\geq \valmu_0 + \theta_n\right) - \Pj_{(\mu_n,\sigma^2_n)}\left(T_{J_n}(Y,\valmu_0 + \theta_n) \leq q_{l_n}\right)
\end{aligned}\\
&\geq  \lim_{n\rightarrow \infty}{1 - 2\Pj_{(\mu_n,\sigma^2_n)}\left(T_{I_n}(Y,\valmu_0 + \theta_n) \leq q_{k_n}\right) -2\Pj_{(\mu_n,\sigma^2_n)}\left(T_{J_n}(Y,\valmu_0 + \theta_n) \leq q_{l_n}\right)}\\
&\geq  \lim_{n\rightarrow \infty}{1 - 4\exp\left(-\frac{1}{48}\left(\Gamma_{I_n}\right)_+^2\right) -4\exp\left(-\frac{1}{48}\left(\Gamma_{J_n}\right)_+^2\right)}
= 1,
\end{align*}
if
\begin{equation*}
\Gamma_{I_n}:=\sqrt{n\vert I_n \vert} \frac{\delta_n-\theta_n}{s_n}-\sqrt{2q_{k_n}} \to \infty \text{ and }
\Gamma_{J_n}:=\sqrt{n\vert J_n \vert }\theta_n-\sqrt{2q_{l_n}} \to \infty,
\end{equation*}
and if the conditions of Lemma \ref{lemma:boundcdfstatfurther} are satisfied. This is the case, since $n\vert I_n\vert\to\infty$ and $n\vert J_n\vert\to\infty$, because of \eqref{eq:onesignalrateIn} and $\vert I_n\vert \to 0$, as well as $q_{k_n}/(n\vert I_n\vert)\leq 1/8$ and $q_{l_n}/(n\vert J_n\vert)\leq 1/8$ hold at least for $n$ large enough:
The first one is a direct consequence of Lemma \ref{lemma:boundcritval} and \eqref{eq:onesignalrateIn}
\begin{equation*}
\frac{q_{k_n}}{n\vert I_n\vert}\leq \frac{8\log\left(\frac{1}{\vert I_n\vert \alpha_n\beta_{k_n,n}}\right)}{n\vert I_n\vert}\leq \frac{1}{8},
\end{equation*}
since then the assumptions of Lemma \ref{lemma:boundcritval} are also fulfilled. The second inequality follows from Lemma \ref{lemma:boundcritval}, \eqref{eq:onesignalrateIn} and \eqref{eq:onesignalratecond4} as well as the fact that $\vert I_n\vert / \vert J_n\vert \to 0$
\begin{equation*}
\lim_{n\to\infty}{\frac{q_{l_n}}{n\vert J_n\vert}}\leq \lim_{n\to\infty}{\frac{8\log\left(\frac{1}{\vert J_n\vert \alpha_n\beta_{l_n,n}}\right)}{n\vert J_n\vert}}\leq \lim_{n\to\infty}{\frac{8\log\left(\frac{1}{\vert J_n\vert \alpha_n\beta_{l_n,n}}\right)}{8\log\left(\frac{1}{\vert I_n\vert \alpha_n\beta_{k_n,n}}\right)}\frac{\vert I_n\vert}{\vert J_n\vert}\frac{8\log\left(\frac{1}{\vert I_n\vert \alpha_n\beta_{k_n,n}}\right)}{n\vert I_n\vert}}\to 0,
\end{equation*}
since then the assumptions of Lemma \ref{lemma:boundcritval} are also fulfilled.\\
We define now $\theta_n=\sqrt{\gamma_n / n}$ via the equation
\begin{equation*}
\sqrt{\frac{\gamma_n\vert I_n \vert}{s_n^2}} = c\epsilon_n \sqrt{\log\left(\frac{1}{\vert I_n\vert}\right)}
\end{equation*}
for $0< c < 1$. Then, it follows from Lemma \ref{lemma:boundcritval} and from $\sqrt{x+y}\leq \sqrt{x}+\sqrt{y}$ for $x,y>0$ together with the assumptions of the theorem that
\begin{equation*}
\begin{split}
\Gamma_{I_n}=&\sqrt{n\vert I_n \vert} \frac{\delta_n-\theta_n}{s_n}-\sqrt{2q_{k_n}}\\
= &\sqrt{n\vert I_n\vert \delta_n^2/s_n^2} - \sqrt{\gamma_n \vert I_n\vert / s_n^2} - \sqrt{2q_{k_n}}\\
\geq & \sqrt{n\vert I_n\vert \Delta_n^2} - \sqrt{\gamma_n \vert I_n\vert / s_n^2} - \sqrt{16\log\left(\frac{1}{\vert I_n\vert \alpha_n\beta_{k_n,n}}\right)}\\
\geq & (4+\epsilon_n)\sqrt{\log\left(\frac{1}{\vert I_n\vert}\right)} - c\epsilon_n\sqrt{\log\left(\frac{1}{\vert I_n\vert}\right)} - 4\sqrt{\log\left(\frac{1}{\vert I_n\vert}\right)} - 4 \sqrt{\log\left(\frac{1}{\alpha_n\beta_{k_n,n}}\right)}\\
\geq & (1-c)\epsilon_n \sqrt{\log\left(\frac{1}{\vert I_n\vert}\right)} - 4 \sqrt{\log\left(\frac{1}{\alpha_n\beta_{k_n,n}}\right)}
\to \infty,
\end{split}
\end{equation*}
since the conditions of Lemma \ref{lemma:boundcritval} are satisfied, as shown above.\\
Moreover, we have $\Gamma_{J_n}:=\sqrt{n\vert J_n \vert}\theta_n-\sqrt{2q_{l_n}}=\sqrt{\vert J_n \vert \gamma_n}-\sqrt{2q_{l_n}} \to \infty$ if
\begin{equation*}
\sqrt{\frac{\vert J_n\vert \gamma_n}{2q_{l_n}}} \geq \sqrt{\frac{\vert J_n\vert \gamma_n}{16\log\left(\frac{1}{\vert J_n\vert \alpha_n\beta_{l_n,n}}\right)}} \to \infty,
\end{equation*}
where we used Lemma \ref{lemma:boundcritval} again. Finally, it follows from the assumptions of the theorem that
$\liminf_{n\to\infty}{\vert J_n\vert}\geq \liminf_{n\to\infty}{\vert I_n^C\vert/8} > 0$ and thus
\begin{equation*}
\begin{split}
\sqrt{\frac{\vert J_n\vert \gamma_n}{\log\left(\frac{1}{\alpha_n\beta_{l_n,n}}\right)}}
= & \frac{\sqrt{\vert I_n\vert \gamma_n}}{s_n\sqrt{\log\left(\frac{1}{\alpha_n\beta_{k_n,n}}\right)}}\frac{s_n\sqrt{\vert J_n\vert}}{\sqrt{\vert I_n \vert}}\frac{\sqrt{\log\left(\frac{1}{\alpha_n\beta_{k_n,n}}\right)}}{\sqrt{\log\left(\frac{1}{\alpha_n\beta_{l_n,n}}\right)}}\\
\geq  & \frac{c\epsilon_n \sqrt{\log\left(\frac{1}{\vert I_n\vert}\right)}}{\sqrt{\log\left(\frac{1}{\alpha_n\beta_{k_n,n}}\right)}}\frac{s_n\sqrt{\vert I_n^C\vert / 8}}{\sqrt{\vert I_n \vert}}\frac{\sqrt{\log\left(\frac{1}{\alpha_n\beta_{k_n,n}}\right)}}{\sqrt{\log\left(\frac{1}{\alpha_n\beta_{\min,n}}\right)}}\to \infty.
\end{split}
\end{equation*}
\end{proof}
\begin{proof}[Proof of Theorem \ref{theorem:multiscalesignalsrate}]
It follows from Theorem \ref{theorem:boundunderestimate} that
\begin{equation*}
\Pj_{(\mu_n,\sigma^2_n)}\left(\hat{K}_n < K_n\right)
\leq 1-\left[1-3\exp\left(-\frac{1}{48}\left(\Gamma_n\right)_+^2\right)\right]_+^{2K_n}
\leq 6K_n\exp\left(-\frac{1}{48}\left(\Gamma_n\right)_+^2\right),
\end{equation*}
with 
\begin{equation*}
\Gamma_n:=\sqrt{\frac{n\lambda_n\Delta_n^2}{32}}-\sqrt{16\log\left(\frac{8}{\lambda_n\alpha_n\beta_{k_n,n}}\right)},
\end{equation*}
since the assumptions of Theorem \ref{theorem:boundunderestimate} are satisfied by \eqref{conditionmultiscalesignalsrate}.\\ 
In case (1) it is enough to show $\Gamma_n\to \infty$, because $K_n$ is bounded. Finally, $\Gamma_n\to \infty$ follows from 
\begin{equation*}
\frac{n\lambda_n\Delta_n^2}{\log\left(\frac{8}{\lambda_n\alpha_n\beta_{k_n,n}}\right)}\to\infty.
\end{equation*}
In case (2) for bounded $K_n$, $\Gamma_n\to\infty$ follows from
\begin{equation*}
\begin{split}
\Gamma_n= & \sqrt{\frac{n\lambda_n\Delta_n^2}{32}}-\sqrt{16\log\left(\frac{8}{\lambda_n\alpha_n\beta_{k_n,n}}\right)}\\
\geq & \left(\frac{\sqrt{512}}{\sqrt{32}}+\frac{\epsilon_n}{\sqrt{32}}\right)\sqrt{\log\left(\frac{1}{\lambda_n}\right)}-\sqrt{16}\sqrt{\log\left(\frac{1}{\lambda_n}\right)}-\sqrt{16}\sqrt{\log\left(\frac{8}{\alpha_n\beta_{k_n,n}}\right)}\\
= & \frac{1}{\sqrt{32}}\left(\epsilon_n\sqrt{\log\left(\frac{1}{\lambda_n}\right)}-\sqrt{512}\sqrt{\log\left(\frac{8}{\alpha_n\beta_{k_n,n}}\right)}\right) \to \infty.
\end{split}
\end{equation*}
For unbounded $K_n$ we have $K_n\leq 1/\lambda_n$. It follows
\begin{equation*}
\begin{split}
& K_n\exp\left(-\frac{1}{48}\left(\Gamma_n\right)_+^2\right)\\
\leq & \exp\left(\log\left(\frac{1}{\lambda_n}\right)-\frac{1}{48}\left(\frac{C}{\sqrt{32}}\sqrt{\log\left(\frac{1}{\lambda_n}\right)}+\frac{1}{\sqrt{32}}\epsilon_n\sqrt{\log\left(\frac{1}{\lambda_n}\right)}-\sqrt{16}\sqrt{\log\left(\frac{8}{\alpha_n\beta_{k_n,n}}\right)}\right)_+^2\right)\\
\leq &\exp\Bigg(-\frac{1}{48}\left(\frac{1}{\sqrt{32}}\epsilon_n\sqrt{\log\left(\frac{1}{\lambda_n}\right)}-\sqrt{16}\sqrt{\log\left(\frac{8}{\alpha_n\beta_{k_n,n}}\right)}\right)^2\Bigg)\to  0.
\end{split}
\end{equation*}
\end{proof}

\subsection{Proofs of Section \ref{sec:computation}}
\begin{proof}[Proof of Theorem \ref{theorem:convergencecritval}]
We prove the assertion with \citep[Theorem 5.9]{vanderVaart07} which states three conditions for the convergence of a Z-estimator. Note that the convergence in probability can be replaced by almost sure convergence, if the assumptions hold almost surely. We define 
\[\Psi\left(\theta\right):=\left\vert F\left(\theta\right)-(1-\alpha)\right\vert+\sum_{k=2}^{\sn}{\left\vert\frac{1-F_1\left(\theta_1\right)}{\beta_1}-\frac{1-F_k\left(\theta_k\right)}{\beta_k}\right\vert}\]
and
\[\Psi_M\left(\theta\right):=\left\vert F_M\left(\theta\right)-(1-\alpha)\right\vert+\sum_{k=2}^{\sn}{\left\vert\frac{1-F_{M,1}\left(\theta_1\right)}{\beta_1}-\frac{1-F_{M,k}\left(\theta_k\right)}{\beta_k}\right\vert}\]
as well as $\Theta:=[0,\infty)^\sn$, $\theta_0:=\textbf{q}$ and $\hat{\theta}_M:=\widehat{\textbf{q}}_M$. Now, \eqref{eq:empsignificancelevel} and \eqref{eq:empbalancing} yield
\[\Psi_M\left(\widehat{\textbf{q}}_M\right)\leq \frac{1}{M}\left(1+\frac{\sn-1}{\min\{\beta_1,\ldots,\beta_\sn\}}\right)=\oo(1)\] almost surely. In addition, Lemma \ref{lemma:uniqunesscritval} shows that the vector of critical values $\textbf{q}$ is unique. Moreover, $\sup_{\theta\in [0,\infty)^\sn}{\Vert F_M(\theta)-F(\theta)\Vert}$ and $\sup_{\theta_k\geq 0}{\Vert F_{M,k}(\theta_k)-F_k(\theta_k)\Vert}$ for all $k\in\{1,\ldots,\sn\}$ converge to zero almost surely. Thus, all assumptions of \citep[Theorem 5.9]{vanderVaart07} are satisfied and the assertion follows.
\end{proof}
\begin{proof}[Proof of Lemma \ref{lemma:computationtime}]
The computation time for the bounds $\underline{b}_{i,j}$ and $\overline{b}_{i,j}$ is $\OO(1)$ for every fixed interval $[i/n,j/n]\in\mathcal{D}$, since they depend only on the sums $\sum_{l=i}^j{Y_l}$ as well as $\sum_{l=i}^j{Y_l^2}$ and these can be obtained from (precomputed) cumulative sums. The computation time for the intersected bounds $\underline{B}_{i,j}$ and $\overline{B}_{i,j}$ are also $\OO(1)$ for a fixed interval $[i/n,j/n]$, since they can be computed iteratively. Therefore, the total time to compute the bounds is $\OO(n)$, since the dyadic partition contains less than $n$ intervals.\\
It follows from its iterative definition that the left limits $L_1,\ldots,L_{\hat{K}}$ can be computed in $\OO(n)$. Therefore, the dynamic programming algorithm has cost $\OO(\sum_{k=1}^{\hat{K}-1}(R_k-L_k+1)$ $(R_{k+1}-L_{k+1}+1))$ besides some linear costs, since for each point in the interval $[L_{k+1},R_{k+1}]$ the optimal change-point in the interval $[L_{k},R_{k}]$ has to be determined by computing the cost functional for each of these points. But, for a single interval the computation time for the restricted maximum likelihood estimator and for the cost functional is $\OO(1)$ if the constraints $\underline{B}_{i,j}$ and $\overline{B}_{i,j}$ are given, since the restricted maximum likelihood estimator and the cost functional depend besides these constraints again only on the sums $\sum_{l=i}^j{Y_l}$ and $\sum_{l=i}^j{Y_l^2}$. This proves the assertion.
\end{proof}

\onehalfspacing

\bibliographystyle{plainnat}
\bibliography{LiteratureHSMUCE}

\end{document}